\documentclass[reprint, superscriptaddress, nofootinbib, amsmath,amssymb,aps,pra]{revtex4-2}

\usepackage[whole]{bxcjkjatype}
\usepackage{lineno}
\usepackage{graphicx}
\usepackage{dcolumn}
\usepackage{bm}
\usepackage[colorlinks,linkcolor=blue,urlcolor=blue,citecolor=blue,hypertexnames=false]{hyperref}
\usepackage{dsfont}
\usepackage{physics}
\usepackage{nicefrac}
\usepackage{enumitem}
\usepackage{tcolorbox}
\usepackage{mathtools}
\usepackage{tabularx}
\newcolumntype{C}{>{\centering\arraybackslash}X}
\usepackage{enumitem}

\usepackage{amsmath, amssymb, amsfonts, amsthm}
\usepackage{thmtools, thm-restate}

\DeclareMathOperator{\id}{id}

\DeclareMathOperator{\E}{E}

\newtheorem{theorem}{Theorem}
\newtheorem*{theorem*}{Theorem}

\newtheorem{corollary}[theorem]{Corollary}
\newtheorem{lemma}[theorem]{Lemma}

\newtheorem{proposition}[theorem]{Proposition}
\newtheorem{example}[theorem]{Example}

\newcommand{\red}{\color{red}}

\allowdisplaybreaks

\def \be {\begin{equation}}
\def \ee {\end{equation}}

\def \argmax{\mathop{\rm argmax}}

\def \sofc2{{\cal S}({\mathbb C}^2)}

\def\>{\rangle}
\def\<{\langle}

\def\Label{\label}

\makeatletter 
\renewcommand\onecolumngrid{
\do@columngrid{one}{\@ne}
\def\set@footnotewidth{\onecolumngrid}
\def\footnoterule{\kern-6pt\hrule width 1.5in\kern6pt}
}
\renewcommand\twocolumngrid{
        \def\footnoterule{
        \dimen@\skip\footins\divide\dimen@\thr@@
        \kern-\dimen@\hrule width.5in\kern\dimen@}
        \do@columngrid{mlt}{\tw@}
}
\makeatother

\begin{document}

\title{Generalized Quantum Stein's Lemma and Second Law of Quantum Resource Theories}

\author{Masahito Hayashi}
\email{hmasahito@cuhk.edu.cn}
\affiliation{School of Data Science, The Chinese University of Hong Kong, Shenzhen, Longgang District, Shenzhen, 518172, China}
\affiliation{International Quantum Academy, Futian District, Shenzhen 518048, China}
\affiliation{Graduate School of Mathematics, Nagoya University, Chikusa-ku, Nagoya 464--8602, Japan}
\author{Hayata Yamasaki}
\email{hayata.yamasaki@gmail.com}
\affiliation{Department of Physics, Graduate School of Science, The University of Tokyo, 7--3--1 Hongo, Bunkyo-ku, Tokyo, 113--0033, Japan}
\affiliation{
Department of Computer Science, Graduate School of Information Science and Technology, The University of Tokyo, 7--3--1 Hongo, Bunkyo-ku, Tokyo, 113--8656, Japan
}

\begin{abstract}
The second law of thermodynamics is the cornerstone of physics, characterizing the convertibility between thermodynamic states through a single function, entropy. Given the universal applicability of thermodynamics, a fundamental question in quantum information theory is whether an analogous second law can be formulated to characterize the convertibility of resources for quantum information processing by a single function. In 2008, a promising formulation was proposed, linking resource convertibility to the optimal performance of a variant of the quantum version of hypothesis testing. Central to this formulation was the generalized quantum Stein's lemma, which aimed to characterize this optimal performance by a measure of quantum resources, the regularized relative entropy of resource. If proven valid, the generalized quantum Stein's lemma would lead to the second law for quantum resources, with the regularized relative entropy of resource taking the role of entropy in thermodynamics. However, in 2023, a logical gap was found in the original proof of this lemma, casting doubt on the possibility of such a formulation of the second law. In this work, we address this problem by developing alternative techniques to successfully prove the generalized quantum Stein's lemma under a smaller set of assumptions than the original analysis. Based on our proof, we reestablish and extend the second law of quantum resource theories, applicable to both static resources of quantum states and a fundamental class of dynamical resources represented by classical-quantum (CQ) channels. These results resolve the fundamental problem of bridging the analogy between thermodynamics and quantum information theory.
\end{abstract}

\maketitle

\paragraph*{Introduction.}
Quantum information processing marks a groundbreaking shift in information technology, offering capabilities beyond those of classical information processing, such as significant speedups in quantum computation and strong security in quantum cryptography~\cite{N4}.
These breakthroughs are driven by the efficient use of intrinsic quantum properties such as entanglement and coherence, which serve as resources amplifying the power of quantum information processing.
To systematically explore and harness these quantum properties, quantum resource theories (QRTs)~\cite{Kuroiwa2020,Chitambar2018} have been developed, offering an operational framework for studying the manipulation and quantification of quantum resources.
QRTs are defined by specifying a restricted class of operations, i.e., free operations, such as local operations and classical communication (LOCC)~\cite{Horodecki2009,chitambar2014everything,Yamasaki2024} for manipulating entanglement across distant laboratories.
Quantum states freely obtainable under these operations are called free states, while non-free states to overcome the operational restrictions are viewed as resources.
By analyzing tasks that involve these operations and quantum resources, QRTs uncover the potential advantages and the fundamental limitations of quantum information processing, all governed by the law of quantum mechanics.

The construction of such operational frameworks has been a successful approach in the field of physics.
Thermodynamics, one of the most traditional operational theories in physics, has revealed potential uses and fundamental limits of energy resources based on a few axioms~\cite{carnot1890reflections,doi:10.1080/14786445608642141,Thomson_1853,LIEB19991,Lieb2004,10.1063/1.883034}.
Although originally motivated by the development of steam engines, thermodynamics applies to various domains, including chemical reactions~\cite{lewis1923thermodynamics,guggenheim1933modern} and energy consumption in computation~\cite{5392446,meier2023energyconsumption}, thanks to its axiomatic formulation ensuring universal applicability.
Central to thermodynamics is the concept of entropy, a real-valued function indicating convertibility between thermodynamic states.
Specifically, as formulated in Refs.~\cite{LIEB19991,Lieb2004,10.1063/1.883034}, the second law of thermodynamics states that the conversion from a thermodynamic state $X_1$ to another state $X_2$ under adiabatic operations is possible if and only if their entropies satisfy $S\qty(X_1)\leq S\qty(X_2)$.

\begin{figure}[t!]
    \centering
    \includegraphics[width=3.4in]{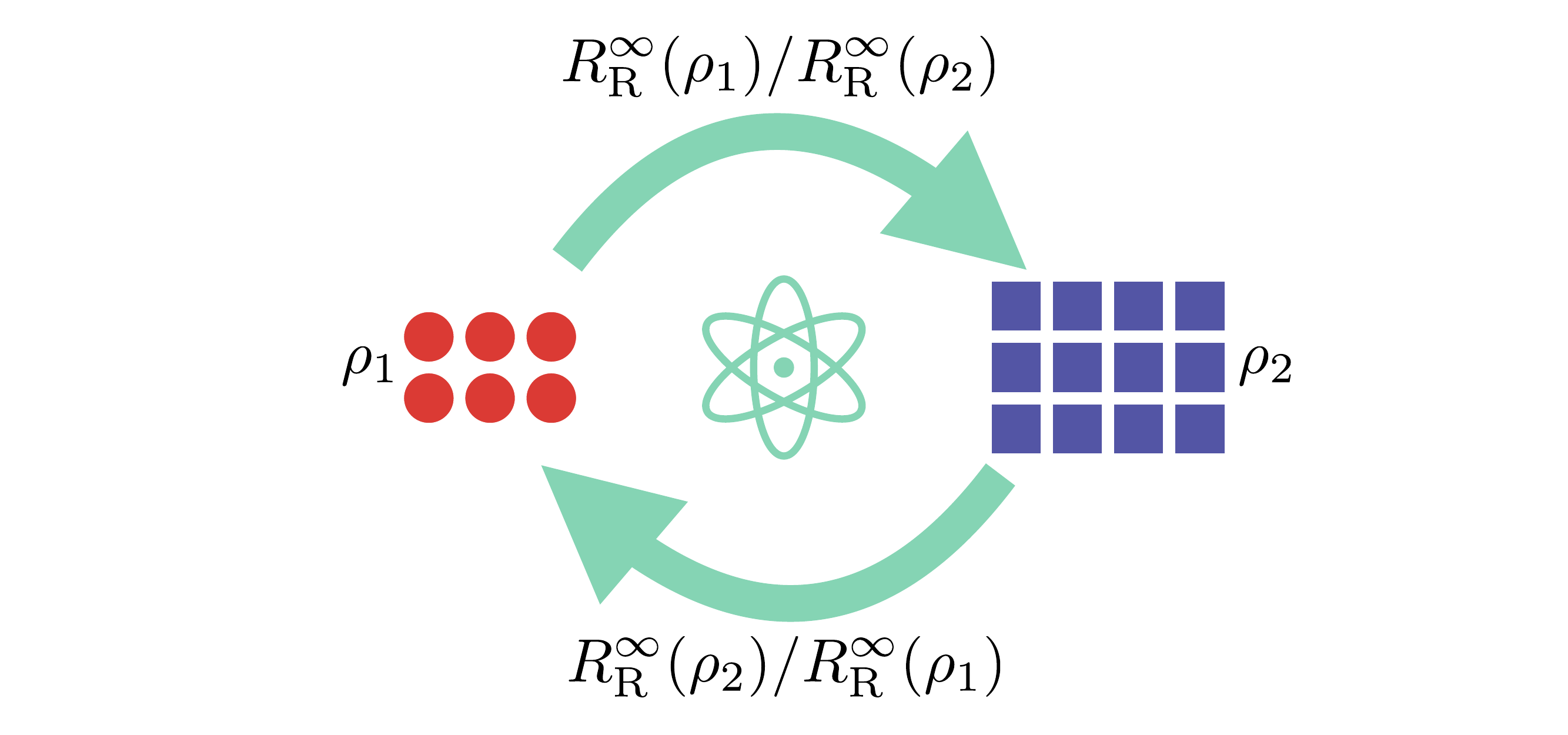}
    \caption{
    The second law of quantum resource theories (QRTs). Our main result, i.e., the proof of the generalized quantum Stein's lemma, leads to an axiomatic formulation of QRTs with a second law analogous to that of thermodynamics. In this formulation, a single function, the regularized relative entropy of resource $R_\mathrm{R}^\infty$, provides the necessary and sufficient condition for the asymptotic convertibility between quantum resources (quantum states $\rho_1$ represented by red circles and $\rho_2$ by blue squares) at the optimal rate, in analogy to entropy in the second law of thermodynamics.
    }
    \label{fig:second_law}
\end{figure}

While thermodynamics provides a quantitative understanding of energy and entropy through the state conversion tasks, QRTs also aim to quantify the potential and limit of quantum resources.
A fundamental task in QRTs is the asymptotic conversion between quantum states~\cite{Kuroiwa2020,Chitambar2018},
which involves converting many independently and identically distributed (IID) copies of quantum state $\rho_1$ into as many copies of state $\rho_2$ as possible with a vanishing error under the restricted operations, in the spirit of Shannon's information theory~\cite{6773024,cover2012elements}.
The maximum number of copies of $\rho_2$ obtained per $\rho_1$ is called the asymptotic conversion rate.
In view of the success of thermodynamics, it is fundamental to seek a universal formulation of QRTs with a second law analogous to that of thermodynamics (Fig.~\ref{fig:second_law}).
Such a second law would provide a necessary and sufficient condition for the asymptotic conversion at a certain rate characterized by a single function of $\rho_j$ ($j\in\{1,2\}$), similar to entropy in thermodynamics.

However, establishing the second law for QRTs has been challenging in general.
In entanglement theory, for instance, the irreversibility of asymptotic state conversion under LOCC hinders the establishment of the second law to characterize the asymptotic conversion rate~\cite{PhysRevLett.86.5803,PhysRevLett.119.180506,Lami2023}.
Nevertheless, significant progress toward the second law of entanglement theory has been made in Refs.~\cite{Brand_o_2008,brandao2010reversible,Brandao2010}, and this has been extended to general QRTs in Ref.~\cite{Brandao2015}.
Conventionally, LOCC is defined using a bottom-up approach, specifying what can be done in the laboratories.
In contrast, Refs.~\cite{Brand_o_2008,brandao2010reversible,Brandao2010} took a top-down approach, defining an axiomatic class of operations by specifying what cannot be done, similar to how thermodynamics axiomatically introduces adiabatic operations~\cite{LIEB19991,Lieb2004,10.1063/1.883034}.

Under an axiomatic class of operations, Refs.~\cite{Brand_o_2008,brandao2010reversible,Brandao2010} established a remarkable connection between the asymptotic conversion of quantum states and a variant of quantum hypothesis testing~\cite{hiai1991proper,887855}, another fundamental task in quantum information theory.
This task aims to distinguish $n$ IID copies of quantum state $\rho^{\otimes n}$ from any state $\sigma_n$ in the set of free states, which may not be in IID form.
The non-IIDness of $\sigma_n$ poses a substantial challenge to its analysis.
References~\cite{Brand_o_2008,brandao2010reversible,Brandao2010} tried to address this challenge to establish a lemma characterizing the optimal performance of this hypothesis-testing task, which is called the generalized quantum Stein's lemma.
If valid, the generalized quantum Stein's lemma would enable a formulation of QRTs with the second law, as intended in Refs.~\cite{Brand_o_2008,brandao2010reversible,Brandao2010,Brandao2015}; however, as described in Refs.~\cite{fang2022ultimate,berta2023gap,Berta2023,yamasaki2024generalized}, a logical gap has been found in the original analysis~\cite{Brandao2010} of the lemma, invalidating its consequences.\footnote{A history of this finding is as follows~\cite{KunFungPersonal}. First, a logical gap was pointed out in a part of the arguments in the early version of Ref.~\cite{fang2022ultimate}, but this part was inspired by the analysis of the generalized quantum Stein's lemma in Ref.~\cite{Brandao2010}. Then, based on this, Ref.~\cite{berta2023gap} pointed out the logical gap of Ref.~\cite{Brandao2010}. References~\cite{Berta2023,yamasaki2024generalized} also summarize this problem.}
Reference~\cite{yamasaki2024generalized} attempted an alternative approach to resolve this issue, but their approach also turned out to be insufficient for completing the proof.
Consequently, the first---and only known---general formulation of QRTs with the second law has lost its validity, reopening the question of whether such a framework can be established at all.

In this work, we prove the generalized quantum Stein's lemma and extend the framework for QRTs with the second law.
Our proof circumvents the logical gaps in the previous analyses of the generalized quantum Stein's lemma in Refs.~\cite{Brandao2010,yamasaki2024generalized} by introducing alternative techniques based on pinching~\cite{hayashi2002optimal} and the information spectrum method~\cite{4069150} to handle the non-IIDness.
These techniques enable us to prove a stronger, more widely applicable version of the generalized quantum Stein's lemma, with fewer assumptions than the previous analyses~\cite{Brandao2010,yamasaki2024generalized}. 
Progressing beyond QRTs for static resources (i.e., states) in Refs.~\cite{Brand_o_2008,brandao2010reversible,Brandao2010,Brandao2015}, we construct a framework for QRTs with the second law for a fundamental class of dynamical resources, i.e., classical-quantum (CQ) channels, relevant in communication scenarios~\cite{hayashi2016quantum} while also encompassing the existing results for states in Refs.~\cite{Brand_o_2008,brandao2010reversible,Brandao2010,Brandao2015} as special cases.
Below, we introduce QRTs for states, present our main results on the generalized quantum Stein's lemma, and discuss its implications for formulating QRTs with the second law for states.
In Methods, we detail our proof of the generalized quantum Stein's lemma and formulate QRTs with the second law for CQ channels.

\paragraph*{Framework of QRTs.}
We present the framework of QRTs for states based on the general formulation in Ref.~\cite{Kuroiwa2020}.
We represent a quantum system by a $d$-dimensional complex Hilbert space $\mathcal{H}=\mathbb{C}^d$ for some finite $d$.
Quantum states of $\mathcal{H}$ are positive semidefinite operators on $\mathcal{H}$ with unit trace, with the set of states of $\mathcal{H}$ denoted by $\mathcal{D}\qty(\mathcal{H})$.
For any $n\in\{1,2,\ldots\}$, a composite system of $n$ subsystems $\mathcal{H}$ is represented as $\mathcal{H}^{\otimes n}$ with tensor product.
The set of quantum operations from an input system $\mathcal{H}_\mathrm{in}$ to an output system $\mathcal{H}_\mathrm{out}$ is represented by that of completely positive and trace-preserving (CPTP) linear maps (or channels), denoted by $\mathcal{C}\qty(\mathcal{H}_\mathrm{in}\to\mathcal{H}_\mathrm{out})$.

Specifying a set of free operations defines QRTs, representing quantum operations used in quantum information processing.
The set of free operations is denoted by $\mathcal{O}\qty(\mathcal{H}_\mathrm{in}\to\mathcal{H}_\mathrm{out})\subset\mathcal{C}\qty(\mathcal{H}_\mathrm{in}\to\mathcal{H}_\mathrm{out})$, which we may write $\mathcal{O}$ if the argument is obvious from the context.
Given $\mathcal{O}$, free states are defined as states $\rho_\mathrm{free}$ such that for any initial state $\rho$, there exist some free operations $\mathcal{E}\in\mathcal{O}$ to convert it into $\rho_\mathrm{free}=\mathcal{E}\qty(\rho)$ (in other words, regardless of how non-resourceful $\rho$ may be, $\rho_\mathrm{free}$ can be freely obtained, even from scratch)~\cite{Kuroiwa2020}.
The set of free states of $\mathcal{H}$ is denoted by $\mathcal{F}\qty(\mathcal{H})\subset\mathcal{D}\qty(\mathcal{H})$.
Following the previous works~\cite{Brand_o_2008,brandao2010reversible,Brandao2010,Brandao2015,yamasaki2024generalized}, we consider QRTs with their sets of free states satisfying the following properties.
\begin{enumerate}[label={State\arabic*}]
    \item \label{p1:main}The set $\mathcal{F}(\mathcal{H})$ of free states is closed and convex.
    \item \label{p4:main}The set of free states is closed under tensor product; that is, if $\rho_\mathrm{free}\in\mathcal{F}\qty(\mathcal{H})$ and $\rho_\mathrm{free}^\prime\in\mathcal{F}\qty(\mathcal{H}^\prime)$, then $\rho_\mathrm{free}\otimes\rho_\mathrm{free}^\prime\in\mathcal{F}\qty(\mathcal{H}\otimes\mathcal{H}^\prime)$.
    \item \label{p2:main}The set $\mathcal{F}(\mathcal{H})$ of free states contains a full-rank state $\rho_{\mathrm{full}}>0$.
\end{enumerate}
As discussed in Refs.~\cite{Kuroiwa2020,Chitambar2018}, it is conventional to additionally assume a certain set of axioms for free operations to ensure that QRTs are physically well-motivated, such as the requirement that the composition of multiple free operations remains a free operation.
From this requirement, an essential property of QRTs follows: free operations always map free states to free states and, therefore, never generate resources from any free state~\cite{Chitambar2018}.
However, in our analysis below, rather than focusing solely on free operations under such axioms, we also consider a relaxed notion of free operations, using the resource-non-generating property of free operations as a guiding principle.

Various QRTs with fundamental motivations, such as those of entanglement, athermality, coherence, asymmetry, and magic states, indeed satisfy the above properties~\cite{Chitambar2018}.
In the original study of the generalized quantum Stein's lemma, Refs.~\cite{Brand_o_2008,brandao2010reversible} considered the entanglement theory, which has these properties.
A later extension to general QRTs in Ref.~\cite{Brandao2015} does not explicitly mention Property~\ref{p2:main}, but Property~\ref{p2:main} is necessary for proving and using the generalized quantum Stein's lemma to reproduce the results of Ref.~\cite{Brandao2015}.
In addition to the above properties, the existing attempts to prove the generalized quantum Stein's lemma in Refs.~\cite{Brandao2010,yamasaki2024generalized} imposed two additional assumptions; one is that $\mathcal{F}$ should be closed under taking the partial trace, and the other is that $\mathcal{F}$ should be invariant under permutation of subsystems.
By contrast, our proof does not impose these additional assumptions, leading to broader applicability.
Furthermore, whereas these properties are attributed to quantum states, we also introduce their generalization to CQ channels (see Methods for details).

A non-free state, assisting free operations to conduct quantum information processing, is called a resource state, and QRTs provide a way to quantify resourcefulness via resource measures.
Resource measures are a family of real functions $R_\mathcal{H}$ of the state of every system $\mathcal{H}$ that is monotonically non-increasing under free operations~\cite{Kuroiwa2020}; i.e., for any free operation $\mathcal{E}$ and any state $\rho$, it should hold that $R(\rho)\geq R(\mathcal{E}(\rho))$, where we may omit the subscript $\mathcal{H}$ of $R_\mathcal{H}$ to write $R$ if it is obvious from the context.
For example, the relative entropy of resource is defined as
$R_\mathrm{R}\qty(\rho)\coloneqq\min_{\rho_\mathrm{free}\in\mathcal{F}\qty(\mathcal{H})}D\left(\rho\middle|\middle|\rho_\mathrm{free}\right)$~\cite{Chitambar2018,Kuroiwa2020},
where $D\left(\rho\middle|\middle|\sigma\right)\coloneqq\Tr[\rho(\log\rho-\log\sigma)]$ is the quantum relative entropy, and $\log$ is the natural logarithm throughout this work.
Its variant, the regularized relative entropy of resources, is defined as
$R_\mathrm{R}^\infty\qty(\rho)\coloneqq\lim_{n\to\infty}\frac{1}{n}R_\mathrm{R}\qty(\rho^{\otimes n})
$~\cite{Chitambar2018,Kuroiwa2020}.
Another example is the generalized robustness of resource (also known as global robustness) defined as
$R_\mathrm{G}\qty(\rho)\coloneqq\min\qty{s\geq 0:\frac{\rho+s\rho^\prime}{1+s}\in\mathcal{F}(\mathcal{H}),\rho^\prime\in\mathcal{D}\qty(\mathcal{H})}$~\cite{Chitambar2018}.
All these $R_\mathrm{R}$, $R_\mathrm{R}^\infty$, and $R_\mathrm{G}$ serve as resource measures, satisfying the monotonicity as required~\cite{Chitambar2018}.

A fundamental task in QRTs is the asymptotic conversion of quantum states.
This task involves converting many copies of one state, $\rho_1$, into as many copies as possible of another state, $\rho_2$, using free operations, within errors that vanish asymptotically.
The conversion rate $r_\mathcal{O}\qty(\rho_1\to\rho_2)$, under a class $\mathcal{O}$ of operations, is the supremum of achievable rates $r$ in this asymptotic conversion.
Specifically, $r$ is achievable if $\liminf_{n\to\infty}\frac{1}{2}\left\|\mathcal{E}_n\qty(\rho_1^{\otimes n})-\rho_2^{\otimes \lceil rn\rceil}\right\|_1=0$ for some sequence $\{\mathcal{E}_n\}_n$ of operations in $\mathcal{O}$, where $\lceil{}\cdots{}\rceil$ is the ceiling function, and $\|\cdots\|_1$ is the trace norm~\cite{Kuroiwa2020}.
For example, in the entanglement theory, $\mathcal{O}$ can be chosen as LOCC, making $\mathcal{F}$ the set of separable states~\cite{Kuroiwa2020,Chitambar2018}.
However, under such $\mathcal{O}$, $r_\mathcal{O}\qty(\rho_1\to\rho_2)$ cannot be characterized by a single resource measure $R$ due to the irreversibility of conversion between mixed entangled states~\cite{PhysRevLett.86.5803,PhysRevLett.119.180506,Lami2023}.
To address this issue, as in the previous works~\cite{Brand_o_2008,brandao2010reversible,Brandao2010,Brandao2015}, one can consider a slightly broader, axiomatically defined class $\tilde{\mathcal{O}}$ of operations as a relaxation of $\mathcal{O}$.
A fundamental question in QRTs is whether it is possible to establish a general framework of QRTs with an appropriate choice of the class $\tilde{\mathcal{O}}$ of operations and a single resource measure $R$ such that the resource measures $R(\rho_j)$ ($j\in\{1,2\}$) characterize the convertibility at rate $r_{\tilde{\mathcal{O}}}\qty(\rho_1\to\rho_2)$, which would constitute the second law of QRTs\@.

\paragraph*{Main result: Generalized quantum Stein's lemma.}---
For establishing such a desired framework of QRTs, the generalized quantum Stein's lemma plays a crucial role, characterizing the optimal performance of a variant of quantum hypothesis testing; our main result is the proof of the generalized quantum Stein's lemma.
In this variant of quantum hypothesis testing, as introduced in Ref.~\cite{Brandao2010}, we are initially given $n\in\{1,2,\ldots\}$, a classical description of a state $\rho$ of $\mathcal{H}$, and an unknown quantum state of $\mathcal{H}^{\otimes n}$.
The task is to perform a two-outcome measurement by a positive operator-valued measure (POVM) $\{T_n,\mathds{1}-T_n\}$ on $\mathcal{H}^{\otimes n}$ (where $\mathds{1}$ is the identity operator, and $0\leq T_n\leq\mathds{1}$) to distinguish the following two cases.
\begin{itemize}
    \item Null hypothesis: The given unknown state is $n$ IID copies $\rho^{\otimes n}$ of $\rho$.
    \item Alternative hypothesis: The given state is some free state $\sigma\in\mathcal{F}\qty(\mathcal{H}^{\otimes n})$ in the set satisfying Properties~\ref{p1:main}--\ref{p2:main}, where $\sigma$ may have a non-IID form over $\mathcal{H}^{\otimes n}$.
\end{itemize}
Hypothesis testing of this type with an alternative hypothesis composed of several elements has been studied in the context of the independence test \cite{IEEE-IT-8231191},
as explained in \cite[Sec.~VIII]{IEEE-IT-8231191}, 
which takes a central role in the meta-converse of the channel coding~\cite[Sec.~3]{Nagaoka},~\cite{IEEE-IT-5452208},~\cite[Sec.~II]{IEEE-IT-6395254},
and in a converse bound for secure random number generation \cite{TW2014,IEEE-IT-7161366}.
Also, Ref.~\cite[Sec.~IV~B]{IEEE-IT-7887663} considers another type of alternative hypothesis composed of several elements in the context of hypothesis testing by using local measurements.

If the measurement outcome is $T_n$, we conclude that the given state was $\rho^{\otimes n}$, and if $\mathds{1}-T_n$, then was some free state $\sigma\in\mathcal{F}\qty(\mathcal{H}^{\otimes n})$.
For this hypothesis testing, we define the following two types of errors.
\begin{itemize}
    \item Type I error: The mistaken conclusion that the given state was some free state $\sigma\in\mathcal{F}\qty(\mathcal{H}^{\otimes n})$ when it was $\rho^{\otimes n}$, which happens with probability $\alpha_n\coloneqq\Tr\qty[\qty(\mathds{1}-T_n)\rho^{\otimes n})]$.
    \item Type II error: The mistaken conclusion that the given state was $\rho^{\otimes n}$ when it was some free state $\sigma\in\mathcal{F}\qty(\mathcal{H}^{\otimes n})$, which happens with probability $\beta_n\coloneqq\max_{\sigma\in\mathcal{F}(\mathcal{H}^{\otimes n})}\Tr\qty[T_n\sigma]$ in the worst case.
\end{itemize}

For any target type I error $\epsilon\in(0,1)$, by choosing appropriate POVMs, we can suppress the type II error exponentially in $n$ while keeping the type I error below $\epsilon$; as shown below, the generalized quantum Stein's lemma characterizes the optimal exponent, i.e., the fastest rate in suppressing the type II error, by the regularized relative entropy of resource.
Note that by choosing 
$\mathcal{F}\qty(\mathcal{H}^{\otimes n})=
\qty{\sigma^{\otimes n}}$ for a fixed full-rank state $\sigma$, 
the generalized quantum Stein's lemma reduces to quantum Stein's lemma for quantum hypothesis testing in the conventional IID setting~\cite{hiai1991proper,887855,Brandao2010}.\footnote{
This choice satisfies Property~\ref{p1:main} by construction. Property~\ref{p4:main} holds because
$\mathcal{F}\qty(\mathcal{H}^{\otimes n}\otimes \mathcal{H}^{\otimes n'})=\mathcal{F}\qty(\mathcal{H}^{\otimes \qty(n+n')})
=\qty{\sigma^{\otimes n} \otimes \sigma^{\otimes n'}}$. Additionally, Property~\ref{p2:main} is satisfied for any full-rank state $\sigma$.
}
See Methods for the details of our proof.

\begin{theorem}[\label{thm:main_generalized_steins_lemma}Generalized quantum Stein's lemma]
Given any family $\mathcal{F}$ of sets of states satisfying Properties~\ref{p1:main}--\ref{p2:main}, any state $\rho$ of $\mathcal{H}$, and any target type I error $\epsilon\in(0,1)$,
the optimal exponent of type II error is
\begin{align}
\label{eq:stein_main}
    \lim_{n\to\infty}-\frac{\log\qty(\min_{T_n}\max_{\sigma\in\mathcal{F}\qty(\mathcal{H}^{\otimes n})}\Tr\qty[T_n\sigma])}{n}= R_\mathrm{R}^\infty(\rho),
\end{align}
where the minimum is taken over all POVMs $\qty{T_n,\mathds{1}-T_n}$ suppressing the type I error below $\epsilon$, i.e., $\Tr[\qty(\mathds{1}-T_n)\rho^{\otimes n}]\leq\epsilon$, and $R_\mathrm{R}^\infty$ is the regularized relative entropy of resource.
\end{theorem}

The challenge in the proof of Theorem~\ref{thm:main_generalized_steins_lemma} arises from the non-IIDness of $\sigma\in\mathcal{F}\qty(\mathcal{H}^{\otimes n})$.
To address the non-IIDness, Ref.~\cite{Brandao2010} considered using symmetry under the permutation of the $n$ subsystems.
Such symmetry implies that almost all states of the $n$ subsystems are virtually identical and independent of each other~\cite{renner2006security,renner2007symmetry}, approximately recovering the IID structure of the state.
Then, Ref.~\cite{Brandao2010} tried to show that this approximation for recovering the IID structure would not change the quantum relative entropy up to a negligibly small amount, so the optimal exponent would coincide with its regularization $R_\mathrm{R}^\infty(\rho)$ if Lemma~III.9 of Ref.~\cite{Brandao2010} were true, where the logical gap of their analysis was found~\cite{berta2023gap}.
Reference~\cite{yamasaki2024generalized} also tried to address this non-IIDness using a continuity bound in the second argument of the quantum relative entropy in Refs.~\cite{10206734,10129917}, but it turned out that these continuity bounds are not tight enough for completing the proof.
In contrast, our proof overcomes the non-IIDness by developing alternative techniques that avoid relying on the assumption of the permutation invariance.
For the strong converse part (i.e., $\leq$ in~\eqref{eq:stein_main}), we show that the proof method for the strong converse of quantum Stein's lemma in the IID setting~\cite{887855} can be suitably adapted, offering a simpler proof than the existing one in Ref.~\cite{Brandao2010}.
To address the non-IIDness in the direct part (i.e., $\geq$ in~\eqref{eq:stein_main}), we employ pinching techniques to ensure commutativity of operators in the core part of our analysis, similar to the analysis of quantum Stein's lemma in the IID setting~\cite{hayashi2002optimal}; then, applying the information spectrum method~\cite{4069150} to these commuting operators, we construct the optimal sequence of free states for $R_\mathrm{R}^\infty$ to achieve~\eqref{eq:stein_main}.
See Methods for details.

\paragraph*{The second law of QRTs.}
Our proof of the generalized quantum Stein's lemma leads to the second law of QRTs, as originally intended in Refs.~\cite{Brand_o_2008,brandao2010reversible,Brandao2010,Brandao2015}.
So far, we have introduced the class $\mathcal{O}$ of free operations in a conventional way, i.e., in a bottom-up approach by specifying what can be done.
By contrast, Refs.~\cite{Brand_o_2008,brandao2010reversible,Brandao2010,Brandao2015} formulated QRTs by introducing a slightly broader class $\tilde{O}$ of operations defined in an axiomatic approach by specifying what cannot be done, as in the adiabatic operations in thermodynamics.

A fundamental requirement for free operations $\mathcal{O}$ is that the free operations should not generate resource states from free states; however, in the context of asymptotic conversion, it is possible to axiomatically define a relaxed class of operations, $\tilde{O}$, which captures this requirement only in an asymptotic sense.
Various axiomatic definitions of $\tilde{O}$ can be considered, but even with such $\tilde{O}$, formulating general QRTs with the second law remains highly challenging~\cite{Lami2023,PhysRevLett.119.180506}.
The second law does hold for special types of quantum resources, such as QRTs of coherence~\cite{PhysRevLett.120.070403,PhysRevA.97.050301,berta2023gap} and athermality~\cite{PhysRevA.67.062104,PhysRevLett.111.250404}; similarly, formulating QRTs with the second law may be possible using some variants of composite quantum Stein's lemmas~\cite{9031743,PhysRevLett.121.190503,hayashi2002optimal,gao2024generalizedsteinslemmaasymptotic,hayashi2016correlation,8231191}, but these formulations are not general enough to cover entanglement~\cite{berta2023gap}.
Theories with the second law can also be developed by considering ``operations'' that extend beyond the limitations of quantum mechanics, such as allowing post-selection~\cite{regula2023reversibility}, non-physical quasi-operations~\cite{wang2023reversible}, and the use of batteries~\cite{ganardi2024secondlawentanglementmanipulation}.
However, to maintain full generality within the law of quantum mechanics, the only known approach to introducing $\tilde{O}$ that leads to QRTs with the second law is based on the generalized quantum Stein's lemma, as originally attempted Refs.~\cite{Brand_o_2008,brandao2010reversible,Brandao2010,Brandao2015}.
Following Refs.~\cite{Brand_o_2008,brandao2010reversible,Brandao2010,Brandao2015}, we define a set $\tilde{O}$ of asymptotically resource-non-generating operations as the set $\tilde{\mathcal{O}}\qty(\mathcal{H}_\mathrm{in}\to\mathcal{H}_\mathrm{out})\coloneqq\qty{\qty{\mathcal{E}_n\in\mathcal{C}\qty(\mathcal{H}_\mathrm{in}^{\otimes n}\to\mathcal{H}_\mathrm{out}^{\otimes n})}_{n=1,2,\ldots}}$ of all sequences of operations (CPTP linear maps) satisfying the following property: any sequence $\qty{\mathcal{E}_n}_n\in\tilde{O}$ of operations in this set asymptotically generates no resource from any free states in terms of the generalized robustness of resource $R_\mathrm{G}$, i.e., for any sequence $\qty{\rho_\mathrm{free}^{(n)}\in\mathcal{F}\qty(\mathcal{H}_\mathrm{in}^{\otimes n})}_n$ of free states
\begin{align}
\label{eq:R_G_condition_main_text}
    \lim_{n\to\infty}R_\mathrm{G}\qty(\mathcal{E}_n\qty(\rho_\mathrm{free}^{(n)}))=0.
\end{align}

Under $\tilde{\mathcal{O}}$ satisfying~\eqref{eq:R_G_condition_main_text}, using the generalized quantum Stein's lemma in Theorem~\ref{thm:main_generalized_steins_lemma}, we show a characterization of the asymptotic convertibility of resource states, as stated by the theorem below.
Also, as presented in Methods, we extend this theorem from the static resource, i.e., states, to a fundamental class of the dynamical resources, i.e., CQ channels, which are well motivated and widely studied in communication scenarios~\cite{hayashi2016quantum}
(see Methods for details).

\begin{theorem}[Second law of QRTs for states]
\label{thm:second_law_main}
Given any family $\mathcal{F}$ of sets of free states satisfying Properties~\ref{p1:main}--\ref{p2:main},
for any states $\rho_1$ and $\rho_2$ satisfying $R_\mathrm{R}^\infty(\rho_j)>0$ ($j\in\{1,2\}$), 
the asymptotic conversion rate between $\rho_1$ and $\rho_2$ under the asymptotically resource-non-generating operations $\tilde{\mathcal{O}}$ satisfying~\eqref{eq:R_G_condition_main_text} is
\begin{equation}
\label{eq:second_law}
    r_{\tilde{\mathcal{O}}}\qty(\rho_1\to\rho_2)=\frac{R_\mathrm{R}^\infty\qty(\rho_1)}{R_\mathrm{R}^\infty\qty(\rho_2)}.
\end{equation}   
\end{theorem}

To prove Theorem~\ref{thm:second_law_main}, assuming the generalized quantum Stein's lemma, Ref.~\cite{Brandao2015} originally showed an inequality $r_{\tilde{\mathcal{O}}}\qty(\rho_1\to\rho_2)\geq\frac{R_\mathrm{R}^\infty\qty(\rho_1)}{R_\mathrm{R}^\infty\qty(\rho_2)}$ in QRTs for states. The proof of the opposite inequality $r_{\tilde{\mathcal{O}}}\qty(\rho_1\to\rho_2)\leq\frac{R_\mathrm{R}^\infty\qty(\rho_1)}{R_\mathrm{R}^\infty\qty(\rho_2)}$ was not discussed in Ref.~\cite{Brandao2015} but later provided in Ref.~\cite{regula2023reversibility} with an extension to probabilistic protocols.
Our contribution is to prove both directions of these inequalities in a more general framework of QRTs for CQ channels, which includes Theorem~\ref{thm:second_law_main} in QRTs for states as a special case.
This generalization requires identifying nontrivial conditions for $\tilde{O}$ in QRTs for CQ channels, as was done in QRTs for states in Refs.~\cite{Brand_o_2008,brandao2010reversible,Brandao2010,Brandao2015}; {\red for this purpose,} we extend the relation between $R_\mathrm{R}^\infty$ and $R_\mathrm{G}$, previously shown for states in Ref.~\cite{Brandao2010}, to CQ channels.
This enables us to offer a more widely applicable proof of the second law of QRTs in this general framework.
See Methods for details.

The equation~\eqref{eq:second_law} represents the second law of QRTs for quantum states.
This is analogous to the axiomatic formulation of the second law of thermodynamics in Refs.~\cite{LIEB19991, Lieb2004,10.1063/1.883034}, which provides the necessary and sufficient condition for convertibility between thermodynamic states $X_1$ and $X_2$ under adiabatic operations solely by comparing the (additive) entropy functions $S(X_1)$ and $S(X_2)$. 
In the framework of QRTs in Theorem~\ref{thm:second_law_main}, the regularized relative entropy of resource $R_\mathrm{R}^\infty$ takes the role of entropy in thermodynamics.
Similar to the idealization of quasi-static processes under adiabatic operations, which are axiomatically introduced in thermodynamics~\cite{LIEB19991,Lieb2004,10.1063/1.883034} and may not be exactly realizable in practice, it may be generally unknown how to realize all operations in the axiomatically defined class $\tilde{O}$.
However, this axiomatic class of operations is broad enough to include all physically realizable operations, ensuring its universal applicability regardless of future technological advances, much like thermodynamics.

Lastly, for instance, in the entanglement theory, $\mathcal{F}(\mathcal{H})$ is taken as the set of separable states on two spatially separated systems $A$ and $B$~\cite{Horodecki2009}, and a fundamental resource state is an ebit, i.e., $\Phi_2\coloneqq\ket{\Phi_2}\bra{\Phi_2}$ with $\ket{\Phi_2}\coloneqq\frac{1}{\sqrt{2}}\qty(\ket{0}^A\otimes\ket{0}^B+\ket{1}^A\otimes\ket{1}^B)$.
For an ebit, we have $R_\mathrm{R}^\infty\qty(\Phi_2)=1$~\cite{Horodecki2009}.
In the asymptotic conversion from any state $\rho$ to ebit $\Phi_2$, the maximum number of ebits $\Phi_2$ obtained per $\rho$ is called the distillable entanglement~\cite{PhysRevA.53.2046}, written under the class $\tilde{O}$ of operations as $E_{\mathrm{D},\tilde{\mathcal{O}}}(\rho)\coloneqq r_{\tilde{\mathcal{O}}}\qty(\rho\to\Phi_2)$.
Also in the asymptotic conversion rate from $\Phi_2$ to $\rho$, the minimum required number of ebits $\Phi_2$ per $\rho$ is called the entanglement cost~\cite{PhysRevA.53.2046,Hayden_2001,Yamasaki2024}, which is given under $\tilde{\mathcal{O}}$ by
$E_{\mathrm{C},\tilde{\mathcal{O}}}(\rho)\coloneqq 1/{r_{\tilde{\mathcal{O}}}\qty(\Phi_2\to\rho)}$~\cite{Kuroiwa2020}.
Due to~\eqref{eq:second_law}, our results establish the asymptotic reversibility between all, pure and mixed, bipartite entangled states, as originally intended in Refs.~\cite{Brand_o_2008,brandao2010reversible}, i.e.,
$E_{\mathrm{D},\tilde{\mathcal{O}}}(\rho)=E_{\mathrm{C},\tilde{\mathcal{O}}}(\rho)=R_\mathrm{R}^\infty(\rho)$,
which resolves the question of the possibility of formulating the reversible framework of entanglement theory raised in Ref.~\cite{krueger2005open}.
See also Methods for other examples.

\paragraph*{Outlook.}
In this work, we proved the generalized quantum Stein's lemma (Theorem~\ref{thm:main_generalized_steins_lemma}), resolving a fundamental open problem in quantum information theory by overcoming the logical gaps of the existing analyses in Refs.~\cite{Brandao2010,yamasaki2024generalized}. By developing alternative proof techniques, we remove some of the assumptions in Refs.~\cite{Brandao2010,yamasaki2024generalized}. Our proof enables the formulation of the framework of QRTs equipped with the second law under Properties~\ref{p1:main}–\ref{p2:main} (Theorem~\ref{thm:second_law_main}), as originally intended in Refs.~\cite{Brand_o_2008,brandao2010reversible,Brandao2010,Brandao2015}.
In this framework, a single resource measure, the regularized relative entropy of resource, characterizes the asymptotic convertibility between resource states at the optimal rate, analogous to entropy in the second law of thermodynamics.
Since the publication of the initial works~\cite{Brand_o_2008,brandao2010reversible,Brandao2010,Brandao2015}, the scope of QRTs has expanded far beyond entanglement~\cite{Chitambar2018,Kuroiwa2020}; therefore, we also extend the framework of QRTs with the second law from those for states to those for CQ channels, i.e., a fundamental class of channels in communication scenarios~\cite{hayashi2016quantum}, broadening their applicability (see Methods for details). A remaining open question is how universally our results further generalize beyond convex and finite-dimensional QRTs satisfying Properties~\ref{p1:main}–\ref{p2:main}, for example, to non-convex QRTs~\cite{Kuroiwa2020,PhysRevA.104.L020401,PRL,PRA} and infinite-dimensional QRTs~\cite{Kuroiwa2020,PhysRevA.104.L020401,Regula2021,Lami2021,ferrari2023asymptotic,Yamasaki2024}.
The generalized quantum Stein's lemma itself also serves as a valuable tool in quantum information theory, such as examining the faithfulness of the regularized relative entropy of entanglement~\cite{PhysRevLett.103.160504,Brandao2010} and the squashed entanglement~\cite{berta2023gap,brandao2011faithful}.
Exploring further applications of this lemma would also be an intriguing research direction.
Given the success of thermodynamics, these general results are expected to be fundamental in studying the vast possibilities of the use of quantum resources in the future.

\begin{acknowledgments}
  H.Y.\ acknowledges Kohdai Kuroiwa for discussion. 
  A part of this work was carried out during the BIRS-IMAG workshop ``Towards Infinite Dimension and Beyond in Quantum Information'' held at the Institute of Mathematics of the University of Granada (IMAG) in Spain.
M.H.\ was supported in part by the National Natural Science Foundation of China under Grant 62171212, and
the General R\&D Projects of 
1+1+1 CUHK-CUHK(SZ)-GDST Joint Collaboration Fund 
(Grant No. GRDP2025-022)\@.
  H.Y.\ was supported by JST PRESTO Grant Number JPMJPR201A, JPMJPR23FC, JSPS KAKENHI Grant Number JP23K19970, and MEXT Quantum Leap Flagship Program (MEXT QLEAP) JPMXS0118069605, JPMXS0120351339\@.
\end{acknowledgments}

\section*{Author contributions}

Both authors contributed to the conception of the work, the analysis and interpretation of the work, and the preparation of the manuscript.

\section*{Competing interests}

The authors declare no competing interests.

\section*{Additional information}

Supplementary Information is available for this paper.
Correspondence and requests for materials should be addressed to Masahito Hayashi and Hayata Yamasaki.

\paragraph*{Note on related work}After posting our work, another paper~\cite{lami2024solutiongeneralisedquantumsteins} on a proof of generalized quantum Stein's lemma appeared, so we here clarify the difference between our results and the results in  Ref.~\cite{lami2024solutiongeneralisedquantumsteins}.
While Ref.~\cite{lami2024solutiongeneralisedquantumsteins} presents different proof techniques from ours, it still relies on the same assumptions as the previous analyses~\cite{Brandao2010,yamasaki2024generalized} of the generalized quantum Stein's lemma.
In contrast, our approach leads to a stronger version of the generalized quantum Stein's lemma by removing some of these assumptions, while offering a simpler proof.

\section*{Methods}

\paragraph*{Proof of generalized quantum Stein's lemma}

We analyze the task of quantum hypothesis testing to prove the generalized quantum Stein's lemma presented in the main text.
In a conventional setting of quantum hypothesis testing, we have two possible quantum states $\rho$ and $\sigma$ of a finite-dimensional quantum system $\mathcal{H}$ as our two hypotheses.
The performance of the discrimination between these two hypotheses is characterized by
\begin{align}
\beta_\epsilon(\rho\| \sigma)
\coloneqq
\min_{T\in\mathcal{T}_{\epsilon,\rho}} \Tr[T \sigma],
\end{align}
where $\mathcal{T}_{\epsilon,\rho}$ is a set of POVM elements given by
\begin{align}
    \mathcal{T}_{\epsilon,\rho}\coloneqq\qty{T:0\leq T\leq \mathds{1}, \Tr[(\mathds{1}-T)\rho]\le \epsilon},
\end{align}
and $\mathds{1}$ is the identity operator.
To study the generalized quantum Stein's lemma, we consider a more general setting where the hypotheses are given by the state $\rho$ and a compact set $\mathcal{S}$ of states on the system $\mathcal{H}$, where the compactness is assumed to guarantee the existence of maxima and minima in the following analysis.
In this case, the performance is characterized by
\begin{align}
&\beta_\epsilon\left(\rho\middle\| \mathcal{S}\right)\coloneqq\min_{T\in\mathcal{T}_{\epsilon,\rho}}\max_{\sigma \in \mathcal{S}} 
\Tr[T \sigma].
\end{align}
Such a hypothesis is called a composite hypothesis. 

Generally, the quantities satisfy the max-min inequality~\cite[Sec.~5.4.1]{boyd_vandenberghe_2004}
\begin{align}
\label{BH4Y}
\max_{\sigma \in \mathcal{S}}
\beta_{\epsilon}\left(\rho\middle\| \sigma\right) 
&\leq \beta_{ \epsilon}\left(\rho\middle\| \mathcal{S}\right).
\end{align}
Interestingly, as shown in Sec.~\ref{S2} of Supplementary Information,  these two quantities have the following relation when the set $\mathcal{S}$ in the second argument is convex, due to the minimax theorem~\cite{v1928theorie,sion1958general,10.2996/kmj/1138038812}.
\begin{lemma}[Lemma~\ref{L1} in Supplementary Information]\label{L1_methods}
For any $\epsilon\geq 0$, any state $\rho$, and any convex compact set $\mathcal{S}$, we have
\begin{align}
\max_{\sigma \in \mathcal{S}}
\beta_{\epsilon}\left(\rho\middle\| \sigma\right)
= \beta_{\epsilon}\left(\rho\middle\| \mathcal{S}\right).
\label{BH4_methods}
\end{align}
\end{lemma}

To characterize these quantities, we consider the $n$-fold tensor product system
${\cal H}^{\otimes n}$ composed of a $d$-dimensional subsystem ${\cal H}$. 
In this system, one may consider $n$-fold tensor product states
$\rho^{\otimes n}$ and $\sigma^{\otimes n}$.
Then, due to the quantum Stein's lemma~\cite{hiai1991proper,887855}, for any $\epsilon\in(0,1)$, one obtains the relation 
\begin{align}
\label{eq:steins_lemma}
\lim_{n\to \infty} -\frac{1}{n}\log \beta_\epsilon\left(\rho^{\otimes n}\middle\| \sigma^{\otimes n}\right)
=D\left(\rho\middle\|\sigma\right),
\end{align}
where $D(\rho\|\sigma)\coloneqq \Tr[\rho (\log \rho- \log \sigma)]$ is the quantum relative entropy, and the function $\log$ is the natural logarithm throughout this work. 
As its natural extension to the composite hypotheses,
for a sequence $\qty{\mathcal{S}^{(n)}}_n$ of convex compact sets of states of $\mathcal{H}^{\otimes n}$,
we will study the relation between
$\lim_{n\to \infty} -\frac{1}{n}\log \beta_\epsilon\left(\rho^{\otimes n}\middle\| \mathcal{S}^{(n)}\right)$
and $\lim_{n\to \infty} \frac{1}{n}\min_{\sigma \in \mathcal{S}^{(n)}}D\left(\rho^{\otimes n}\middle\|\sigma\right)$.
In contrast to~\eqref{eq:steins_lemma}, it is not clear whether these quantities are equal in general.
To address this problem, 
as in the conditions for the set of free states in the main text,
we introduce the following conditions for the set $\mathcal{S}^{(n)}$.
\begin{enumerate}[label={Set\arabic*}]
\item\label{a_methods}For a $d$-dimensional system $\mathcal{H}$, the set $\mathcal{S}^{(n)}$ is a convex compact subset of the set of states of $\mathcal{H}^{\otimes n}$.\footnote{Due to the finite dimension $d$, the compactness and the closedness of the bounded set $\mathcal{S}^{(n)}$ are equivalent.}
\item\label{b_methods}The set $\mathcal{S}^{(n)}$ is closed under tensor product; i.e., for any positive integers $n$ and $m$, if $\sigma_n\in\mathcal{S}^{(n)}$ and $\sigma_m\in\mathcal{S}^{(m)}$, then $\sigma_n\otimes\sigma_m\in \mathcal{S}^{(n+m)}$.
\item\label{c_methods}The set $\mathcal{S}^{(1)}$ contains a full-rank state $\sigma_\mathrm{full}$, i.e., a state with full support on $\mathcal{H}$.
\end{enumerate}

Our main result is stated as follows.
\begin{theorem}[Generalized quantum Stein's lemma]\label{TH1}
For any $\epsilon\in(0,1)$ and any sequence $\qty{\mathcal{S}^{(n)}}_n$ of sets of states satisfying Conditions~\ref{a_methods},~\ref{b_methods}, and~\ref{c_methods}, it holds that
\begin{align}
&\lim_{n\to \infty} -\frac{1}{n}\log \beta_\epsilon\left(\rho^{\otimes n}\middle\| \mathcal{S}^{(n)}\right)\notag\\
&=\lim_{n\to \infty} \frac{1}{n}\min_{\sigma \in \mathcal{S}^{(n)}}D\left(\rho^{\otimes n}\middle\|\sigma\right).   
\label{ZXI_methods}
\end{align}
\end{theorem}

To prove Theorem~\ref{TH1},
due to Lemma~\ref{L1_methods},
it suffices to deal with $\max_{\sigma \in \mathcal{S}^{(n)}} \beta_{\epsilon}\left(\rho^{\otimes n}\middle\| \sigma\right)$ and $\min_{\sigma \in \mathcal{S}^{(n)}}D\left(\rho^{\otimes n}\middle\|\sigma\right)$; in particular, our strategy focuses on analyzing the following relation
instead of that in Theorem~\ref{TH1}:
\begin{align}
&\lim_{n\to \infty} -\frac{1}{n}\log 
\max_{\sigma \in \mathcal{S}^{(n)}}
\beta_\epsilon\left(\rho^{\otimes n}\middle\| \sigma\right)\nonumber\\
&=\lim_{n\to \infty} \frac{1}{n}\min_{\sigma \in \mathcal{S}^{(n) }}D\left(\rho^{\otimes n}\middle\|\sigma\right).
\label{XBI}
\end{align}
Before proceeding with our proof of~\eqref{XBI}, we note that the existence of limits in~\eqref{XBI} is not assumed a priori;
however, we see that due to Fekete's subadditive lemma~\cite[Lemma~A.1]{hayashi2016quantum}, the limit $\lim_{n\to \infty} \frac{1}{n}\min_{
\sigma \in \mathcal{S}^{(n)}
}D\left(\rho^{\otimes n}\middle\|\sigma\right)$ on the right-hand side of~\eqref{XBI} exists, as shown in the following lemma, which is proven in Sec.~\ref{S2} of Supplementary Information.
We also remark that the convexity of $\mathcal{S}^{(n)}$ is unnecessary for this lemma.
\begin{lemma}[Lemma~\ref{L4} in Supplementary Information]\label{L4_methods}
For any sequence $\qty{\mathcal{S}^{(n)}}_n$ of sets satisfying Conditions~\ref{a_si},~\ref{b_si}, and~\ref{c_si} (except for the convexity of $\mathcal{S}^{(n)}$),
the limit
\begin{align}
\lim_{n\to \infty} \frac{1}{n}\min_{
\sigma \in \mathcal{S}^{(n)}
} D\left(\rho^{\otimes n}\middle\|\sigma\right)
\end{align}
exists.
\end{lemma}
To prove~\eqref{XBI}, due to Lemma~\ref{L4_methods}, our analysis will show
\begin{align}
&\limsup_{n\to\infty}-\frac{1}{n}\log 
\max_{\sigma \in \mathcal{S}^{(n)}}
\beta_\epsilon\left(\rho^{\otimes n}\middle\| \sigma\right)\notag\\
&\leq\lim_{n\to \infty} \frac{1}{n}\min_{
\sigma \in \mathcal{S}^{(n)}
} D\left(\rho^{\otimes n}\middle\|\sigma\right)
\label{XBI_leq}
\end{align}
and
\begin{align}
&\liminf_{n\to\infty}-\frac{1}{n}\log 
\max_{\sigma \in \mathcal{S}^{(n)}}
\beta_\epsilon\left(\rho^{\otimes n}\middle\| \sigma\right)\notag\\
&\geq \lim_{n\to \infty} \frac{1}{n}\min_{
\sigma \in \mathcal{S}^{(n)}
} D\left(\rho^{\otimes n}\middle\|\sigma\right).
\label{XBI_geq}
\end{align}

To show the inequality~\eqref{XBI_leq}, i.e., the strong converse part (see also Proposition~\ref{prp:strong_converse} in the Supplementary Information), we employ the lemma presented below.
While Ref.~\cite{Brandao2010} also provides a valid proof of the strong converse part of the generalized quantum Stein's lemma, it relies on assumptions that we do not impose; in contrast, we present a simpler proof of the strong converse using only the smaller set of assumptions listed as Conditions~\ref{a_methods},~\ref{b_methods}, and~\ref{c_methods}.
Our argument is based on the following lemma, which is proven in Sec.~\ref{S5} of the Supplementary Information by modifying the standard proof of the strong converse for the quantum Stein's lemma in the IID setting~\cite{887855}.

\begin{lemma}[Lemma~\ref{L5} in Supplementary Information]
\label{L5_methods}
Given any fixed positive integer $m$, any state $\rho$ of $\mathcal{H}$, any state $\sigma_m\in\mathcal{S}^{(m)}$, and any full-rank state $\sigma_\mathrm{full}\in\mathcal{S}^{(1)}$,
for each $n\in\{1,2,\ldots\}$ and $l\in\{0,1,\ldots\}$ satisfying $lm\leq n< (l+1)m$, we define
\begin{align}
\label{eq:tilde_sigma_n}
\sigma_{n}&\coloneqq
\sigma_{m}^{\otimes l}
\otimes \sigma_\mathrm{full}^{\otimes (n- lm)}.
\end{align}
For any $\epsilon\in[0,1)$, we have
\begin{align}
&\limsup_{n \to \infty}-\frac{1}{n}\log
\beta_{\epsilon} \left(\rho^{\otimes n} \middle\| \sigma_{n}\right) 
\leq\frac{1}{m}D \left(\rho^{\otimes m}\middle\| \sigma_{m}\right).
\end{align}
\end{lemma}

From Lemma~\ref{L5_methods}, for any $m$, any state $\sigma_{m}\in\mathcal{S}^{(m)}$, and any full-rank state $\sigma_\mathrm{full}\in\mathcal{S}^{(1)}$, it follows that
\begin{align}
&\limsup_{n\to \infty}-\frac{1}{n} \log \max_{\sigma\in \mathcal{S}^{(n)}}\beta_\epsilon \left(\rho^{\otimes n}\middle\|\sigma\right) \nonumber\\
&\leq
\limsup_{n\to \infty}-\frac{1}{n}
\log \beta_\epsilon \left(\rho^{\otimes n}\middle\| \sigma_{n}\right)
\leq\frac{1}{m}
D\left(\rho^{\otimes m}\middle\| \sigma_m\right),
\label{eq:strong_converse_arguement}
\end{align}
where $\sigma_{n}$ is given by~\eqref{eq:tilde_sigma_n}, satisfying $\sigma_{n}\in\mathcal{S}^{(n)}$ under Conditions~\ref{b_methods} and~\ref{c_methods} of $\mathcal{S}^{(n)}$.
Therefore, by choosing $\sigma_m$ as a state minimizing $\min_{\sigma\in\mathcal{S}^{(m)}}D\left(\rho^{\otimes m}\middle\|\sigma\right)$ and also taking the limit $m\to\infty$,
we obtain from Lemma~\ref{L5_methods}
\begin{align}
&\limsup_{n\to \infty}-\frac{1}{n} \log\max_{\sigma\in \mathcal{S}^{(n)}} \beta_\epsilon \left(\rho^{\otimes n}\middle\| \sigma\right) \nonumber\\
&\leq\lim_{m\to \infty}\frac{1}{m}D\left(\rho^{\otimes m}\middle\|\sigma_m\right)=
\lim_{n\to \infty}\frac{1}{n}\min_{
\sigma \in\mathcal{S}^{(n)}
}
D\left(\rho^{\otimes n}\middle\| \sigma\right),
\label{BVC_methods}
\end{align}
which shows the inequality~\eqref{XBI_leq} for any $\epsilon\in[0,1)$.
Thus, recalling Lemma~\ref{L1_methods}, we see that the inequality $\leq$ in~\eqref{ZXI_methods} has been shown for any $\epsilon\in[0,1)$.

Therefore, the key part of the generalized quantum Stein's lemma is the proof of the opposite inequality~\eqref{XBI_geq}, i.e., the direct part (see also Proposition~\ref{prp:direct} in Supplementary Information). 
We here present a proof of the direct part by contradiction, while we also provide a constructive proof in Sec.~\ref{S7} of Supplementary information.
For the proof by contradiction, assume that for any $\epsilon\in(0,1)$,
\begin{align}
&\liminf_{n\to \infty}-\frac{1}{n}\max_{\sigma \in\mathcal{S}^{(n)}}
\log \beta_\epsilon \left(\rho^{\otimes n}\middle\| \sigma\right)\nonumber\\
&<\lim_{n\to \infty}\frac{1}{n}\min_{
\sigma \in\mathcal{S}^{(n)}
}
D\left(\rho^{\otimes n}\middle\| \sigma\right).
\label{BHC_assumption}
\end{align}
We write the left-hand side as
\begin{align}
\label{eq:R_1_methods}
R_{1,\epsilon}\coloneqq\liminf_{n\to \infty}-\frac{1}{n}\max_{\sigma \in\mathcal{S}^{(n)}}
\log \beta_\epsilon \left(\rho^{\otimes n}\middle\| \sigma\right).
\end{align}
To reach the contradiction, we prepare the following lemma, which is shown in Sec~\ref{S7} of Supplementary Information.

\begin{lemma}[Lemma~\ref{L6} in Supplementary Information]
\label{L6_methods}
Given any sequence $\qty{\sigma_n \in\mathcal{S}^{(n)}}_n$ of states in the sets $\mathcal{S}^{(n)}$ satisfying Conditions~\ref{a_methods},~\ref{b_methods}, and~\ref{c_methods}, suppose that
\begin{align}
R_2\coloneqq\liminf_{n\to \infty}\frac{1}{n}
D\left(\rho^{\otimes n}\middle\| \sigma_n\right)
> R_{1,\epsilon}, 
\end{align}
where $R_{1,\epsilon}$ is defined as~\eqref{eq:R_1_methods} for any $\epsilon\in(0,1)$.
Then, for any fixed parameter $\tilde{\epsilon} \in (0,\epsilon)$, there exists a sequence $\qty{\sigma_n' \in\mathcal{S}^{(n)}}_n$ of states such that
\begin{align}
\liminf_{n\to \infty}\frac{1}{n}
D\left(\rho^{\otimes n}\middle\| \sigma_n'\right) -R_{1,\epsilon}
&\leq 
 (1-\tilde{\epsilon})(R_2-R_{1,\epsilon}).
 \label{NNIT_methods}
\end{align}
\end{lemma}

Under the assumption~\eqref{BHC_assumption}, we apply Lemma~\ref{L6_methods} to the optimal sequence $\qty{\sigma_n \in \mathcal{S}^{(n)}}_n$ of states that achieve the minimum of $\min_{\sigma \in \mathcal{S}^{(n)}} D\left(\rho^{\otimes n} \middle\| \sigma\right)$ for each $n$.
This yields an updated sequence $\qty{\sigma_n^{\prime} \in \mathcal{S}^{(n)}}_n$ such that
\begin{align}
&\liminf_{n \to \infty} \frac{1}{n} D\left(\rho^{\otimes n} \middle| \sigma_n'\right)
	-	\liminf_{n \to \infty} \frac{1}{n} D\left(\rho^{\otimes n} \middle| \sigma_n\right)\notag\\
&\leq -\tilde{\epsilon} \qty(R_2 - R_{1,\epsilon})\\
&< 0,
\end{align}
where the first inequality follows from~\eqref{NNIT_methods}.
However, this contradicts the optimality of the choice of $\qty{\sigma_n \in\mathcal{S}^{(n)}}_n$.
Therefore, we have
\begin{align}
&\liminf_{n\to \infty}-\frac{1}{n}\max_{\sigma \in\mathcal{S}^{(n)}}
\log \beta_\epsilon \left(\rho^{\otimes n}\middle\| \sigma\right)\nonumber\\
&\geq\lim_{n\to \infty}\frac{1}{n}\min_{
\sigma \in\mathcal{S}^{(n)}
}
D\left(\rho^{\otimes n}\middle\| \sigma\right),
\label{BHC_methods}
\end{align}
showing the inequality~\eqref{XBI_geq} for any $\epsilon\in(0,1)$.
The combination of~\eqref{XBI_leq} and~\eqref{XBI_geq} completes the proof of~\eqref{XBI} and thus~\eqref{ZXI_methods} in Theorem~\ref{TH1} for any $\epsilon\in(0,1)$.

Finally, we sketch the proof of Lemma~\ref{L6_methods}, which requires our key techniques for addressing the non-IIDness. 
Fix a parameter
\begin{align}
    \epsilon_0\coloneqq\frac{\epsilon-\tilde{\epsilon}}{1-\epsilon}\qty(R_2-R_{1,\epsilon}).
\end{align}
For this fixed $\epsilon_0$, we choose a sufficiently large integer $m$ such that
\begin{align}
\frac{1}{m}
D\left(\rho^{\otimes m}\middle\| \sigma_{m}\right) \le R_2+ \epsilon_0.
\end{align}
Then, using the states
$\sigma_{n}^\ast\in\argmax_{\sigma \in \mathcal{S}^{(n)}} \beta_\epsilon\left(\rho^{\otimes n}\middle\| \sigma\right)$ and
$\sigma_{n}=\sigma_{m}^{\otimes l}\otimes \sigma_\mathrm{full}^{\otimes (n-lm)}$
(as in Lemma~\ref{L5_methods}, where $l$ is the integer satisfying $lm \le n < (l+1)m$),
we construct
\begin{align}
\sigma_n'\coloneqq\frac{1}{3}\qty(\sigma_{n}^\ast+ \sigma_{n}+ \sigma_\mathrm{full}^{\otimes n}
)\in\mathcal{S}^{(n)},
\end{align}
where $\sigma_\mathrm{full}$ is the full-rank state in $\mathcal{S}^{(1)}$.

We then introduce the pinching map ${\cal E}_n$ as follows; here, roughly speaking, the pinching is a technique for making the relevant operators commutative without significantly affecting the quantum relative entropy---up to negligible terms that vanish under regularization.
With $\lambda_\mathrm{full}\in(0,1]$ denoting the constant representing the minimum eigenvalue of the full-rank state $\sigma_{\mathrm{full}}$ in $\mathcal{S}^{(1)}$,
we have
\begin{align}
    \sigma_n' \geq \frac{\sigma_\mathrm{full}^{\otimes n}}{3}\geq \frac{\lambda_\mathrm{full}^n}{3}\mathds{1}=e^{-nc_n}\mathds{1},
\end{align}
where $c_n$ is an $O(1)$ quantity given by
$c_n\coloneqq \log\frac{1}{\lambda_\mathrm{full}} +\frac{\log 3}{n}$.
We write the spectral decomposition of $\sigma_n'$ as
\begin{align}
    \sigma_n'=\sum_{j} \lambda_j^\prime E_j^\prime,
\end{align}
where $\lambda_j^\prime\neq\lambda_{j'}^\prime$ for $j\neq j'$, and $E_j^\prime$ is the projection onto the eigenspace of the eigenvalue $\lambda_j'$.
Then, using a function
\begin{align}
f_n(\lambda)\coloneqq  \left\lceil \frac{\log \lambda + nc_n}{c_n} \right\rceil c_n
- nc_n,
\end{align}
we modify $\sigma_n'$ into a state
\begin{align}
\tilde{\sigma}_n^{\prime}\coloneqq \frac{\sum_{j} e^{f_n\qty(\lambda_j^\prime)} E_j^\prime}{\Tr[\sum_{j} e^{f_n\qty(\lambda_j^\prime)} E_j^\prime]}.
\end{align}
We define the pinching map ${\cal E}_n$ with respect to the state $\tilde{\sigma}_n^{\prime}$; that is, we write the spectral decomposition of $\tilde{\sigma}_n^{\prime}$
as
\begin{align}
    \tilde{\sigma}_n^{\prime}=\sum_{j} \tilde{\lambda}_j^{\prime} \tilde{E}_j^{\prime},
\end{align}
and define ${\cal E}_n$ as
\begin{align}
{\cal E}_n(\sigma)\coloneqq \sum_{j=0}^{d_n-1} \tilde{E}_j^{\prime} \sigma \tilde{E}_j^{\prime},
\end{align}
where $d_n$ represents the number of projections $\tilde{E}_j^{\prime}$ in the pinching map ${\cal E}_n$, turning out to be upper bounded by a polynomial of $n$, i.e., $d_n\leq n+1$,
as shown by~\eqref{eq:d_n} in Sec.~\ref{S7} of Supplementary Information. 
Using the facts that $\left|D\left(\rho \middle\| \tilde{\sigma}_n'\right)-D\left(\rho \middle\| \sigma_n'\right)\right| 
\leq
c_n$ and $D\left(\rho^{\otimes n}\middle\|\mathcal{E}_n\qty(\rho^{\otimes n})\right)\leq\log d_n\leq \log(n+1)$,
we then show that
\begin{align}
\liminf_{n\to \infty}\frac{1}{n}
D\left(\rho^{\otimes n}\middle\| \sigma_n'\right)
=\liminf_{n\to \infty}\frac{1}{n}
D({\cal E}_n\left(\rho^{\otimes n})\middle\| \tilde{\sigma}_n^{\prime}\right),
\end{align}
as in~\eqref{ZX6} of Sec.~\ref{S7} of Supplementary Information
Therefore,~\eqref{NNIT_methods} will follow from the relation
\begin{align}
\liminf_{n\to \infty}\frac{1}{n}
D\left({\cal E}_n(\rho^{\otimes n})\middle\| \tilde{\sigma}_n^{\prime}\right) -R_{1,\epsilon}
&\le 
 (1-\tilde{\epsilon})(R_2-R_{1,\epsilon}),\Label{NNI2T}
\end{align}
where, importantly, ${\cal E}_n(\rho^{\otimes n})$ and $\tilde{\sigma}_n^{\prime}$ commute due to the pinching.

To prove~\eqref{NNI2T}, we employ the information spectrum method~\cite{4069150}, which, roughly speaking, is a technique for analyzing bounds on cumulative probability distributions based on the performance of hypothesis testing.
In particular, for any $\epsilon_1\in(0,1]$, we show the relations
\begin{align}
\liminf_{n\to \infty}-\frac{1}{n}
\log \beta_{\epsilon} \left(\mathcal{E}_n\qty(\rho^{\otimes n})\middle\| \tilde{\sigma}_n'\right)
&\le
R_{1,\epsilon},\Label{ZP2}\\
\limsup_{n\to \infty}-\frac{1}{n}
\log \beta_{1-\epsilon_1} \left(\mathcal{E}_n\qty(\rho^{\otimes n})\middle\| \tilde{\sigma}_n'\right) 
&\le \frac{1}{m} D\left(\rho^{\otimes m}\middle\| \sigma_{m}\right) \nonumber\\
&\le R_2+ \epsilon_0, \Label{ZP1}
\end{align}
which are~\eqref{ZP2T} and~\eqref{ZP1T} in Sec.~\ref{S7} of Supplementary Information.
For any $\epsilon_2>0$, we then define two projections
\begin{align}
P_{n,1}&\coloneqq\qty{ {\cal E}_n\qty(\rho^{\otimes n}) \geq e^{n (R_{1,\epsilon}+\epsilon_2) }\tilde{\sigma}_n' },\\
P_{n,2}&\coloneqq\qty{ {\cal E}_n\qty(\rho^{\otimes n}) \geq e^{n (R_2+ \epsilon_0+\epsilon_2) }\tilde{\sigma}_n' },
\end{align}
where $\{A\geq B\}$ is a projection onto the eigenspaces of nonnegative eigenvalues of the operator $A-B$.
From~\eqref{ZP2} and~\eqref{ZP1}, using the technique from the information spectrum method~\cite{4069150}, we obtain
\begin{align}
\label{eq:liminf_epsilon}
\liminf_{n \to \infty} \Tr[P_{n,1}{\cal E}_n\qty(\rho^{\otimes n})]
&\leq 1-\epsilon,\\
\label{eq:limsup_epsilon}
\limsup_{n \to \infty}\Tr[P_{n,2}{\cal E}_n\qty(\rho^{\otimes n})] &\leq \epsilon_1,
\end{align}
as shown by~\eqref{eq:projection_limsup} and~\eqref{eq:projection_liminf} in Sec.~\ref{S7} of Supplementary Information.
To provide an intuitive understanding of this method, note that since ${\cal E}_n(\rho^{\otimes n})$ and $\tilde{\sigma}_n'$ commute,
we can interpret 
$\frac{1}{n}(\log {\cal E}_n\qty(\rho^{\otimes n}) -\log \tilde{\sigma}_n')$ as a classical random variable.
Moreover, $\frac{1}{n}D\left({\cal E}_n(\rho^{\otimes n})\middle\| \tilde{\sigma}_n'\right)$ corresponds to its expectation value under the probability distribution induced by the state ${\cal E}_n\qty(\rho^{\otimes n})$.
The commutativity of $P_{n,1}$ and $P_{n,2}$ implies  $P_{n,1} \geq P_{n,2}$, owing to the relation $R_{1,\epsilon}+\epsilon_2\leq R_2+\epsilon_0+\epsilon_2$.
The probabilities
$\Tr [P_{n,1}{\cal E}_n\qty(\rho^{\otimes n})]$ and
$\Tr [P_{n,2}{\cal E}_n\qty(\rho^{\otimes n})]$
then characterize the cumulative distribution function, satisfying~\eqref{eq:liminf_epsilon} and~\eqref{eq:limsup_epsilon}.

Therefore, 
as shown in~\eqref{BN2T} of Supplementary Information,
we have
\begin{align}
&\frac{1}{n}D\left({\cal E}_n(\rho^{\otimes n})\middle\| \tilde{\sigma}_n'\right) \nonumber\\
&\leq \Tr [\qty(\mathds{1}-P_{n,1}){\cal E}_n\qty(\rho^{\otimes n})](R_{1,\epsilon}+\epsilon_2)\nonumber\\
&\quad+\Tr [\qty(P_{n,1}-P_{n,2}){\cal E}_n\qty(\rho^{\otimes n})](R_2+ \epsilon_0+\epsilon_2)\nonumber\\
&\quad+\Tr [P_{n,2}{\cal E}_n\qty(\rho^{\otimes n})]c_n^\prime\\
&= \qty(R_{1,\epsilon}+\epsilon_2)\nonumber\\
&\quad+\Tr [ P_{n,1}{\cal E}_n(\rho)] \qty(\qty(R_2+ \epsilon_0+\epsilon_2)-\qty(R_{1,\epsilon}+\epsilon_2))\nonumber\\
&\quad+\Tr [ P_{n,2}{\cal E}_n(\rho)](c_n^\prime-\qty(R_2+ \epsilon_0+\epsilon_2)),
\Label{BN2_methods}
\end{align}
where $c_n'\coloneqq\max\qty{c_n+\frac{c_n}{n},R_2+\epsilon_0+\epsilon_2}=O(1)$.
From~\eqref{eq:liminf_epsilon},~\eqref{eq:limsup_epsilon}, and~\eqref{BN2_methods}, 
taking the limit $n\to \infty$ and then $\epsilon_1,\epsilon_2 \to 0$, we obtain
\begin{align}
&\liminf_{n\to \infty}\frac{1}{n}D\left({\cal E}_n\qty(\rho^{\otimes n})\middle\| \tilde{\sigma}_n'\right)\nonumber\\
&\leq R_{1,\epsilon} +\qty(1-\epsilon)(R_2+\epsilon_0-R_{1,\epsilon}),
\end{align}
which implies~\eqref{NNI2T} and thus completes the proof of Lemma~\ref{L6_methods}, i.e., the key lemma in our proof of the generalized quantum Stein's lemma.

\paragraph*{The second law of QRTs for states and classical-quantum (CQ) channels.}
Having proven the generalized quantum Stein's lemma, we now explore its implications for QRTs\@.
In the main text, we have focused on QRTs for states, i.e., static resources~\cite{Kuroiwa2020,Chitambar2018}.
However, the scope of QRTs is even broader than quantum states.
For example, applications of QRTs include communication scenarios to quantify and analyze properties of quantum channels such as their capacities~\cite{Takagi2020}, for which one needs to consider QRTs for dynamical resources~\cite{Chitambar2018} rather than static resources.
As in the case of QRTs for static resources presented in the main text, special types of QRTs for dynamical resources may have the second law, such as that of athermality~\cite{PhysRevLett.122.200601} and that involving the quantum reverse Shannon theorem~\cite{1035117,berta2011quantum,6757002}.
However, in general, formulating QRTs with the second law for dynamical resources remains challenging due to the inherent difficulty in analyzing quantum channels compared to quantum states.

To address this challenge, we here introduce and study a framework of QRTs for classical-quantum (CQ) channels, which are a fundamental class of quantum channels and well-motivated for applications in communication scenarios~\cite{hayashi2016quantum}.
A measure-and-prepare channel (also known as an entanglement breaking channel)~\cite{holevo1998quantum,doi:10.1142/S0129055X03001709} is a CPTP linear map in the form of
\begin{align}
\label{eq:measure-and-prepare_channel}
    \mathcal{N}\qty(\rho)=\sum_k\Tr[E_k \rho]\rho_k,
\end{align}
where $\{E_k\}_k$ ($\sum_k E_k=\mathds{1}$) is a POVM, and $\rho_k$ for each $k$ is a quantum state.
In this work, a CQ channel is defined as a measure-and-prepare channel in the form of
\begin{align}
\label{eq:CQ_measure-and-prepare_channel}
    \mathcal{N}\qty(\rho)=\sum_k\Tr[\ket{k}\bra{k} \rho]\rho_k.
\end{align}
A CQ channel with one-dimensional input is considered a quantum state.
As presented below in more detail, this framework successfully covers QRTs for states (hence the existing results in Refs.~\cite{Brand_o_2008,Brandao2010,brandao2010reversible,Brandao2015}) as a special case and generalizes them to QRTs for the widely studied class of quantum channels.
To further motivate this setting, in Sec.~\ref{sec:examples} of the Supplementary Information, we present examples of QRTs for both states and CQ channels that fall within the scope of our setting, along with their applications to fundamental tasks in quantum information theory.

To study the conversion of dynamical resources,
we begin with recalling the conventional formulation of QRTs for quantum-quantum (QQ) channels~\cite{Chitambar2018} and then formulate those for CQ channels as its special case.
We write a set of superchannels transforming (QQ) channels from $\mathcal{C}\qty(\mathcal{H}_\mathrm{in}^{(1)}\to\mathcal{H}_\mathrm{out}^{(1)})$ to $\mathcal{C}\qty(\mathcal{H}_\mathrm{in}^{(2)}\to\mathcal{H}_\mathrm{out}^{(2)})$ as
$\mathcal{C}\qty(\qty(\mathcal{H}_\mathrm{in}^{(1)}\to\mathcal{H}_\mathrm{out}^{(1)})\to\qty(\mathcal{H}_\mathrm{in}^{(2)}\to\mathcal{H}_\mathrm{out}^{(2)}))$,
where a superchannel
\begin{align}
\label{eq:quantum_comb}
    \Theta\in\mathcal{C}\qty(\qty(\mathcal{H}_\mathrm{in}^{(1)}\to\mathcal{H}_\mathrm{out}^{(1)})\to\qty(\mathcal{H}_\mathrm{in}^{(2)}\to\mathcal{H}_\mathrm{out}^{(2)}))
\end{align}
is defined as a linear map $\Theta$ that converts any CPTP linear map
\begin{align}
    \mathcal{N}\in\mathcal{C}\qty(\qty(\mathcal{H}_\mathrm{in}^{(\mathrm{aux})}\otimes\mathcal{H}_\mathrm{in}^{(1)})\to\qty(\mathcal{H}_\mathrm{out}^{(\mathrm{aux})}\otimes\mathcal{H}_\mathrm{out}^{(1)}))
\end{align}
into a CPTP linear map
\begin{align}
    \qty(\id\otimes\Theta)\qty(\mathcal{N})\in\mathcal{C}\qty(\qty(\mathcal{H}_\mathrm{in}^{(\mathrm{aux})}\otimes\mathcal{H}_\mathrm{in}^{(2)})\to\qty(\mathcal{H}_\mathrm{out}^{(\mathrm{aux})}\otimes\mathcal{H}_\mathrm{out}^{(2)})),
\end{align}
where $\mathcal{H}_\mathrm{in}^{(\mathrm{aux})}$ and $\mathcal{H}_\mathrm{out}^{(\mathrm{aux})}$ represent any auxiliary systems, and $\id$ is the identity linear supermap that converts any channel $\mathcal{N}^{(\mathrm{aux})}\in\mathcal{C}\qty(\mathcal{H}_\mathrm{in}^{(\mathrm{aux})}\to\mathcal{H}_\mathrm{out}^{(\mathrm{aux})})$ to itself (such $\Theta$ is also represented as a quantum comb~\cite{PhysRevLett.101.060401,PhysRevA.80.022339}).
We may write a set of channels or superchannels as $\mathcal{C}$ if the argument is obvious from the context.
For $j\in\{1,2\}$ and $x\in\{\mathrm{in},\mathrm{out}\}$, we let $\mathds{1}_x^{(j)}$ denote the identity operator on $\mathcal{H}_x^{(j)}$, and $d_x^{(j)}\coloneqq\dim\qty(\mathcal{H}_x^{(j)})$. 
When $d_\mathrm{in}^{(j)}=1$ ($j\in\{1,2\}$), the superchannels reduce to channels (CPTP linear maps) transforming states from $\mathcal{H}_\mathrm{out}^{(1)}$ to $\mathcal{H}_\mathrm{out}^{(2)}$.

For a channel $\mathcal{N}\in\mathcal{C}(\mathcal{H}_\mathrm{in}\to\mathcal{H}_\mathrm{out})$, we let
\begin{align}
\label{eq:Choi}
    J(\mathcal{N})\coloneqq(\id\otimes\mathcal{N})\qty(\Phi_d)\in\mathcal{D}\qty(\mathcal{H}_\mathrm{in}\otimes\mathcal{H}_\mathrm{out})
\end{align}
denote the (normalized) Choi state, where $\id$ is the identity map,
$\Phi_d\coloneqq\ket{\Phi_d}\bra{\Phi_d}$, $\ket{\Phi_d}\coloneqq\frac{1}{\sqrt{d}}\sum_{k=0}^{d-1}\ket{k}\otimes\ket{k}$,
and $d=\dim(\mathcal{H}_\mathrm{in})$.
For any linear map $\mathcal{N}$, the Choi operator $J\qty(\mathcal{N})$ is defined as the operator in the same way as~\eqref{eq:Choi}, and  for any linear map $\Theta$ of $\mathcal{N}$, the corresponding transformation $\tilde{\Theta}$ of Choi operators 
\begin{align}
\label{eq:linear_map_theta}
    &\tilde{\Theta}\coloneqq J\circ\Theta\circ J^{-1}: J(\mathcal{N})\mapsto  J\qty(\Theta\qty(\mathcal{N}))
\end{align}
is a linear map of operators.
As shown in Ref.~\cite[(39) and Theorem~5]{PhysRevA.80.022339}, $\Theta$ is a superchannel 
defined in~\eqref{eq:quantum_comb}
if and only if the linear map $\tilde{\Theta}$ in~\eqref{eq:linear_map_theta} is a CP linear map, and its Choi operator $J_2\coloneqq J\qty(\tilde{\Theta})\geq 0$ satisfies for some operator $J_1\geq 0$
\begin{align}
\label{eq:quantum_comb_condition}
    \Tr_{\mathrm{out}{(2)}}[J_2]=\frac{\mathds{1}_{\mathrm{out}}^{(1)}}{d_{\mathrm{out}}^{(1)}}\otimes J_1,\quad
    \Tr_{\mathrm{in}{(1)}}[J_1]=\frac{\mathds{1}_{\mathrm{in}}^{(2)}}{d_{\mathrm{in}}^{(2)}},
\end{align}
where $\Tr_{\mathrm{out}{(2)}}$ and $\Tr_{\mathrm{in}{(2)}}$ are the partial traces over $\mathcal{H}_\mathrm{out}^{(2)}$ and $\mathcal{H}_\mathrm{in}^{(1)}$, respectively.
By contrast, the CP linear map $\tilde{\Theta}$ is trace-non-increasing if and only if its Choi operator $J_2\geq 0$ satisfies
\begin{align}
\label{eq:TP_condition}
    \Tr_{\mathrm{in}(2),\mathrm{out}(2)}\qty[J_2]\leq\frac{\mathds{1}_\mathrm{in}^{(1)}}{d_{\mathrm{in}}^{(1)}}\otimes\frac{\mathds{1}_\mathrm{out}^{(1)}}{d_{\mathrm{out}}^{(1)}},
\end{align}
where $\Tr_{\mathrm{in}(2),\mathrm{out}(2)}$ is the partial trace over $\mathcal{H}_\mathrm{in}^{(2)}\otimes\mathcal{H}_\mathrm{out}^{(2)}$; hence, in general, the CP linear map $\tilde{\Theta}$ in~\eqref{eq:linear_map_theta} for a superchannel $\Theta$ may increase the trace of operators.
Also, as shown in Ref.~\cite[Theorem~6]{PhysRevA.80.022339}, $\Theta$ is a superchannel if and only if it can be implemented for $\mathcal{N}\in\mathcal{C}\qty(\mathcal{H}_\mathrm{in}^{(1)}\to\mathcal{H}_\mathrm{out}^{(1)})$ as
\begin{align}
\label{eq:implementability}
    \Theta(\mathcal{N})=\mathcal{N}_\mathrm{post}\circ\qty(\id\otimes\mathcal{N})\circ\mathcal{N}_\mathrm{pre}, 
\end{align}
where $\id\in\mathcal{C}\qty(\mathcal{H}^\mathrm{(aux)}\to\mathcal{H}^\mathrm{(aux)})$ is the identity map on an auxiliary system $\mathcal{H}^\mathrm{(aux)}$, and $\mathcal{N}_\mathrm{pre}\in\mathcal{C}\qty(\mathcal{H}_\mathrm{in}^\mathrm{(2)}\to\mathcal{H}^\mathrm{(aux)}\otimes\mathcal{H}_\mathrm{in}^\mathrm{(1)})$ and $\mathcal{N}_\mathrm{post}\in\mathcal{C}\qty(\mathcal{H}^\mathrm{(aux)}\otimes\mathcal{H}_\mathrm{out}^\mathrm{(1)}\to\mathcal{H}_\mathrm{out}^\mathrm{(2)})$ are some CPTP linear maps for pre- and post-processing, respectively.
Due to the equivalence to physical implementability shown in~\eqref{eq:implementability}, it is essential to formulate the framework of QRTs using superchannels that satisfy the condition specified in their definition~\eqref{eq:quantum_comb}.

In our framework of QRTs, we specifically consider a subclass of superchannels that transform CQ channels into CQ channels, where the CQ channels are defined as measure-and-prepare channels---a special case of QQ channels---as given in~\eqref{eq:CQ_measure-and-prepare_channel}.
Let $\mathcal{C}_\mathrm{CQ}\subset\mathcal{C}$ denote the set of CQ channels in the form of~\eqref{eq:CQ_measure-and-prepare_channel}.
We write a set of superchannels transforming CQ channels from $\mathcal{C}_\mathrm{CQ}\qty(\mathcal{H}_\mathrm{in}^{(1)}\to\mathcal{H}_\mathrm{out}^{(1)})$ to $\mathcal{C}_\mathrm{CQ}\qty(\mathcal{H}_\mathrm{in}^{(2)}\to\mathcal{H}_\mathrm{out}^{(2)})$ as
\begin{align}
\label{eq:superchannel_measure-and-prepare}
    &\mathcal{C}_\mathrm{CQ}\qty(\qty(\mathcal{H}_\mathrm{in}^{(1)}\to\mathcal{H}_\mathrm{out}^{(1)})\to\qty(\mathcal{H}_\mathrm{in}^{(2)}\to\mathcal{H}_\mathrm{out}^{(2)}))\nonumber\\
    &\subset \mathcal{C}\qty(\qty(\mathcal{H}_\mathrm{in}^{(1)}\to\mathcal{H}_\mathrm{out}^{(1)})\to\qty(\mathcal{H}_\mathrm{in}^{(2)}\to\mathcal{H}_\mathrm{out}^{(2)})), 
\end{align}
which is a subset of the set of superchannels in~\eqref{eq:quantum_comb} and may also be written as $\mathcal{C}_\mathrm{CQ}$ if the argument is obvious from the context.
For any measure-and-prepare channel $\mathcal{N}$, the Choi state $J(\mathcal{N})$ of $\mathcal{H}_\mathrm{in}\otimes\mathcal{H}_\mathrm{out}$ in~\eqref{eq:Choi} is a separable state between subsystems $\mathcal{H}_\mathrm{in}$ and $\mathcal{H}_\mathrm{out}$~\cite[Theorem~4]{doi:10.1142/S0129055X03001709}; in particular, for any CQ channel $\mathcal{N}$ in~\eqref{eq:CQ_measure-and-prepare_channel}, its Choi state is a CQ state
\begin{align}
\label{eq:choi_CQ}
    J\qty(\mathcal{N})=\sum_{k=0}^{d-1}\frac{1}{d}\ket{k}\bra{k}\otimes\rho_k\in\mathcal{D}_\mathrm{CQ}\qty(\mathcal{H}_\mathrm{in}\otimes\mathcal{H}_\mathrm{out}),
\end{align}
where $d=\dim\qty(\mathcal{H}_\mathrm{in})$, and $\mathcal{D}_\mathrm{CQ}\qty(\mathcal{H}_\mathrm{in}\otimes\mathcal{H}_\mathrm{out})$ denotes the set of CQ states of $\mathcal{H}_\mathrm{in}\otimes\mathcal{H}_\mathrm{out}$.
Then, any superchannel $\Theta\in\mathcal{C}_\mathrm{CQ}$ transforming CQ channels into CQ channels preserves these forms in the sense that if $J\qty(\mathcal{N})$ is in the form of~\eqref{eq:choi_CQ}, then $J\qty(\Theta\qty(\mathcal{N}))$ is also in the form of~\eqref{eq:choi_CQ}. 

QRTs for CQ channels are defined by specifying, as free operations, a family
\begin{align}
&\mathcal{O}\qty(\qty(\mathcal{H}_\mathrm{in}^{(1)}\to\mathcal{H}_\mathrm{out}^{(1)})\to\qty(\mathcal{H}_\mathrm{in}^{(2)}\to\mathcal{H}_\mathrm{out}^{(2)}))\nonumber\\
&\subset\mathcal{C}_\mathrm{CQ}\qty(\qty(\mathcal{H}_\mathrm{in}^{(1)}\to\mathcal{H}_\mathrm{out}^{(1)})\to\qty(\mathcal{H}_\mathrm{in}^{(2)}\to\mathcal{H}_\mathrm{out}^{(2)})) 
\end{align}
of superchannels transforming CQ channels into CQ channels, which we may write $\mathcal{O}$ if the argument is obvious from the context.
Similar to the set of free states in the main text, a set $\mathcal{F}(\mathcal{H}_\mathrm{in}\to\mathcal{H}_\mathrm{out})\subset \mathcal{C}_\mathrm{CQ}(\mathcal{H}_\mathrm{in}\to\mathcal{H}_\mathrm{out})$ of free CQ channels is given by those obtained from any given CQ channel (regardless of how non-resourceful the given CQ channel is) by some superchannels in $\mathcal{O}$, i.e., for given $\mathcal{O}$,
\begin{align}
\label{eq:free_set}
    &\mathcal{F}(\mathcal{H}_\mathrm{in}\to\mathcal{H}_\mathrm{out})\coloneqq\left\{\mathcal{N}_\mathrm{free}\in\mathcal{C}_\mathrm{CQ}(\mathcal{H}_\mathrm{in}\to\mathcal{H}_\mathrm{out}):\right.\nonumber\\
    &\quad\left.\forall\mathcal{N}\in\mathcal{C}_\mathrm{CQ},\exists\Theta\in\mathcal{O}~\text{such that}~\mathcal{N}_\mathrm{free}=\Theta\qty(\mathcal{N})\right\}.
\end{align}
We may also omit the argument to write the set of free CQ channels as $\mathcal{F}$.
When $\dim\qty(\mathcal{H}_\mathrm{in}^{(j)})=1$ ($j\in\{1,2\}$) for CQ channels, this formulation reduces to QRTs for states presented in the main text.

We generalize the properties of QRTs for states presented in the main text to those for CQ channels.
We consider the QRTs with their sets of free CQ channels satisfying the following properties.
\begin{enumerate}[label={CQ\arabic*}]
    \item \label{p1:methods}The set $\mathcal{F}(\mathcal{H}_\mathrm{in}\to\mathcal{H}_\mathrm{out})$ is closed and convex.
    \item \label{p3:methods}For any $\mathcal{N}_\mathrm{free}\in\mathcal{F}\qty(\mathcal{H}_\mathrm{in}\to\mathcal{H}_\mathrm{out})$ and $\mathcal{N}_\mathrm{free}^\prime\in\mathcal{F}(\mathcal{H}_\mathrm{in}^\prime\to\mathcal{H}_\mathrm{out}^\prime)$, it holds that $\mathcal{N}_\mathrm{free}\otimes\mathcal{N}_\mathrm{free}^\prime\in\mathcal{F}(\mathcal{H}_\mathrm{in}\otimes\mathcal{H}_\mathrm{in}^\prime\to\mathcal{H}_\mathrm{out}\otimes\mathcal{H}_\mathrm{out}^\prime)$.
    \item \label{p4:methods}For each $\mathcal{H}_\mathrm{in}$ and $\mathcal{H}_\mathrm{out}$, $\mathcal{F}(\mathcal{H}_\mathrm{in}\to\mathcal{H}_\mathrm{out})$ contains $\mathcal{N}_\mathrm{full}\in\mathcal{F}(\mathcal{H}_\mathrm{in}\to\mathcal{H}_\mathrm{out})$ that outputs a full-rank state $\rho_\mathrm{full}$ of $\mathcal{H}_\mathrm{out}$ for any input; that is, its Choi state is $J(\mathcal{N}_\mathrm{full})=\qty(\mathds{1}/d)\otimes\rho_\mathrm{full}>0$, where $d=\dim(\mathcal{H}_\mathrm{in})$.
\end{enumerate}
Note that, following the conventional formulation of QRTs~\cite{Kuroiwa2020,Chitambar2018}, we have operationally formulated QRTs by specifying free operations $\mathcal{O}$ while Properties~\ref{p1:methods}--\ref{p4:methods} to be used for our analysis are imposed on the set $\mathcal{F}$ of free CQ channels derived from $\mathcal{O}$; in our analysis, $\mathcal{O}$ does not explicitly appear, but QRTs must be appropriately specified via any $\mathcal{O}$ with its $\mathcal{F}$ in~\eqref{eq:free_set} satisfying Properties~\ref{p1:methods}--\ref{p4:methods}.
As explained in the main text, the use of $\mathcal{O}$ itself is insufficient for the second law of QRTs, so we will also introduce a larger class $\tilde{\mathcal{O}}$ of operations below.
Properties~\ref{p1:methods}--\ref{p4:methods} reduce to those of QRTs for states shown in the main text when $\dim(\mathcal{H}_\mathrm{in})=1$.

We introduce generalizations of resource measures presented in the main text. 
In place of the relative entropy of resource, for $\mathcal{N}\in\mathcal{C}_\mathrm{CQ}\qty(\mathcal{H}_\mathrm{in}\to\mathcal{H}_\mathrm{out})$, we define a function
\begin{align}
\label{eq:relative_entropy_of_resource}
R_\mathrm{R}\qty(\mathcal{N})\coloneqq\min_{\mathcal{N}_\mathrm{free}\in\mathcal{F}\qty(\mathcal{H}_\mathrm{in}\to\mathcal{H}_\mathrm{out})}D(J(\mathcal{N})\|J(\mathcal{N}_\mathrm{free})),
\end{align}
and its regularization
\begin{align}
\label{eq:regularized_relative_entropy_resource_methods}
R_\mathrm{R}^\infty\qty(\mathcal{N})\coloneqq\lim_{n\to\infty}\frac{1}{n}R_\mathrm{R}\qty(\mathcal{N}^{\otimes n}).   
\end{align}
For $\mathcal{N}\in\mathcal{C}_\mathrm{CQ}$, generalized robustness is defined as
\begin{align}
\label{eq:R_G_methods}
    R_\mathrm{G}\qty(\mathcal{N})&\coloneqq\min\left\{s\geq 0:\frac{\mathcal{N}+s\mathcal{N}^\prime}{1+s}\in\mathcal{F}\qty(\mathcal{H}_\mathrm{in}\to\mathcal{H}_\mathrm{out}),\right.\nonumber\\
    &\quad\left.{\mathcal{N}^\prime\in\mathcal{C}_\mathrm{CQ}\qty(\mathcal{H}_\mathrm{in}\to\mathcal{H}_\mathrm{out})}\right\}.
\end{align}
These functions reduce to resource measures of QRTs for states defined in the main text when $\dim(\mathcal{H}_\mathrm{in})=1$.

Generalizing the asymptotically resource-non-generating operations in QRTs for states presented in the main text, we here introduce asymptotically free operations in QRTs for CQ channels.
We define a set $\tilde{\mathcal{O}}\Big(\qty(\mathcal{H}_\mathrm{in}^{(1)}\to\mathcal{H}_\mathrm{out}^{(1)})\to \qty(\mathcal{H}_\mathrm{in}^{(2)}\to\mathcal{H}_\mathrm{out}^{(2)})\Big)$ of asymptotically free operations as the set
\begin{align}
\tilde{\mathcal{O}}\coloneqq&\Big\{\Big\{\Theta_n\in\mathcal{C}_\mathrm{CQ}\Big(\qty(\mathcal{H}_\mathrm{in}^{(1)\otimes n}\to\mathcal{H}_\mathrm{out}^{(1)\otimes n})\to\nonumber\\
&\qquad\qty(\mathcal{H}_\mathrm{in}^{(2)\otimes n}\to\mathcal{H}_\mathrm{out}^{(2)\otimes n})\Big)\Big\}_{n=1,2,\ldots}:\nonumber\\
&\qquad\text{$\qty{\Theta_n}_n$ satisfies the following properties}\Big\}   
\end{align}
of all sequences of operations (i.e., superchannels transforming CQ channels into CQ channels, as in~\eqref{eq:superchannel_measure-and-prepare}) satisfying the following properties (we may write this set as $\tilde{\mathcal{O}}$ if the argument is obvious from the context).
\begin{description}
    \item[Asymptotically resource-non-generating property] Any sequence $\qty{\Theta_n}_n\in\tilde{\mathcal{O}}$ of operations in this set asymptotically generates no resource from any free CQ channels in terms of the generalized robustness, i.e., for any sequence 
    $\qty{\mathcal{N}_\mathrm{free}^{(n)}\in\mathcal{F}\qty(\mathcal{H}_\mathrm{in}^{(1)\otimes n}\to\mathcal{H}_\mathrm{out}^{(1)\otimes n})}_n$ of free CQ channels
    \begin{align}
    \label{eq:asymptotically_resource_non_generating_operations}
    \lim_{n\to\infty}
    R_\mathrm{G}\qty(\Theta_n\qty(\mathcal{N}_\mathrm{free}^{(n)}))=0,
    \end{align}
    where $R_\mathrm{G}$ is defined in~\eqref{eq:R_G_methods}.
    \item[Asymptotic continuity] For any two sequences $\qty{\mathcal{N}_n}_n$ and $\qty{\mathcal{N}_n^\prime}_n$ of CQ channels satisfying $\lim_{n\to\infty}\frac{1}{2}\|J\qty(\mathcal{N}_n)-J\qty(\mathcal{N}_n^\prime)\|_1=0$, any sequence $\qty{\Theta_n}_n\in\tilde{\mathcal{O}}$ of operations in this set satisfies
    \begin{align}
    \label{eq:condition_asymptotic_continuity}
        \lim_{n\to\infty}\frac{1}{2}\left\|J\qty(\Theta_n\qty(\mathcal{N}_n))-J\qty(\Theta_n\qty(\mathcal{N}_n^\prime))\right\|_1=0.
    \end{align}
\end{description}
Note that in the QRTs for states,~\eqref{eq:condition_asymptotic_continuity} is automatically satisfied so we may not need to impose~\eqref{eq:condition_asymptotic_continuity} explicitly, as in the main text.
More generally,~\eqref{eq:condition_asymptotic_continuity} also holds in the QRTs for CQ channels such that the CP linear map $\tilde{\Theta}_n: J\qty(\mathcal{N})\mapsto J\qty(\Theta_n\qty(\mathcal{N}))$ of Choi operators in~\eqref{eq:linear_map_theta} is trace-non-increasing for all $\Theta_n$ in $\qty{\Theta_n}_n\in\tilde{\mathcal{O}}$, since the trace distance does not increase under the CP trace-non-increasing linear map.
However, in general, this CP linear map $\tilde{\Theta}_n$ may increase the trace of operators as discussed in~\eqref{eq:quantum_comb_condition} and~\eqref{eq:TP_condition}, and hence, we require~\eqref{eq:condition_asymptotic_continuity} explicitly.
Under $\tilde{\mathcal{O}}$, the asymptotic conversion rate of parallel quantum channels is
\begin{align}
\label{eq:conversion_rate_methods}
    &r_{\tilde{\mathcal{O}}}\qty(\mathcal{N}_1\to\mathcal{N}_2)\coloneqq\sup\left\{r\geq 0:\exists\qty{\Theta_n}_n\in\tilde{\mathcal{O}},\right.\nonumber\\
    &\quad\left.\liminf_{n\to\infty}\frac{1}{2}\left\|J\qty(\Theta_n\qty(\mathcal{N}_1^{\otimes n}))-J\qty(\mathcal{N}_2^{\otimes \lceil rn\rceil})\right\|_1=0\right\}.
\end{align}
These definitions coincide with those defined for QRTs for states in the main text when $\dim(\mathcal{H}_\mathrm{in})=1$.

In this general framework of QRTs for CQ channels, we prove the second law of QRTs, as shown by the following theorem.
See Sec.~\ref{sec:second_law} of Supplementary Information for the details of the proof.

\begin{theorem}[Second law of QRTs for states and CQ channels]
\label{thm:second_law}
Given any family $\mathcal{F}$ of sets of free CQ channels satisfying Properties~\ref{p1:methods}--\ref{p4:methods},
for any CQ channels $\mathcal{N}_1\in\mathcal{C}_\mathrm{CQ}\qty(\mathcal{H}_\mathrm{in}^{(1)}\to\mathcal{H}_\mathrm{out}^{(1)})$ and $\mathcal{N}_2\in\mathcal{C}_\mathrm{CQ}\qty(\mathcal{H}_\mathrm{in}^{(2)}\to\mathcal{H}_\mathrm{out}^{(2)})$ satisfying $R_\mathrm{R}^\infty\qty(\mathcal{N}_j)>0$ ($j\in\{1,2\}$),
the asymptotic conversion rate~\eqref{eq:conversion_rate_methods} between $\mathcal{N}_1$ and $\mathcal{N}_2$ under $\tilde{\mathcal{O}}$  satisfying~\eqref{eq:asymptotically_resource_non_generating_operations} and~\eqref{eq:condition_asymptotic_continuity} is
\begin{align}
    r_{\tilde{\mathcal{O}}}\qty(\mathcal{N}_1\to\mathcal{N}_2)=\frac{R_\mathrm{R}^\infty\qty(\mathcal{N}_1)}{R_\mathrm{R}^\infty\qty(\mathcal{N}_2)},
\end{align}
where $R_\mathrm{R}^\infty$ is defined in~\eqref{eq:regularized_relative_entropy_resource_methods}.
\end{theorem}

To prove Theorem~\ref{thm:second_law}, as has already been observed in the cases of QRTs for states~\cite{Brand_o_2008,brandao2010reversible,Brandao2010,Brandao2015,berta2023gap,Lami2023}, we inevitably need to impose strong conditions on the operations $\tilde{\mathcal{O}}$, which are~\eqref{eq:asymptotically_resource_non_generating_operations} and~\eqref{eq:condition_asymptotic_continuity} in our framework.
For the proof, it is crucial to use the fact that $R_\mathrm{R}^\infty$ is defined using Choi states as in~\eqref{eq:relative_entropy_of_resource}.
For this $R_\mathrm{R}^\infty$, we can prove
\begin{align}
\label{eq:R_R_R_G_characterization}
R_\mathrm{R}^\infty(\mathcal{N}) &= 
\min_{\qty{\tilde{\mathcal{N}}_{n}}}\left\{
\lim_{n\to\infty}
\frac{1}{n}
\log\qty(1+R_\mathrm{G}\qty(\tilde{\mathcal{N}}_{n})):\right.\nonumber\\
&\quad\left.\lim_{n\to \infty}\frac{1}{2}\left\| J\qty(\tilde{\mathcal{N}}_{n})-J(\mathcal{N}^{\otimes n}) \right\|_1
= 0
\right\},
\end{align}
where the minimum on the right-hand side is taken over all sequences $\qty{\tilde{\mathcal{N}}_{n}}_n$ of CQ channels such that the limit $\lim_{n\to\infty}\frac{1}{n} \log\qty(1+R_\mathrm{G}\qty(\tilde{\mathcal{N}}_{n}))$ exists.
To construct a sequence $\qty{\tilde{\mathcal{N}}_{n}}$ achieving the minimum in~\eqref{eq:R_R_R_G_characterization}, it is essential to follow the same argument as our proof of the strong converse part of the generalized quantum Stein's lemma (note that our proof of the generalized quantum Stein's lemma does not use~\eqref{eq:R_R_R_G_characterization}).
We remark that~\eqref{eq:R_R_R_G_characterization} is a generalization of Proposition~II.1 and Corollary~III.2 of Ref.~\cite{Brandao2010} in QRTs for states to those for CQ channels; in particular,~\eqref{eq:R_R_R_G_characterization} implies, in the special case of the QRTs for states, that for any state $\rho$, we have
\begin{align}
\label{eq:QRT_state_characterization}
R_\mathrm{R}^\infty(\rho) &= 
\min_{\qty{\tilde{\rho}_{n}}_n}\left\{
\lim_{n\to\infty}
\frac{1}{n}
\log\qty(1+R_\mathrm{G}\qty(\tilde{\rho}_{n})):\right.\nonumber\\
&\quad\left.\lim_{n\to \infty}\frac{1}{2}\left\| \tilde{\rho}_{n}-\rho^{\otimes n} \right\|_1
= 0
\right\},
\end{align}
where the minimum on the right-hand side is taken over all sequences $\qty{\tilde{\rho}_{n}}_n$ of states such that the limit $\lim_{n\to\infty}\frac{1}{n}
\log\qty(1+R_\mathrm{G}\qty(\tilde{\rho}_{n}))$ exists.
Whereas Proposition~II.1 and Corollary III.2 of Ref.~\cite{Brandao2010} do not prove the existence of the minimum and the limit in~\eqref{eq:QRT_state_characterization}, we prove the stronger statement with such existence.

Then, using the relation~\eqref{eq:R_R_R_G_characterization} and the direct part of the generalized quantum Stein's lemma, we can generalize the argument for QRTs for states in Ref.~\cite{Brandao2010} to prove the direct part of Theorem~\ref{thm:second_law} for QRTs for CQ channels, i.e.,
\begin{align}
    r_{\tilde{\mathcal{O}}}\qty(\mathcal{N}_1\to\mathcal{N}_2)\geq\frac{R_\mathrm{R}^\infty\qty(\mathcal{N}_1)}{R_\mathrm{R}^\infty\qty(\mathcal{N}_2)}.
\end{align}
Note that, in QRTs for channels, a channel divergence would also be used as a more conventional resource measure than $R_\mathrm{R}^\infty$, defined not using the Choi state but using the worst-case input state $\rho$ as~\cite{Gour2019a}
\begin{align}
\label{eq:channel_divergence}
    \tilde{R}_\mathrm{R}\qty(\mathcal{N})\coloneqq\min_{\mathcal{N}_\mathrm{free}\in\mathcal{F}}\max_{\rho}D((\id\otimes\mathcal{N})(\rho)\|(\id\otimes\mathcal{N}_\mathrm{free})(\rho));
\end{align}
however, if one were to use the regularization $\tilde{R}_\mathrm{R}^\infty$ of $\tilde{R}_\mathrm{R}$ in place of $R_\mathrm{R}^\infty$ in Theorem~\ref{thm:second_law}, one would need to extend the generalized quantum Stein's lemma itself to address the non-IIDness of the input $\rho$ in~\eqref{eq:channel_divergence}, in addition to the non-IIDness of $\mathcal{F}$, which poses a further challenge beyond the proof of the generalized quantum Stein's lemma.
To avoid the non-IIDness of $\rho$, our framework is based on the Choi states, which use IID maximally entangled states in place of non-IID $\rho$, and thus can successfully extend the results in Refs.~\cite{Brand_o_2008,brandao2010reversible,Brandao2010,Brandao2015}.

At this point, another challenge arises in proving the other direction of inequality since $R_\mathrm{R}$ may increase under free superchannels, unlike the monotonicity of $\tilde{R}_\mathrm{R}$.
We nevertheless prove, using~\eqref{eq:R_R_R_G_characterization} under the conditions of both~\eqref{eq:asymptotically_resource_non_generating_operations} and~\eqref{eq:condition_asymptotic_continuity}, that its regularization $R_\mathrm{R}^\infty$ has an asymptotic version of monotonicity under $\tilde{\mathcal{O}}$, i.e., for any $\{\Theta_n\}_n\in\tilde{\mathcal{O}}$,
\begin{align}
\label{eq:asymptotically_monotonic}
    R_\mathrm{R}^\infty\qty(\mathcal{N})\geq\liminf_{n\to\infty}\frac{1}{n}R_\mathrm{R}^\infty\qty(\Theta_n\qty(\mathcal{N}^{\otimes n})).
\end{align}
Then, we show that this asymptotic monotonicity~\eqref{eq:asymptotically_monotonic} suffices to prove the converse part of Theorem~\ref{thm:second_law} for QRTs for CQ channels, i.e.,
\begin{align}
    r_{\tilde{\mathcal{O}}}\qty(\mathcal{N}_1\to\mathcal{N}_2)\leq\frac{R_\mathrm{R}^\infty\qty(\mathcal{N}_1)}{R_\mathrm{R}^\infty\qty(\mathcal{N}_2)}.
\end{align}

Lastly, we discuss the applicability and limitations of our framework of QRTs with the second law for CQ channels.
When $\dim(\mathcal{H}_\mathrm{in})=1$ for CQ channels, our results recover the second law of QRTs for states, as originally intended in Ref.~\cite{Brandao2015} assuming the generalized quantum Stein's lemma.
As discussed in the main text, our results cover representative convex QRTs such as those of entanglement, athermality, coherence, asymmetry, and magic states~\cite{Kuroiwa2020,Chitambar2018}.
Moreover, our framework for CQ channels makes it possible to apply QRTs to the study of communication scenarios using CQ channels, as in Ref.~\cite{Takagi2020}.
In Sec.~\ref{sec:examples} of the Supplementary Information, we explicitly present examples of QRTs that fall within the scope of our results.

As for the limitations, Theorem~\ref{thm:second_law} may not be applicable in the following situations.
\begin{enumerate}
    \item Theorem~\ref{thm:second_law} can be used only if $R_\mathrm{R}^\infty$ is positive. QRTs with resource states $\rho$ having $R_\mathrm{R}^\infty(\rho)=0$ are shown in Ref.~\cite[Sec.~IV]{Brandao2010}. While this example may be artificial, a more physically relevant example, a QRT of asymmetry, may also have the case of $R_\mathrm{R}^\infty(\ket{\psi})=0$ for all pure states $\ket{\psi}$, as shown by  Ref.~\cite[Corollary~6]{PhysRevA.80.012307}. In these cases, Theorem~\ref{thm:second_law} is not applicable, similar to the existing results in Ref.~\cite{Brandao2015}.
    \item In general, QRTs may have catalytically replicable resource states~\cite{Kuroiwa2020}, which are resource states $\rho$ such that the asymptotic conversion rate from $\rho$ to $\rho$ itself under some class $\mathcal{O}$ of operations is $r_\mathcal{O}\qty(\rho\to\rho)=\infty$. For example, the QRT of imaginarity~\cite{hickey2018quantifying} has a catalytically replicable state $\ket{S}=\frac{1}{\sqrt{2}}\qty(\ket{0}+i\ket{1})$, i.e., $r_\mathcal{O}\qty(\ket{S}\to\ket{S})=\infty$ under its resource-non-generating operations (see, e.g., the circuit in Fig.~8 of Ref.~\cite{PhysRevA.96.042302}). In this case, we see that $R_\mathrm{R}^\infty\qty(\ket{S})=0$, using Theorem~\ref{thm:second_law} conversely.
    \item QRTs that do not satisfy Properties~\ref{p1:methods}--\ref{p4:methods} are out of the scope of Theorem~\ref{thm:second_law}. These QRTs include non-convex QRTs~\cite{Kuroiwa2020,PhysRevA.104.L020401,PRL,PRA} and infinite-dimensional QRTs~\cite{Kuroiwa2020,PhysRevA.104.L020401,Regula2021,Lami2021,ferrari2023asymptotic,Yamasaki2024}; also, a QRT with no full-rank free state is presented in Ref.~\cite[Example~6]{Kuroiwa2020}.
\end{enumerate}
In Sec.~\ref{sec:examples} of the Supplementary Information, we also discuss further counterexamples that do not satisfy Properties~\ref{p1:methods}–\ref{p4:methods}.
Along with these limitations, as explained above, it is not straightforward to show another generalization of Theorem~\ref{thm:second_law} using worst-case inputs to the channels (i.e., using $\tilde{R}_\mathrm{R}$ in~\eqref{eq:channel_divergence} and the diamond norm for the distance in~\eqref{eq:conversion_rate_methods}) to eliminate the requirement of the asymptotic continuity~\eqref{eq:condition_asymptotic_continuity}, which we leave for future work.
A further generalization to QRTs for channels that are not necessarily in the form of CQ channels is also left as a challenging yet interesting open problem.
However, our contribution is to open a way of generalization that encompasses the existing results on QRTs for states in Refs.~\cite{Brand_o_2008,brandao2010reversible,Brandao2010,Brandao2015} as special cases and has wider applicability to CQ channels, i.e., a fundamental class of channels.

\section*{Data availability}

No data is used in this study.

\section*{Code availability}

No code is used in this study.

\clearpage
\onecolumngrid

\renewcommand{\theequation}{S\arabic{equation}}
\renewcommand{\thetheorem}{S\arabic{theorem}}
\renewcommand{\theproposition}{S\arabic{proposition}}
\renewcommand{\thelemma}{S\arabic{theorem}}
\renewcommand{\thecorollary}{S\arabic{theorem}}
\renewcommand{\thedefinition}{S\arabic{definition}}
\renewcommand{\theremark}{S\arabic{remark}}
\renewcommand{\theexample}{S\arabic{example}}
\setcounter{equation}{0}
\setcounter{theorem}{0}

\section*{Supplementary Information}

The Supplementary Information for ``Generalized Quantum Stein's Lemma and Second Law of Quantum Resource Theories'' is organized as follows.
Section~\ref{sec_proof_of_the_generalized_quantum_steins_lemma} presents the proof of the generalized quantum Stein's lemma.
Section~\ref{sec:second_law} provides the proof of the second law of quantum resource theories (QRTs) for classical-quantum (CQ) channels.
In Sec.~\ref{sec:examples}, we discuss examples of QRTs for both static and dynamical resources that fall within the scope of our results, as well as counterexamples that arise when certain assumptions in our analysis are relaxed.

\section{Proof of the generalized quantum Stein's lemma}
\label{sec_proof_of_the_generalized_quantum_steins_lemma}

In this section, we prove the generalized quantum Stein's lemma, that is, Theorem~\ref{TH1} in Methods.
As formulated in Methods, we write
\begin{align}
\beta_\epsilon(\rho\| \sigma)
\coloneqq
\min_{T\in\mathcal{T}_{\epsilon,\rho}} \Tr[T \sigma],
\end{align}
where $\mathcal{T}_{\epsilon,\rho}$ is given by
\begin{align}
    \mathcal{T}_{\epsilon,\rho}\coloneqq\qty{T:0\leq T\leq \mathds{1}, \Tr[(\mathds{1}-T)\rho]\le \epsilon},
\end{align}
and $\mathds{1}$ is the identity operator.
We also write
\begin{align}
&\beta_\epsilon\left(\rho\middle\| \mathcal{S}\right)\coloneqq\min_{T\in\mathcal{T}_{\epsilon,\rho}}\max_{\sigma \in \mathcal{S}} 
\Tr[T \sigma].
\end{align}
We let
\begin{align}
    D\left(\rho\middle|\middle|\sigma\right)\coloneqq\Tr[\rho(\log\rho-\log\sigma)]
\end{align}
denote the quantum relative entropy, where $\log$ is the natural logarithm throughout this work.

For readability, we here repeat the statement of the generalized quantum Stein's lemma.
As formulated in Methods, we consider the following conditions for the set $\mathcal{S}^{(n)}$ of states.
\begin{enumerate}[label={Set\arabic*}]
\item\label{a_si}For a $d$-dimensional system $\mathcal{H}$, 
the set $\mathcal{S}^{(n)}$ is a convex compact subset of the set of states of $\mathcal{H}^{\otimes n}$.
\item\label{b_si}The set $\mathcal{S}^{(n)}$ is closed under tensor product; i.e., for any positive integers $n$ and $m$, if $\sigma_n\in\mathcal{S}^{(n)}$ and $\sigma_m\in\mathcal{S}^{(m)}$, then $\sigma_n\otimes\sigma_m\in \mathcal{S}^{(n+m)}$.
\item\label{c_si}The set $\mathcal{S}^{(1)}$ contains a full-rank state $\sigma_\mathrm{full}$, i.e., a state with full support on $\mathcal{H}$.
\end{enumerate}
Then, the generalized quantum Stein's lemma is stated as follows.
\begin{theorem}[Generalized quantum Stein's lemma]\label{TH1_supplementary_information}
For any $\epsilon\in(0,1)$ and any sequence $\qty{\mathcal{S}^{(n)}}_n$ of sets of states satisfying Conditions~\ref{a_si},~\ref{b_si}, and~\ref{c_si}, it holds that
\begin{align}
\lim_{n\to \infty} -\frac{1}{n}\log \beta_\epsilon\left(\rho^{\otimes n}\middle\| \mathcal{S}^{(n)}\right)=\lim_{n\to \infty} \frac{1}{n}\min_{\sigma \in \mathcal{S}^{(n)}}D\left(\rho^{\otimes n}\middle\|\sigma\right).   
\label{ZXI}
\end{align}
\end{theorem}

In the following subsections, we will present our proof of Theorem~\ref{TH1_supplementary_information}.
In Sec.~\ref{S2}, for the ease of our analysis, we derive a simplification of Theorem~\ref{TH1_supplementary_information}, given by Proposition~\ref{prop:simplification}, which reduces the proof to the strong converse part and the direct part.
In Sec.~\ref{S5}, we prove the strong converse part, as presented in Proposition~\ref{prp:strong_converse}.
In Sec.~\ref{S7}, we prove the direct part, as presented in Proposition~\ref{prp:direct}.
Finally, we summarize the overall proof in Sec.~\ref{sec:summary}.

\subsection{A simplification of the generalized quantum Stein's lemma}\label{S2}
In this section, for the ease of our analysis, we derive the following equivalent statement to Theorem~\ref{TH1_supplementary_information}.
Accordingly, in the subsequent subsections, we will prove the strong converse part and the direct part of Proposition~\ref{prop:simplification}.

\begin{proposition}[A simplification of the generalized quantum Stein's lemma]\label{prop:simplification}
Theorem~\ref{TH1_supplementary_information}, i.e., the generalized quantum Stein's lemma, holds if, for any $\epsilon\in(0,1)$ and any sequence $\qty{\mathcal{S}^{(n)}}_n$ of sets of states satisfying Conditions~\ref{a_si},~\ref{b_si}, and~\ref{c_si}, the following relations hold. \begin{description}
    \item[Strong converse]
    \begin{align}
    \label{eq:strong_converse}
    \limsup_{n\to \infty} -\frac{1}{n}\log\max_{\sigma\in\mathcal{S}^{(n)}} \beta_\epsilon\left(\rho^{\otimes n}\middle\| \sigma\right)\leq\lim_{n\to \infty} \frac{1}{n}\min_{\sigma \in \mathcal{S}^{(n)}}D\left(\rho^{\otimes n}\middle\|\sigma\right).   
    \end{align}
    \item[Direct]
    \begin{align}
    \label{eq:direct}
    \liminf_{n\to \infty} -\frac{1}{n}\log\max_{\sigma\in\mathcal{S}^{(n)}} \beta_\epsilon\left(\rho^{\otimes n}\middle\| \sigma\right)\geq\lim_{n\to \infty} \frac{1}{n}\min_{\sigma \in \mathcal{S}^{(n)}}D\left(\rho^{\otimes n}\middle\|\sigma\right).   
    \end{align}
\end{description}
\end{proposition}

To derive this proposition, in the rest of this section, we will first analyze the left-hand side of~\eqref{ZXI} in Theorem~\ref{TH1_supplementary_information}, and then the right-hand side.
Regarding the quantity appearing on the left-hand side of~\eqref{ZXI} in Theorem~\ref{TH1_supplementary_information}, as discussed in Methods, we have the following.
\begin{lemma}\label{L1}
For any $\epsilon\geq 0$, any state $\rho$, and any convex compact set $\mathcal{S}$, we have
\begin{align}
\max_{\sigma \in \mathcal{S}}
\beta_{\epsilon}\left(\rho\middle\| \sigma\right)
= \beta_{\epsilon}\left(\rho\middle\| \mathcal{S}\right).
\Label{BH4}
\end{align}
\end{lemma}
\begin{proof}
The minimax theorem~\cite{v1928theorie,sion1958general,10.2996/kmj/1138038812} shows that for any convex compact sets $\mathcal{X}$, $\mathcal{Y}$, and any continuous concave-convex function $f:\mathcal{X}\times\mathcal{Y}\to\mathbb{R}$, it holds that $\max_{x\in\mathcal{X}}\min_{y\in\mathcal{Y}}f(x,y)=\min_{y\in\mathcal{Y}}\max_{x\in\mathcal{X}}f(x,y)$ (References~\cite{sion1958general,10.2996/kmj/1138038812} show a more general theorem than this statement, but this statement suffices for our proof).
In our case, the set $\mathcal{S}$ is a convex compact set, and the set $\mathcal{T}_{\epsilon,\rho}$ is also a convex compact set.
The objective function $\Tr[T\sigma]$ is bilinear and thus continuous and concave-convex.
Therefore, we can apply the minimax theorem to obtain the conclusion.
\end{proof}

As for the quantity on the right-hand side of~\eqref{ZXI} in Theorem~\ref{TH1_supplementary_information}, as discussed in Methods, we see the existence of the limit due to the following lemma.

\begin{lemma}\label{L4}
For any sequence $\qty{\mathcal{S}^{(n)}}_n$ of sets satisfying Conditions~\ref{a_si},~\ref{b_si}, and~\ref{c_si} (except for the convexity of $\mathcal{S}^{(n)}$),
the limit
\begin{align}
\lim_{n\to \infty} \frac{1}{n}\min_{
\sigma \in \mathcal{S}^{(n)}
} D\left(\rho^{\otimes n}\middle\|\sigma\right)
\end{align}
exists.
\end{lemma}
\begin{proof}
The proof is based on Fekete's subadditive lemma~\cite[Lemma~A.1]{hayashi2016quantum}.
For any $n$, we choose a state $\sigma_n$ minimizing $\min_{\sigma\in\mathcal{S}^{(n)}}D\left(\rho^{\otimes n}\middle\| \sigma\right)$, i.e.,
\begin{align}
D\left(\rho^{\otimes n}\middle\| \sigma_n\right)=
\min_{
\sigma\in\mathcal{S}^{(n)}
}
D\left(\rho^{\otimes n}\middle\| \sigma\right).
\end{align}
In the same way, for any $m$, we choose a state $\sigma_m$ minimizing $\min_{\sigma\in\mathcal{S}^{(m)}}D\left(\rho^{\otimes m}\middle\| \sigma\right)$.
The compactness of the sets and the existence of a full-rank state in the sets for the second argument guarantee the existence of the minima.
Since it is also guaranteed that
$\sigma_{n}\otimes \sigma_{m}
\in \mathcal{S}^{(n+m)}$, we have
    \begin{align}
\min_{
\sigma\in\mathcal{S}^{(n+m)}
}D\left(\rho^{\otimes (n+m)}\middle\|\sigma\right)
    &\leq
    D\left(\rho^{\otimes n}\otimes\rho^{\otimes m}\middle\| \sigma_n\otimes \sigma_m\right)
\notag    \\
&= D\left(\rho^{\otimes n}\middle\| \sigma_n\right)
+    D\left(\rho^{\otimes m}\middle\| \sigma_m\right)\notag\\
&=
\min_{
\sigma\in\mathcal{S}^{(n)}
}
D\left(\rho^{\otimes n}\middle\| \sigma\right)
+\min_{
\sigma\in\mathcal{S}^{(m)}
}
D\left(\rho^{\otimes m}\middle\| \sigma\right).
\end{align}
Therefore, applying Fekete's subadditive lemma, we see that the limit 
$\lim_{n\to \infty} \frac{1}{n}\min_{
\sigma\in\mathcal{S}^{(n)}
}
D\left(\rho^{\otimes n}\middle\|\sigma\right)$ exists.
\end{proof}

Therefore, using Lemmas~\ref{L1} and~\ref{L4}, we prove Proposition~\ref{prop:simplification} as follows.

\begin{proof}[Proof of Proposition~\ref{prop:simplification}]
    Due to Lemma~\ref{L4}, the limits on the right-hand sides of~\eqref{eq:strong_converse} and~\eqref{eq:direct} exist.
    Due to $\limsup_{n\to\infty}\eta_n\geq\liminf_{n\to\infty}\eta_n$ for any sequence $\qty{\eta_n}_n$, it follows from~\eqref{eq:strong_converse} and~\eqref{eq:direct} that
    \begin{align}
    \lim_{n\to \infty} -\frac{1}{n}\log\max_{\sigma\in\mathcal{S}^{(n)}} \beta_\epsilon\left(\rho^{\otimes n}\middle\| \sigma\right)=\lim_{n\to \infty} \frac{1}{n}\min_{\sigma \in \mathcal{S}^{(n)}}D\left(\rho^{\otimes n}\middle\|\sigma\right),
    \end{align}
    where $\limsup$ in~\eqref{eq:strong_converse} and $\liminf$ in~\eqref{eq:direct} coincide.
    Therefore, we obtain from Lemma~\ref{L1}
    \begin{align}
    \lim_{n\to \infty} -\frac{1}{n}\log \beta_\epsilon\left(\rho^{\otimes n}\middle\|\mathcal{S}^{(n)}\right)=\lim_{n\to \infty} -\frac{1}{n}\log\max_{\sigma\in\mathcal{S}^{(n)}} \beta_\epsilon\left(\rho^{\otimes n}\middle\| \sigma\right)=\lim_{n\to \infty} \frac{1}{n}\min_{\sigma \in \mathcal{S}^{(n)}}D\left(\rho^{\otimes n}\middle\|\sigma\right).
    \end{align}
\end{proof}

\subsection{The strong converse part of the generalized quantum Stein's lemma}\label{S5}

In this section, we prove the strong converse part of Proposition~\ref{prop:simplification}.
In particular, we prove the following proposition.

\begin{proposition}[The strong converse part of Proposition~\ref{prop:simplification}]\label{prp:strong_converse}
For any $\epsilon\in[0,1)$ and any sequence $\qty{\mathcal{S}^{(n)}}_n$ of sets of states satisfying Conditions~\ref{a_si},~\ref{b_si}, and~\ref{c_si}, it holds that
\begin{align}
\limsup_{n\to \infty} -\frac{1}{n}\log\max_{\sigma\in\mathcal{S}^{(n)}} \beta_\epsilon\left(\rho^{\otimes n}\middle\| \sigma\right)\leq\lim_{n\to \infty} \frac{1}{n}\min_{\sigma \in \mathcal{S}^{(n)}}D\left(\rho^{\otimes n}\middle\|\sigma\right).   
\end{align}
\end{proposition}

To prove this proposition, as described in Methods, we show the following lemma.

\begin{lemma}\label{L5}
Given any fixed positive integer $m$, any state $\rho$ of $\mathcal{H}$, any state $\sigma_m\in\mathcal{S}^{(m)}$, and any full-rank state $\sigma_\mathrm{full}\in\mathcal{S}^{(1)}$,
for each $n\in\{1,2,\ldots\}$ and $l\in\{0,1,\ldots\}$ satisfying $lm\leq n< (l+1)m$, we define
\begin{align}
\label{eq:tilde_sigma_n_si}
\sigma_{n}&\coloneqq
\sigma_{m}^{\otimes l}
\otimes \sigma_\mathrm{full}^{\otimes (n- lm)}.
\end{align}
For any $\epsilon\in[0,1)$, we have
\begin{align}
&\limsup_{n \to \infty}-\frac{1}{n}\log
\beta_{\epsilon} \left(\rho^{\otimes n} \middle\| \sigma_{n}\right) 
\leq\frac{1}{m}D \left(\rho^{\otimes m}\middle\| \sigma_{m}\right).
\end{align}
\end{lemma}

\begin{proof}
To show Lemma~\ref{L5}, we here employ the sandwiched R\'{e}nyi relative entropy
\begin{align}
\label{eq:renyi_relative_entropy}
    \widetilde{D}_\alpha\left(\rho\middle\|\sigma\right)\coloneqq
        \frac{1}{\alpha-1}\log\qty(\Tr[\qty(\sigma^{\frac{1-\alpha}{2\alpha}}\rho\sigma^{\frac{1-\alpha}{2\alpha}})^\alpha]).
\end{align}
Note that for the proof of Lemma~\ref{L5}, we can follow the argument of the original proof of the strong converse part of quantum Stein's lemma in Ref.~\cite{887855} (see also Ref.~\cite[Sec.~3.8]{hayashi2016quantum}) using the Petz R\'{e}nyi relative entropy
$D_{\alpha}(\rho\|\sigma)\coloneqq \frac{1}{\alpha-1}\log \Tr[\rho^{\alpha}\sigma^{1-\alpha}]$ of order $\alpha>1$ in place of $\widetilde{D}_\alpha$ used below; however, we present our analysis in terms of $\widetilde{D}_\alpha$ since $D_\alpha$ may be a looser upper bound than $\widetilde{D}_\alpha$ for $\alpha>1$. 

As shown in Ref.~\cite[Lemma~5]{cooney2016strong}, for any $\alpha>0$ and $\epsilon\in[0,1)$, we have
\begin{align}
\label{eq:beta_upper_renyi}
    -\log \beta_\epsilon(\rho \| \sigma)\leq \widetilde{D}_\alpha(\rho \| \sigma)+\frac{\alpha}{\alpha-1}\log (\frac{1}{1-\epsilon}).
\end{align}
(Note that Ref.~\cite[Lemma~5]{cooney2016strong} considers $\epsilon\in(0,1)$, but the proof, based on the data processing inequality, works also in the case of $\epsilon=0$.)
Since $\sigma_\mathrm{full}\in\mathcal{S}^{(1)}$ has a full support,
$\widetilde{D}_\alpha \left(\rho\middle\| \sigma_\mathrm{full}\right)$ takes a finite value.
Since the additivity of $\widetilde{D}_\alpha$~\cite{10.1063/1.4838856,wilde2014strong} implies
$\widetilde{D}_\alpha \left(\rho^{\otimes n}\middle\| \sigma_{n}\right)
=l \widetilde{D}_\alpha \left(\rho^{\otimes m}\middle\| \sigma_{m}\right)
+(n-lm)\widetilde{D}_\alpha \left(\rho \middle\| \sigma_\mathrm{full}\right)$,
applying~\eqref{eq:beta_upper_renyi} to 
$\rho^{\otimes n}$ and $\sigma_{n}$, we have
\begin{align}
\limsup_{n \to \infty}-\frac{1}{n}\log
\beta_{\epsilon} \left(\rho^{\otimes n} \middle\| \sigma_{n}\right)
&\le
\lim_{n\to \infty}\frac{1}{n}
\widetilde{D}_\alpha \left(\rho^{\otimes n}\middle\| \sigma_{n}\right)\notag\\
&=
\frac{1}{m}\widetilde{D}_\alpha \left(\rho^{\otimes m}\middle\| \sigma_{m}\right),
\end{align}
which holds for any $\alpha>1$.
Taking the limit $\alpha\to 1$ (i.e., $\widetilde{D}_\alpha (\rho\| \sigma)\to D (\rho\| \sigma)$~\cite{10.1063/1.4838856,wilde2014strong}), we have
\begin{align}
\limsup_{n \to \infty}-\frac{1}{n}\log
\beta_{\epsilon} \left(\rho^{\otimes m}
\middle\| \sigma_{n}\right) 
\leq &
\lim_{\alpha \to 1}
\frac{1}{m}\widetilde{D}_\alpha \left(\rho^{\otimes m}\middle\| \sigma_m\right)\notag\\
=& \frac{1}{m}D \left(\rho^{\otimes m}\middle\| \sigma_m\right).
\end{align}
\end{proof}

Using this lemma, we prove Proposition~\ref{prp:strong_converse} as follows.

\begin{proof}[Proof of Proposition~\ref{prp:strong_converse}]
For any $m$, any state $\sigma_{m}\in\mathcal{S}^{(m)}$, and any full-rank state $\sigma_\mathrm{full}\in\mathcal{S}^{(1)}$, let $\sigma_{n}$ denote the state given by~\eqref{eq:tilde_sigma_n_si}, i.e.,
for each $n\in\{1,2,\ldots\}$ and $l\in\{0,1,\ldots\}$ satisfying $lm\leq n< (l+1)m$,
\begin{align}
\sigma_{n}&\coloneqq
\sigma_{m}^{\otimes l}
\otimes \sigma_\mathrm{full}^{\otimes (n- lm)}.
\end{align}
This state satisfies $\sigma_{n}\in\mathcal{S}^{(n)}$ under Conditions~\ref{b_si}, and~\ref{c_si} of $\mathcal{S}^{(n)}$.
Then, due to Lemma~\ref{L5}, it follows that
\begin{align}
&\limsup_{n\to \infty}-\frac{1}{n} \log \max_{\sigma\in \mathcal{S}^{(n)}}\beta_\epsilon \left(\rho^{\otimes n}\middle\|\sigma\right)\leq
\limsup_{n\to \infty}-\frac{1}{n}
\log \beta_\epsilon \left(\rho^{\otimes n}\middle\| \sigma_{n}\right)
\leq\frac{1}{m}
D\left(\rho^{\otimes m}\middle\| \sigma_m\right).
\end{align}

Therefore, by choosing $\sigma_m$ as a state minimizing $\min_{\sigma\in\mathcal{S}^{(m)}}D\left(\rho^{\otimes m}\middle\|\sigma\right)$ and also taking the limit $m\to\infty$,
for any $\epsilon\in[0,1)$,
we obtain from Lemmas~\ref{L4} and~\ref{L5}
\begin{align}
&\limsup_{n\to \infty}-\frac{1}{n} \log\max_{\sigma\in \mathcal{S}^{(n)}} \beta_\epsilon \left(\rho^{\otimes n}\middle\| \sigma\right)\leq\lim_{m\to \infty}\frac{1}{m}D\left(\rho^{\otimes m}\middle\|\sigma_m\right)=
\lim_{n\to \infty}\frac{1}{n}\min_{
\sigma \in\mathcal{S}^{(n)}
}
D\left(\rho^{\otimes n}\middle\| \sigma\right),
\label{BVC}
\end{align}
which yields the conclusion.
\end{proof}

\subsection{Direct part of the generalized quantum Stein's lemma}\label{S7}

In this section, we prove the direct part of Proposition~\ref{prop:simplification}.
In particular, we prove the following proposition.

\begin{proposition}[The direct part of Proposition~\ref{prop:simplification}]\label{prp:direct}
For any $\epsilon\in(0,1)$ and any sequence $\qty{\mathcal{S}^{(n)}}_n$ of sets of states satisfying Conditions~\ref{a_si},~\ref{b_si}, and~\ref{c_si}, it holds that
\begin{align}
\liminf_{n\to \infty} -\frac{1}{n}\log\max_{\sigma\in\mathcal{S}^{(n)}} \beta_\epsilon\left(\rho^{\otimes n}\middle\| \sigma\right)\geq\lim_{n\to \infty} \frac{1}{n}\min_{\sigma \in \mathcal{S}^{(n)}}D\left(\rho^{\otimes n}\middle\|\sigma\right).   
\end{align}
\end{proposition}

As described in Methods, for any $\epsilon\in(0,1)$, we write
\begin{align}
\label{eq:R_1}
R_{1,\epsilon}\coloneqq\liminf_{n\to \infty}-\frac{1}{n}\max_{\sigma \in\mathcal{S}^{(n)}}
\log \beta_\epsilon \left(\rho^{\otimes n}\middle\| \sigma\right);
\end{align}
then, the core of our proof is the following key lemma.

\begin{lemma}\label{L6}
Given any sequence $\qty{\sigma_n \in\mathcal{S}^{(n)}}_n$ of states in the sets $\mathcal{S}^{(n)}$ satisfying Conditions~\ref{a_si},~\ref{b_si}, and~\ref{c_si}, suppose that
\begin{align}
\label{eq:R_2_si}
R_2\coloneqq\liminf_{n\to \infty}\frac{1}{n}
D\left(\rho^{\otimes n}\middle\| \sigma_n\right)
> R_{1,\epsilon}, 
\end{align}
where $R_{1,\epsilon}$ is defined as~\eqref{eq:R_1} for any $\epsilon\in(0,1)$.
Then, for any fixed parameter $\tilde{\epsilon} \in (0,\epsilon)$, there exists a sequence $\qty{\sigma_n' \in\mathcal{S}^{(n)}}_n$ of states such that
\begin{align}
\liminf_{n\to \infty}\frac{1}{n}
D\left(\rho^{\otimes n}\middle\| \sigma_n'\right) -R_{1,\epsilon}
&\leq  (1-\tilde{\epsilon})(R_2-R_{1,\epsilon}).
\label{NNIT}
\end{align}
\end{lemma}

To present the full proof of Proposition~\ref{prp:direct}, we divide our analysis into subsequent subsections.
In Sec.~\ref{sec:fundamental_lemmas}, we prepare fundamental lemmas used for our analysis.
Then, in Sec.~\ref{sec:lemma_1}, we show our proof of Lemma~\ref{L5}.
Finally, in Sec.~\ref{sec:proof_of_the_direct_part}, we provide the proof of Proposition~\ref{prp:direct} using Lemma~\ref{L5}.

\subsubsection{Fundamental lemmas}
\label{sec:fundamental_lemmas}
Here, we prepare several fundamental lemmas for the information spectrum method~\cite{4069150}.
To apply the information spectrum method, we consider general sequences
$\qty{\rho_n}_n$ and $\qty{\sigma_n}_n$ of states.

First, we have the following non-asymptotic formula.
\begin{lemma}\Label{L9A}
For any $\epsilon_3\geq 0$, any states $\rho$ and $\sigma$, and any 
completely positive and trace-preserving (CPTP)
map $\mathcal{E}$, it holds that
\begin{align}
\beta_{\epsilon_3} \left(\mathcal{E}(\rho)\middle\| \mathcal{E}(\sigma)\right)  &\geq \beta_{\epsilon_3} \left(\rho\middle\| \sigma\right). \Label{ZNX_2}
\end{align}
\end{lemma}

\begin{proof}
We define the dual map $\mathcal{E}^*$ as
\begin{align}
\Tr[X \mathcal{E}(Y)]=
\Tr[\mathcal{E}^*(X ) Y].
\end{align}
In the quantum hypothesis testing, the performance of any 
positive operator-valued measure (POVM)
$\{T,\mathds{1}-T\}$ for two states
$\mathcal{E}(\rho)$ and $\mathcal{E}(\sigma)$
can be simulated by 
the performance of any POVM $\qty{\mathcal{E}^*(T),\mathds{1}-\mathcal{E}^*(T)}$ for two states
$\rho$ and $\sigma$.
Hence, we obtain
\begin{align}
 \beta_{\epsilon_3} \left(\mathcal{E}(\rho)\middle\| \mathcal{E}(\sigma)\right)
 &=\min_{T:0\leq T\leq\mathds{1}}\qty{\Tr[\mathcal{E}^\ast\qty(T)\sigma]:\Tr[\qty(\mathds{1}-\mathcal{E}^\ast(T))\rho]\leq\epsilon_3}\notag\\  
 &\geq\min_{T:0\leq T\leq\mathds{1}}\qty{\Tr[T\sigma]:\Tr[\qty(\mathds{1}-T)\rho]\leq\epsilon_3}\notag\\
 &= \beta_{\epsilon_3} \left(\rho\middle\| \sigma\right).
\end{align}
\end{proof}

Then, we have the following asymptotic facts.
\begin{lemma}\Label{L9}
For any $\epsilon_3\in[0,1]$, $\epsilon_4>0$, any sequences
$\qty{\rho_n}_n$ and $\qty{\sigma_n}_n$ of states, and any parameters $\underline{R}$ and $\overline{R}$ satisfying
\begin{align}
\underline{R}&\geq\liminf_{n\to \infty}-\frac{1}{n} \log \beta_{\epsilon_3} \left(\rho_n\middle\| \sigma_n\right), \Label{BI3} \\
\overline{R}&\geq \limsup_{n\to \infty}-\frac{1}{n} \log \beta_{\epsilon_3} \left(\rho_n\middle\| \sigma_n\right),\Label{BI4}
\end{align}
we have
\begin{align}
\liminf_{n\to \infty}\Tr[\qty{ \rho_n \geq e^{n (\underline{R}+\epsilon_4)} \sigma_n}\rho_n ]
& \leq 1- \epsilon_3, \Label{FT2}\\
\limsup_{n\to \infty}\Tr[\qty{ \rho_n \geq e^{n (\overline{R}+\epsilon_4)} \sigma_n}\rho_n]
& \leq 1- \epsilon_3. \Label{FT1}
\end{align}
where $\qty{ A \geq B}$ is a projection onto the eigenspaces of non-negative eigenvalues of $A-B$.
\end{lemma}

\begin{proof}
The conclusion trivially holds in the case of $\epsilon_3=0$. Thus, we henceforth consider the case of $\epsilon_3>0$.

We show~\eqref{FT2} by contradiction.
Assume that
\begin{align}
\liminf_{n\to \infty}\Tr[\qty{ \rho_n \geq e^{n (\underline{R}+\epsilon_4)} \sigma_n}\rho_n ]
> 1- \epsilon_3.
\end{align}
There exists a sufficiently large integer $n_0$ such that
\begin{align}
\Tr[\qty{ \rho_n \geq e^{n (\underline{R}+\epsilon_4)} \sigma_n}\rho_n ]
\geq 1- \epsilon_3 
\end{align}
for every integer $n \ge n_0$.
We choose
\begin{align}
T_{n}\coloneqq\qty{ \rho_n \geq e^{n (\underline{R}+\epsilon_4)} \sigma_n},
\end{align}
which satisfies the condition
\begin{align}
    \Tr[\qty(\mathds{1}-T_{n})\rho_n] \leq \epsilon_3
\end{align}
for every $n \ge n_0$.
Since it holds for all $\rho_n$ that
\begin{align}
\Tr[\qty{ \rho_n \geq e^{n (\underline{R}+\epsilon_4)} \sigma_n} \qty(e^{-n (\underline{R}+\epsilon_4)}\rho_n-\sigma_{n})]
\ge 0,
\end{align}
we have
\begin{align}
\beta_{\epsilon_3} \left(\rho_n\middle\| \sigma_{n}\right)\leq\Tr[T_{n}\sigma_{n} ]
\le \Tr[T_{n} \qty(e^{-n (\underline{R}+\epsilon_4)} \rho_n) 
]
\le e^{-n (\underline{R}+\epsilon_4)},
\end{align}
which implies
\begin{align}
\underline{R}+\epsilon_4
\le -\frac{1}{n} \log \beta_{\epsilon_3} \left(\rho_n\middle\| \sigma_{n}\right) 
\end{align}
for every $n \ge n_0$.
This contradicts~\eqref{BI3}.

Next, we show~\eqref{FT1} by contradiction.
Assume that
\begin{align}
\limsup_{n\to \infty}\Tr[\qty{ \rho_n \geq e^{n (\overline{R}+\epsilon_4)} \sigma_n}\rho_n ]
> 1-\epsilon_3 .
\end{align}
There exists a subsequence $\{n_k\}_k$ such that
\begin{align}
\Tr[\qty{\rho_{n_k} \geq e^{n_k (\overline{R}+\epsilon_4)} \sigma_{n_k}}\rho_{n_k} ]
\geq 1- \epsilon_3
\end{align}
for every $n_k$.
We choose
\begin{align}
    T_{n_k}'\coloneqq\qty{ \rho_{n_k} \geq e^{n_k (\overline{R}+\epsilon_4)} \sigma_{n_k}},
\end{align}
which satisfies the condition
\begin{align}
    \Tr[(\mathds{1}-T_{n_k}')\rho_{n_k}]\leq \epsilon_3
\end{align}
for every $n_k$.
Since it holds for all $\rho_{n_k}$ that
\begin{align}
\Tr[ 
\qty{\rho_{n_k}\geq e^{n_k (\overline{R}+\epsilon_4)}\sigma_{n_k}}\qty(e^{-n_k (\overline{R}+\epsilon_4)} \rho_{n_k} -\sigma_{n_k} )]
\geq 0,
\end{align}
we have
\begin{align}
\beta_{\epsilon_3} \left(\rho_{n_k}\middle\| \sigma_{n_k}\right)
\leq \Tr[T_{n_k}'\sigma_{n_k} ]
\leq \Tr[T_{n_k}'\qty(e^{-n_k (\overline{R}+\epsilon_4)}\rho_{n_k})   
]
\le e^{-n_k (\overline{R}+\epsilon_4)},
\end{align}
which implies
\begin{align}
\overline{R}+\epsilon_4
\le -\frac{1}{n_k} \log \beta_{\epsilon_3} \left(\rho_{n_k}\middle\| \sigma_{n_k}\right)
\end{align}
for every $n_k$.
This contradicts~\eqref{BI4}.
\end{proof}

\begin{lemma}\Label{lem:lemmaS3}
If sequences $\{\sigma_n\}_n$ and $\{\sigma_n'\}_n$ of states asymptotically satisfy an operator inequality
\begin{align}
    \sigma_n\geq e^{-o(n)}\sigma_{n}^\prime,
\end{align}
then for any $\epsilon_3\geq 0$ and $\{\rho_n\}_n$, we have
\begin{align}
\liminf_{n\to \infty}-\frac{1}{n} \log \beta_{\epsilon_3} \left(\rho_n\middle\| \sigma_n\right) 
\leq &
\liminf_{n\to \infty}-\frac{1}{n} \log \beta_{\epsilon_3} \left(\rho_n\middle\| \sigma_{n}^\prime\right),
\\
\limsup_{n\to \infty}-\frac{1}{n} \log \beta_{\epsilon_3} \left(\rho_n\middle\| \sigma_n\right) 
\leq &
\limsup_{n\to \infty}-\frac{1}{n} \log \beta_{\epsilon_3} \left(\rho_n\middle\| \sigma_{n}^\prime\right).
\end{align}
\end{lemma}
\begin{proof}
Due to $\sigma_n \geq e^{-o(n)} \sigma_{n}^\prime$,
we have $e^{-o(n)}\Tr[T\sigma_{n}^\prime]\leq\Tr[T\sigma_{n}]$ for all $T\in\mathcal{T}_{\epsilon_3,\rho_n}$.
Thus, we have
\begin{align}
-\frac{1}{n} \log \beta_{\epsilon_3} \left(\rho_n\middle\| \sigma_n\right)
\leq 
-\frac{1}{n} \log \beta_{\epsilon_3} \left(\rho_n\middle\| \sigma_{n}'\right)+ o(1).
\end{align}
By taking the limit (i.e., $\liminf_{n\to \infty}$ or $\limsup_{n\to \infty}$),
we obtain the desired inequalities.
\end{proof}

\subsubsection{Proof of the key lemma in the direct part of the generalized quantum Stein's lemma}
\label{sec:lemma_1}

\begin{proof}[Proof of Lemma~\ref{L6}]
We divide our proof into two parts: the main part of the overall proof and a proof of the key relation, given by~\eqref{NNI2G}.

\textit{Main part of proof of Lemma~\ref{L6}}:
Fix a parameter
\begin{align}
\label{eq:epsilon_0}
    \epsilon_0\coloneqq\frac{\epsilon-\tilde{\epsilon}}{1-\epsilon}\qty(R_2-R_{1,\epsilon})>0.
\end{align}
For this fixed $\epsilon_0$,
due to the definition~\eqref{eq:R_2_si} of $R_2$,
we choose a sufficiently large integer $m$
such that
\begin{align}
\label{eq:D_m_R_2}
\frac{1}{m}
D\left(\rho^{\otimes m}\middle\| \sigma_{m}\right) \le R_2+\epsilon_0.
\end{align}
As in Lemma~\ref{L5}, we write
\begin{align}
\sigma_n&\coloneqq\sigma_{m}^{\otimes l}\otimes \sigma_\mathrm{full}^{\otimes (n-lm)},
\end{align}
where $l$ is an integer satisfying $lm \leq n < (l+1)m$, and $\sigma_\mathrm{full}$ is a full-rank state in $\mathcal{S}^{(1)}$.
The optimal state maximizing $\max_{\sigma \in \mathcal{S}^{(n)}}\beta_\epsilon\left(\rho^{\otimes n}\middle\| \sigma\right)$ is denoted by
\begin{align}
\sigma_{n}^\ast&\in\argmax_{\sigma \in \mathcal{S}^{(n)}} \beta_\epsilon\left(\rho^{\otimes n}\middle\| \sigma\right).
\end{align}
Then, we construct the state
\begin{align}
\label{eq:sigma_n_prime}
\sigma_n'&\coloneqq\frac{1}{3}\qty(\sigma_{n}^\ast+ \sigma_n+ \sigma_\mathrm{full}^{\otimes n}
)\in\mathcal{S}^{(n)}.
\end{align}
In the following, we will analyze an upper bound of $\liminf_{n\to\infty}\frac{1}{n}D\left(\rho^{\otimes n} \middle\| \sigma_n'\right)$.

For this analysis, we will introduce a pinching map ${\cal E}_n$ as follows.
With $\lambda_\mathrm{full}\in(0,1]$ denoting the constant representing the minimum eigenvalue of the full-rank state $\sigma_{\mathrm{full}}$ in $\mathcal{S}^{(1)}$,
we have
\begin{align}
\label{eq:sigma_prime_lower_bound}
    \sigma_n' \geq \frac{\sigma_\mathrm{full}^{\otimes n}}{3}\geq \frac{\lambda_\mathrm{full}^n}{3}\mathds{1}=e^{-nc_n}\mathds{1},
\end{align}
where $c_n$ is an $O(1)$ quantity given by
\begin{align}
    c_n\coloneqq \log\frac{1}{\lambda_\mathrm{full}} +\frac{\log 3}{n}=O(1)\quad\text{as $n\to\infty$}.
\label{NMS}
\end{align}
We define a function
\begin{align}
\label{eq:f_n}
f_n(\lambda)\coloneqq  \left\lceil \frac{\log \lambda + nc_n}{c_n} \right\rceil c_n
- nc_n,
\end{align}
so that we have
\begin{align}
\label{eq:f_n_ineq}
\log\lambda\leq f_n(\lambda)\leq\log\lambda+c_n.
\end{align}
We write the spectral decomposition of $\sigma_n'$ as
\begin{align}
    \sigma_n'=\sum_{j} \lambda_j^\prime E_j^\prime,
\end{align}
and using $f_n$, we modify $\sigma_n'$ into a state
\begin{align}
\tilde{\sigma}_n^{\prime}\coloneqq \frac{\sum_{j} e^{f_n\qty(\lambda_j^\prime)} E_j^\prime}{\Tr[\sum_{j} e^{f_n\qty(\lambda_j^\prime)} E_j^\prime]}.
\end{align}
Since~\eqref{eq:f_n_ineq} yields
\begin{align}
{\sigma}_n'&\leq\sum_{j} e^{f_n\qty(\lambda_j^\prime)} E_j^\prime\leq
e^{c_n}
{\sigma}_n',\\
1 &\leq \Tr[\sum_{j} e^{f_n\qty(\lambda_j^\prime)} E_j^\prime]\leq
e^{c_n},
\end{align}
it holds that
\begin{align}
e^{-c_n}
{\sigma}_n'\le 
\tilde{\sigma}_n'\le
e^{c_n}
{\sigma}_n'.
\Label{BNF}
\end{align}
Since~\eqref{eq:sigma_prime_lower_bound} indicates $\log \lambda_j^\prime \in \qty[- nc_n,0]$,
by definition of $f_n$ in~\eqref{eq:f_n},
the exponent $f_n\qty(\lambda_j^\prime)$ of each eigenvalue of $\tilde{\sigma}_n^{\prime}$ takes a value in a discrete set
$\{a_{n,j}\}_{j=0,\ldots,n}$ of $n+1$ real numbers, where
\begin{align}
a_{n,j}\coloneqq jc_n-nc_n.
\end{align}
We define the pinching map ${\cal E}_n$ with respect to the state $\tilde{\sigma}_n^{\prime}$; that is, we write the spectral decomposition of $\tilde{\sigma}_n^{\prime}$
as
\begin{align}
    \tilde{\sigma}_n^{\prime}=\sum_{j=0}^{n} \tilde{\lambda}_j^{\prime} \tilde{E}_j^{\prime},
\end{align}
and define ${\cal E}_n$ as
\begin{align}
{\cal E}_n(\sigma)\coloneqq \sum_{j=0}^{n} \tilde{E}_j^{\prime} \sigma \tilde{E}_j^{\prime}.
\end{align}
The number $d_n$ of projections $\tilde{E}_j^{\prime}$ in the definition of the pinching map ${\cal E}_n$ is upper bounded by a polynomial
\begin{align}
\label{eq:d_n}
    d_n\leq n+1.
\end{align}

Using this pinching map $\mathcal{E}_n$ with respect to $\tilde{\sigma}_n^{\prime}$, in the following, we will reduce the analysis of $\liminf_{n\to\infty}\frac{1}{n}D\left(\rho^{\otimes n} \middle\| \sigma_n'\right)$ to that of $\liminf_{n\to\infty}\frac{1}{n}D\left(\mathcal{E}_n\qty(\rho^{\otimes n}) \middle\| \tilde{\sigma}_n^{\prime}\right)$.
Applying the operator monotone function $\log$ 
to the operator inequality~\eqref{BNF} (Ref.~\cite[Sec.~A.4]{hayashi2016quantum}), we have
$\log \sigma_n'- c_n\le -\log \tilde{\sigma}_n'\le \log \sigma_n'+ c_n$, i.e.,
\begin{align}
D\left(\rho^{\otimes n} \middle\| \sigma_n'\right) - c_n
\leq
D\left(\rho^{\otimes n} \middle\| \tilde{\sigma}_n'\right)
\leq
D\left(\rho^{\otimes n} \middle\| \sigma_n'\right)+ c_n,
\end{align}
which implies 
\begin{align}
\liminf_{n\to \infty}\frac{1}{n}
D\left(\rho^{\otimes n} \middle\| \sigma_n'\right) 
=
\liminf_{n\to \infty}\frac{1}{n}
D\left(\rho^{\otimes n} \middle\| \tilde{\sigma}_n'\right).\Label{BVT}
\end{align}
Also, as shown in Ref.~\cite[Lemma 3.1]{hiai1991proper} (see also Ref.~\cite[Exercise 2.8]{hayashi2017group}),
the relation
\begin{align}
D\left(\rho^{\otimes n}\middle\|{\cal E}_n\qty(\rho^{\otimes n})\right)=
D\left(\rho^{\otimes n}\middle\| \tilde{\sigma}_n^{\prime}\right) -
D\left({\cal E}_n\qty(\rho^{\otimes n}) \middle\| \tilde{\sigma}_n^{\prime}\right)
\Label{E59}
\end{align}
holds.
Moreover, as shown in Ref.~\cite[Lemma 3.2]{hiai1991proper} 
(also following from the pinching inequality $\rho^{\otimes n}\leq d_n \mathcal{E}_n\qty(\rho^{\otimes n})$~\cite[Lemma~3.10]{hayashi2016quantum} with the operator monotonicity of $\log$ as in Ref.~\cite[Proposition~S17]{yamasaki2024generalized}),
the relation
\begin{align}
D\left(\rho^{\otimes n}\middle\|{\cal E}_n\qty(\rho^{\otimes n})\right)\leq \log d_n \leq \log (n+1)
\Label{E60}
\end{align}
holds, where the last inequality follows from~\eqref{eq:d_n}.
The combination of~\eqref{BVT},~\eqref{E59}, and~\eqref{E60}
implies the relation
\begin{align}
\liminf_{n\to \infty}\frac{1}{n}
D\left(\rho^{\otimes n} \middle\| \sigma_n'\right)=\liminf_{n\to \infty}\frac{1}{n}
D\left(\mathcal{E}_n\qty(\rho^{\otimes n})\middle\| \tilde{\sigma}_n'\right) 
.\label{ZX6}
\end{align}
Finally, as will be shown in the rest of this proof, we have the relation
\begin{align}
\liminf_{n\to \infty}\frac{1}{n}
D\left(\mathcal{E}_n\qty(\rho^{\otimes n})\middle\| \tilde{\sigma}_n'\right)
-R_{1,\epsilon}
&\leq 
 (1-\tilde{\epsilon})(R_2-R_{1,\epsilon}).\Label{NNI2G}
\end{align}
Therefore,
the combination of~\eqref{ZX6} and~\eqref{NNI2G} completes the proof of Lemma~\ref{L6}.

\textit{Proof of~\eqref{NNI2G}}:
To prove~\eqref{NNI2G}, we begin with bounding
$\liminf_{n\to \infty}-\frac{1}{n} \log \beta_{\epsilon} \left(\mathcal{E}_n\qty(\rho^{\otimes n})\middle\| \tilde{\sigma}_n'\right)$ and $\limsup_{n\to \infty}-\frac{1}{n} \log \beta_{1-\epsilon_1} \left(\mathcal{E}_n\qty(\rho^{\otimes n})\middle\| \tilde{\sigma}_n'\right)$ for any $\epsilon_1\in(0,1]$, using operator inequalities obtained from~\eqref{eq:sigma_n_prime} and~\eqref{BNF}, i.e.,
\begin{align}
\Label{BNF2}
\tilde{\sigma}_n'&\geq\frac{e^{-c_n}}{3}\sigma_n^\ast,\\
\Label{BNF2_2}
\tilde{\sigma}_n'&\geq\frac{e^{-c_n}}{3}\sigma_n,\\
\Label{BNF2_3}
\tilde{\sigma}_n'&\geq\frac{e^{-c_n}}{3}\sigma_\mathrm{full}^{\otimes n}.
\end{align}
We have
\begin{align}
\liminf_{n\to \infty}-\frac{1}{n}
\log \beta_{\epsilon} \left({\cal E}_n\qty(\rho^{\otimes n})\middle\| \tilde{\sigma}_n'\right)
&\stackrel{(a)}{=}\liminf_{n\to \infty}-\frac{1}{n}
\log\beta_{\epsilon} \left({\cal E}_n\qty(\rho^{\otimes n})\middle\| {\cal E}_n\qty(\tilde{\sigma}_n')\right) \notag\\
&\stackrel{(b)}{\le}
\liminf_{n\to \infty}-\frac{1}{n}
\log \beta_{\epsilon} \left(\rho^{\otimes n}\middle\| \tilde{\sigma}_n'\right)\notag\\
&\stackrel{(c)}{\le}
\liminf_{n\to \infty}-\frac{1}{n}
\log \beta_{\epsilon} \left(\rho^{\otimes n}\middle\| \sigma_{n}^\ast\right)\notag\\
&\stackrel{(d)}{=} R_{1,\epsilon},
\Label{ZP2T}
\end{align}
where $(a)$ follows from the relation
${\cal E}_n\qty(\tilde{\sigma}_n')=\tilde{\sigma}_n'$,
$(b)$ from Lemma~\ref{L9A},
$(c)$ from Lemma~\ref{lem:lemmaS3} due to~\eqref{BNF2},
and $(d)$ from the choice of $\sigma_{n}^\ast$.
Similarly, we have
\begin{align}
\limsup_{n\to \infty}-\frac{1}{n}
\log \beta_{1-\epsilon_1} \left(\mathcal{E}_n\qty(\rho^{\otimes n})\middle\| \tilde{\sigma}_n'\right)
&\le \limsup_{n\to \infty}-\frac{1}{n} \log \beta_{1-\epsilon_1} \left(\rho^{\otimes n}\middle\| \tilde{\sigma}_n'\right)\notag\\
&\stackrel{(a)}{\le} \limsup_{n\to \infty}-\frac{1}{n} \log \beta_{1-\epsilon_1} \left(\rho^{\otimes n}\middle\| \sigma_n\right)\notag\\
&\stackrel{(b)}{\leq} \frac{1}{m} D\left(\rho^{\otimes m}\middle\| \sigma_{m}\right)\notag\\ 
&\stackrel{(c)}{\leq} R_2+ \epsilon_0 \Label{ZP1T},
\end{align}
where $(a)$ follows from Lemma~\ref{lem:lemmaS3} due to~\eqref{BNF2_2}, and $(b)$ from Lemma~\ref{L5}, and $(c)$ from~\eqref{eq:D_m_R_2}.

Next, for any $\epsilon_2>0$, as in Lemma~\ref{L9}, we define two projections
\begin{align}
P_{n,1}&\coloneqq\qty{ {\cal E}_n\qty(\rho^{\otimes n}) \geq e^{n (R_{1,\epsilon}+\epsilon_2) }\tilde{\sigma}_n' },\\ 
P_{n,2}&\coloneqq\qty{ {\cal E}_n\qty(\rho^{\otimes n}) \geq e^{n (R_2+ \epsilon_0+\epsilon_2) }\tilde{\sigma}_n' }.
\end{align}
Applying Lemma~\ref{L9} to~\eqref{ZP1T} and~\eqref{ZP2T}, we have
\begin{align}
\label{eq:projection_limsup}
\liminf_{n \to \infty} \Tr[P_{n,1}{\cal E}_n\qty(\rho^{\otimes n})]
&\leq 1-\epsilon,\\
\label{eq:projection_liminf}
\limsup_{n \to \infty}\Tr[P_{n,2}{\cal E}_n\qty(\rho^{\otimes n})] &\leq \epsilon_1.
\end{align}

In the following, we will bound $D\left({\cal E}_n\qty(\rho^{\otimes n})\middle\| \tilde{\sigma}_n'\right)=\Tr[{\cal E}_n\qty(\rho^{\otimes n})\qty(\log{\cal E}_n\qty(\rho^{\otimes n})-\log\tilde{\sigma}_n')]$ using the projections $P_{n,1}$ and $P_{n,2}$.
Since ${\cal E}_n(\rho^{\otimes n})$ and $\tilde{\sigma}_n'$
commute,
we can consider that 
$\frac{1}{n}(\log {\cal E}_n\qty(\rho^{\otimes n}) -\log \tilde{\sigma}_n')$ is a classical random variable,
and $\frac{1}{n}D\left({\cal E}_n(\rho^{\otimes n})\middle\| \tilde{\sigma}_n'\right)$ is its expectation under the distribution
defined by the state ${\cal E}_n\qty(\rho^{\otimes n})$.
In addition, the probabilities
$\Tr [P_{n,1}{\cal E}_n\qty(\rho^{\otimes n})]$ and
$\Tr [P_{n,2}{\cal E}_n\qty(\rho^{\otimes n})]$
characterize its cumulative distribution function.
Since the projections $P_{n,1}$ and $P_{n,2}$ commute and satisfy $ P_{n,1} \geq P_{n,2}$ (due to $R_{1,\epsilon}+\epsilon_2\leq R_2+\epsilon_0+\epsilon_2$), 
we have the decomposition of the identity operator $\mathds{1}$ into projections
\begin{align}
\label{eq:E_n_1}
E_{n,1}&\coloneqq \mathds{1}-P_{n,1},\\
\label{eq:E_n_2}
E_{n,2}&\coloneqq P_{n,1}-P_{n,2},\\
\label{eq:E_n_3}
E_{n,3}&\coloneqq P_{n,2},
\end{align}
satisfying $\sum_{j=1}^{3}E_{n_j}=\mathds{1}$.
Since two states
${\cal E}_n\qty(\rho^{\otimes n})$, $ \tilde{\sigma}_n'$, the projections in $\qty{E_{n,j}}_{j=1,2,3}$ commute,
and
\begin{align}
    E_{n,1}&=\mathds{1}-P_{n,1}=\qty{ {\cal E}_n\qty(\rho^{\otimes n}) < e^{n (R_{1,\epsilon}+\epsilon_2) }\tilde{\sigma}_n' },\\
    E_{n,2} &\le \mathds{1}-P_{n,2}=
\qty{ {\cal E}_n\qty(\rho^{\otimes n}) < e^{n (R_2+ \epsilon_0+\epsilon_2) }\tilde{\sigma}_n' },
\end{align}
we have
\begin{align}
\frac{1}{n}E_{n,1} \qty(\log {\cal E}_n\qty(\rho^{\otimes n}) - \log \tilde{\sigma}_n'  ) 
& \leq \qty(R_{1,\epsilon}+\epsilon_2)E_{n,1},
\Label{TY1}
\\
\frac{1}{n}E_{n,2} \qty(\log {\cal E}_n\qty(\rho^{\otimes n}) - \log \tilde{\sigma}_n'  ) 
&\leq \qty(R_2+ \epsilon_0+\epsilon_2)E_{n,2}.
\Label{TY2}
\end{align}
In addition, 
using the $O(1)$ quantity $c_n$ defined in \eqref{NMS},
we define an $O(1)$ quantity
\begin{align}
c_n^\prime\coloneqq \max\qty{c_n+\frac{c_n}{n},R_2+\epsilon_0+\epsilon_2}=O(1)\quad\text{as $n\to\infty$},
\end{align}
so that we obtain from~\eqref{BNF2_3}
\begin{align}
\tilde{\sigma}_n'\geq  \frac{e^{-c_n}}{3}
\sigma_\mathrm{full}^{\otimes n} \geq
e^{-n\qty(c_n+\frac{c_n}{n})} \mathds{1}\geq e^{-nc_n^\prime} \mathds{1}.
\end{align}
Then, due to the commutativity, we have
\begin{align}
\frac{1}{n}E_{n,3} \qty(\log {\cal E}_n\qty(\rho^{\otimes n}) - \log \tilde{\sigma}_n'  ) 
\leq
- \frac{1}{n}E_{n,3} \log \tilde{\sigma}_n' \leq c_n^\prime E_{n,3}.
\Label{TY3}
\end{align}

Therefore, it holds that
\begin{align}
\frac{1}{n}D\left({\cal E}_n\qty(\rho^{\otimes n})\middle\| \tilde{\sigma}_n'\right)
&= \sum_{j=1}^3 \frac{1}{n}\Tr [{\cal E}_n\qty(\rho^{\otimes n}) 
E_{n,j} 
\qty(\log {\cal E}_n\qty(\rho^{\otimes n}) - \log \tilde{\sigma}_n'  ) ]\notag\\
&\stackrel{(a)}{\leq} \Tr [{\cal E}_n\qty(\rho^{\otimes n}) E_{n,1}](R_{1,\epsilon}+\epsilon_2)+\Tr [{\cal E}_n\qty(\rho^{\otimes n}) E_{n,2}](R_2+ \epsilon_0+\epsilon_2)+\Tr [{\cal E}_n\qty(\rho^{\otimes n}) E_{n,3}]c_n^\prime\notag\\
&\stackrel{(b)}{=} \qty(R_{1,\epsilon}+\epsilon_2)
+\Tr [ P_{n,1}{\cal E}_n(\rho)] \qty(\qty(R_2+ \epsilon_0+\epsilon_2)-\qty(R_{1,\epsilon}+\epsilon_2))
+\Tr [ P_{n,2}{\cal E}_n(\rho)](c_n^\prime-\qty(R_2+ \epsilon_0+\epsilon_2)),
\Label{BN2T}
\end{align}
where
$(a)$ follows from the relations
\eqref{TY1}, \eqref{TY2}, and \eqref{TY3}, and
$(b)$ follows from the definitions of the projections
$E_{n,1}$, $E_{n,2}$, and $E_{n,3}$ in~\eqref{eq:E_n_1},~\eqref{eq:E_n_2}, and~\eqref{eq:E_n_3}.

Using the general relation 
$\liminf_{n \to \infty}(\eta_n+\eta_n')
\le
\liminf_{n \to \infty}\eta_n
+\limsup_{n \to \infty}\eta_n' $ for two sequences
$\eta_n,\eta_n'$,
by taking the limit $n\to \infty$ in \eqref{BN2T}, 
it follows from~\eqref{eq:projection_limsup} and~\eqref{eq:projection_liminf} that
\begin{align}
\liminf_{n\to \infty}\frac{1}{n}D({\cal E}_n(\rho_n)\| \tilde{\sigma}_n')
&\leq 
\qty(R_{1,\epsilon}+\epsilon_2)\nonumber\\
&\quad+\liminf_{n\to\infty}\Tr [ P_{n,1}{\cal E}_n\qty(\rho_n)] \qty(\qty(R_2+ \epsilon_0+\epsilon_2)-\qty(R_{1,\epsilon}+\epsilon_2))\nonumber\\
&\quad+\limsup_{n\to\infty}\Tr [ P_{n,2}{\cal E}_n\qty(\rho_n)](c_n^\prime-\qty(R_2+ \epsilon_0+\epsilon_2))\notag\\
\label{eq:S92}
&\leq 
\qty(R_{1,\epsilon}+\epsilon_2)\nonumber\\
&\quad+\qty(1-\epsilon) \qty(\qty(R_2+ \epsilon_0+\epsilon_2)-\qty(R_{1,\epsilon}+\epsilon_2))\nonumber\\
&\quad+\epsilon_1(c'-\qty(R_2+ \epsilon_0+\epsilon_2))\\
&\to R_{1,\epsilon}+(1-\epsilon)(R_2+\epsilon_0-R_{1,\epsilon})\quad\text{as $\epsilon_1,\epsilon_2 \to 0$},
\end{align}
where $c'\coloneqq\limsup_{n\to\infty}c_n'$.
Since~\eqref{eq:S92} holds for arbitrarily small $\epsilon_1$ and $\epsilon_2$, for the fixed $\epsilon_0$ in~\eqref{eq:epsilon_0}, we obtain
\begin{align}
\liminf_{n\to \infty}\frac{1}{n}D\left({\cal E}_n(\rho^{\otimes n})\middle\|\tilde{\sigma}_n'\right)
&\leq R_{1,\epsilon}+(1-\epsilon)(R_2+\epsilon_0-R_{1,\epsilon})
=R_{1,\epsilon} +\qty(1-\tilde{\epsilon})( R_2-R_{1,\epsilon}),
\end{align}
which yields~\eqref{NNI2G}.
\end{proof}

\subsubsection{Proof of the direct part of the generalized quantum Stein's lemma}
\label{sec:proof_of_the_direct_part}

Using Lemma~\ref{L5}, we prove Proposition~\ref{prp:direct} as follows.
Note that, in Methods, we have presented a short proof of Proposition~\ref{prp:direct} from Lemma~\ref{L6} by contradiction; by contrast, we here present a constructive proof.
A merit of the constructive proof is that we see from the proof how to construct the optimal sequence of the states in the minimization of the regularized relative entropy of resource, clarifying its connection to the optimal state in maximizing the type II error in the corresponding non-IID version of quantum hypothesis testing.

\begin{proof}[Proof of Proposition~\ref{prp:direct}]

For any $\epsilon\in(0,1)$ and $R_{1,\epsilon}$ in~\eqref{eq:R_1}, the goal of our proof is to find a sequence $\qty{\sigma_{n,*}\in\mathcal{S}^{(n)}}_n$ of states such that 
\begin{align}
R_{1,\epsilon}\geq
\liminf_{n \to \infty}
\frac{1}{n}
D\left(\rho^{\otimes n}\middle\| \sigma_{n,*}\right),
\label{NBE}
\end{align}
which would imply the conclusion
\begin{align}
&\liminf_{n\to \infty}-\frac{1}{n}\max_{\sigma \in\mathcal{S}^{(n)}}
\log \beta_\epsilon \left(\rho^{\otimes n}\middle\| \sigma\right)\geq\lim_{n\to \infty}\frac{1}{n}\min_{
\sigma \in\mathcal{S}^{(n)}
}
D\left(\rho^{\otimes n}\middle\| \sigma\right).
\label{BHC}
\end{align}
    
Given $R_{1,\epsilon}$ in~\eqref{eq:R_1}, we start with any sequence $\qty{\sigma_n\in\mathcal{S}^{(n)}}_n$ that does not satisfy~\eqref{NBE}.
We apply Lemma~\ref{L6} to 
the sequence 
$\qty{\sigma_n\in\mathcal{S}^{(n)}}_n$, which guarantees the existence of an updated sequence $\qty{\sigma_{n,1}}_n$ satisfying
\begin{align}
\liminf_{n\to \infty}\frac{1}{n}
D\left(\rho^{\otimes n}\middle\| \sigma_{n,1}\right) -R_{1,\epsilon}
&\le 
(1-\tilde{\epsilon})(R_2-R_{1,\epsilon}).
\end{align}
Applying Lemma~\ref{L6} to the case when 
$\{\sigma_n\}$ is $\{\sigma_{n,1}\}$, 
we obtain a sequence
$\qty{\sigma_{n,2}}_n$ satisfying
\begin{align}
\liminf_{n\to \infty}\frac{1}{n}
D\left(\rho^{\otimes n}\middle\| \sigma_{n,2}\right) -R_{1,\epsilon}
&\le 
(1-\tilde{\epsilon})^2(R_2-R_{1,\epsilon}).
\end{align}
When Lemma~\ref{L6} is applied $k$ times, there exists a sequence 
$\qty{\sigma_{n,k}}_n$ to satisfy
\begin{align}
\liminf_{n\to \infty}\frac{1}{n}
D\left(\rho^{\otimes n}\middle\| \sigma_{n,k}\right) -R_{1,\epsilon}
&\le 
(1-\tilde{\epsilon})^k(R_2-R_{1,\epsilon}).\Label{ZXC}
\end{align}
Due to \eqref{ZXC}, for any vanishing sequence $\{\epsilon_k>0\}_k$ ($\lim_{k\to\infty}\epsilon_k=0$), there exists a subsequence $\{n_k\}_k$ such that for all $k$
\begin{align}
\frac{1}{n_k}
D\left(\rho^{\otimes n_k}\middle\| \sigma_{n_k,k}\right) -R_{1,\epsilon}
&\le 
(1-\tilde{\epsilon})^k(R_2-R_{1,\epsilon})+\epsilon_k\\
&\to 0\quad\text{as $k\to\infty$}.
\end{align}
Therefore, a sequence $\{\sigma_{n,*}\}_n$ such that $\sigma_{n_k,*}=\sigma_{n_k,k}$ for all $k$ satisfies~\eqref{NBE}.
\end{proof}

\subsection{Overall proof of the generalized quantum Stein's lemma}
\label{sec:summary}
To summarize, our proof of the generalized quantum Stein's lemma in Theorem~\ref{TH1_supplementary_information} is completed as follows.

\begin{proof}[Proof of Theorem~\ref{TH1_supplementary_information}]
Due to Proposition~\ref{prp:strong_converse}, we have the strong converse part
\begin{align}
\limsup_{n\to \infty} -\frac{1}{n}\log\max_{\sigma\in\mathcal{S}^{(n)}} \beta_\epsilon\left(\rho^{\otimes n}\middle\| \sigma\right)\leq\lim_{n\to \infty} \frac{1}{n}\min_{\sigma \in \mathcal{S}^{(n)}}D\left(\rho^{\otimes n}\middle\|\sigma\right).   
\end{align}
Also, due to Proposition~\ref{prp:direct}, we obtain the direct part
\begin{align}
\liminf_{n\to \infty} -\frac{1}{n}\log\max_{\sigma\in\mathcal{S}^{(n)}} \beta_\epsilon\left(\rho^{\otimes n}\middle\| \sigma\right)\geq\lim_{n\to \infty} \frac{1}{n}\min_{\sigma \in \mathcal{S}^{(n)}}D\left(\rho^{\otimes n}\middle\|\sigma\right).   
\end{align}
With these two parts, Proposition~\ref{prop:simplification} indicates that Theorem~\ref{TH1_supplementary_information} holds, i.e.,
\begin{align}
\lim_{n\to \infty} -\frac{1}{n}\log \beta_\epsilon\left(\rho^{\otimes n}\middle\| \mathcal{S}^{(n)}\right)=\lim_{n\to \infty} \frac{1}{n}\min_{\sigma \in \mathcal{S}^{(n)}}D\left(\rho^{\otimes n}\middle\|\sigma\right).   
\end{align}
\end{proof}

\section{Proof of the second law of QRTs for CQ channels}
\label{sec:second_law}

As an application of the generalized quantum Stein's lemma to QRTs, we prove the second law of QRTs for CQ channels, that is, Theorem~\ref{thm:second_law} in Methods, which includes that for states as a special case. 
As formulated in Methods, 
with $J(\mathcal{N})\coloneqq(\id\otimes\mathcal{N})\qty(\Phi_d)$ denoting the (normalized) Choi state of a channel $\mathcal{N}\in\mathcal{C}(\mathcal{H}_\mathrm{in}\to\mathcal{H}_\mathrm{out})$, where $\id$ is the identity map, $\Phi_d\coloneqq\ket{\Phi_d}\bra{\Phi_d}$, $\ket{\Phi_d}\coloneqq\frac{1}{\sqrt{d}}\sum_{k=0}^{d-1}\ket{k}\otimes\ket{k}$, and $d=\dim(\mathcal{H}_\mathrm{in})$,
consider any QRT with a family of sets $\mathcal{F}(\mathcal{H}_\mathrm{in}\to\mathcal{H}_\mathrm{out})\subset\mathcal{C}_\mathrm{CQ}(\mathcal{H}_\mathrm{in}\to\mathcal{H}_\mathrm{out})$ of free CQ channels satisfying the following properties.
\begin{enumerate}[label={CQ\arabic*}]
    \item \label{p1:si}The set $\mathcal{F}(\mathcal{H}_\mathrm{in}\to\mathcal{H}_\mathrm{out})$ is closed and convex.
    \item \label{p3:si}For any $\mathcal{N}_\mathrm{free}\in\mathcal{F}\qty(\mathcal{H}_\mathrm{in}\to\mathcal{H}_\mathrm{out})$ and $\mathcal{N}_\mathrm{free}^\prime\in\mathcal{F}(\mathcal{H}_\mathrm{in}^\prime\to\mathcal{H}_\mathrm{out}^\prime)$, it holds that $\mathcal{N}_\mathrm{free}\otimes\mathcal{N}_\mathrm{free}^\prime\in\mathcal{F}(\mathcal{H}_\mathrm{in}\otimes\mathcal{H}_\mathrm{in}^\prime\to\mathcal{H}_\mathrm{out}\otimes\mathcal{H}_\mathrm{out}^\prime)$.
    \item \label{p4:si}For each $\mathcal{H}_\mathrm{in}$ and $\mathcal{H}_\mathrm{out}$, $\mathcal{F}(\mathcal{H}_\mathrm{in}\to\mathcal{H}_\mathrm{out})$ contains $\mathcal{N}_\mathrm{full}\in\mathcal{F}(\mathcal{H}_\mathrm{in}\to\mathcal{H}_\mathrm{out})$ that outputs a full-rank state $\rho_\mathrm{full}$ of $\mathcal{H}_\mathrm{out}$ for any input; that is, its Choi state is $J(\mathcal{N}_\mathrm{full})=\qty(\mathds{1}/d)\otimes\rho_\mathrm{full}>0$, where $d=\dim(\mathcal{H}_\mathrm{in})$.
\end{enumerate}
We define
\begin{align}
\label{eq:relative_entropy_of_resource_si}
R_\mathrm{R}\qty(\mathcal{N})&\coloneqq\min_{\mathcal{N}_\mathrm{free}\in\mathcal{F}\qty(\mathcal{H}_\mathrm{in}\to\mathcal{H}_\mathrm{out})}D(J(\mathcal{N})\|J(\mathcal{N}_\mathrm{free})),\\
\label{eq:regularized_relative_entropy_resource_si}
R_\mathrm{R}^\infty\qty(\mathcal{N})&\coloneqq\lim_{n\to\infty}\frac{1}{n}R_\mathrm{R}\qty(\mathcal{N}^{\otimes n}),\\
\label{eq:R_G_si}
R_\mathrm{G}\qty(\mathcal{N})&\coloneqq\min\left\{s\geq 0:\frac{\mathcal{N}+s\mathcal{N}^\prime}{1+s}\in\mathcal{F}\qty(\mathcal{H}_\mathrm{in}\to\mathcal{H}_\mathrm{out}),{\mathcal{N}^\prime\in\mathcal{C}_\mathrm{CQ}\qty(\mathcal{H}_\mathrm{in}\to\mathcal{H}_\mathrm{out})}\right\}.
\end{align}
Let $\tilde{\mathcal{O}}$ denote the set of asymptotically free operations satisfying the following.
\begin{description}
    \item[Asymptotically resource-non-generating property] Any sequence $\qty{\Theta_n}_n\in\tilde{\mathcal{O}}$ of operations in this set asymptotically generates no resource from any free CQ channels in terms of the generalized robustness, i.e., for any sequence 
    $\qty{\mathcal{N}_\mathrm{free}^{(n)}\in\mathcal{F}\qty(\mathcal{H}_\mathrm{in}^{(1)\otimes n}\to\mathcal{H}_\mathrm{out}^{(1)\otimes n})}_n$ of free CQ channels
    \begin{align}
    \label{eq:asymptotically_resource_non_generating_operations_si}
    \lim_{n\to\infty}
    R_\mathrm{G}\qty(\Theta_n\qty(\mathcal{N}_\mathrm{free}^{(n)}))=0,
    \end{align}
    where $R_\mathrm{G}$ is defined in~\eqref{eq:R_G_si}.
    \item[Asymptotic continuity] For any two sequences $\qty{\mathcal{N}_n}_n$ and $\qty{\mathcal{N}_n^\prime}_n$ of CQ channels satisfying $\lim_{n\to\infty}\frac{1}{2}\|J\qty(\mathcal{N}_n)-J\qty(\mathcal{N}_n^\prime)\|_1=0$, any sequence $\qty{\Theta_n}_n\in\tilde{\mathcal{O}}$ of operations in this set satisfies
    \begin{align}
    \label{eq:condition_asymptotic_continuity_si}
        \lim_{n\to\infty}\frac{1}{2}\left\|J\qty(\Theta_n\qty(\mathcal{N}_n))-J\qty(\Theta_n\qty(\mathcal{N}_n^\prime))\right\|_1=0.
    \end{align}
\end{description}
Under $\tilde{\mathcal{O}}$, the asymptotic conversion rate of parallel quantum channels is
\begin{align}
\label{eq:conversion_rate_si}
    &r_{\tilde{\mathcal{O}}}\qty(\mathcal{N}_1\to\mathcal{N}_2)\coloneqq\sup\left\{r\geq 0:\exists\qty{\Theta_n}_n\in\tilde{\mathcal{O}},\liminf_{n\to\infty}\frac{1}{2}\left\|J\qty(\Theta_n\qty(\mathcal{N}_1^{\otimes n}))-J\qty(\mathcal{N}_2^{\otimes \lceil rn\rceil})\right\|_1=0\right\}.
\end{align}

For readability, we here repeat the statement of the second law of QRTs for states and CQ channels.
\begin{theorem}[Second law of QRTs for states and CQ channels]
\label{thm:second_law_si}
Given any family $\mathcal{F}$ of sets of free CQ channels satisfying Properties~\ref{p1:si},\ref{p3:si}, and~\ref{p4:si},
for any CQ channels $\mathcal{N}_1\in\mathcal{C}_\mathrm{CQ}\qty(\mathcal{H}_\mathrm{in}^{(1)}\to\mathcal{H}_\mathrm{out}^{(1)})$ and $\mathcal{N}_2\in\mathcal{C}_\mathrm{CQ}\qty(\mathcal{H}_\mathrm{in}^{(2)}\to\mathcal{H}_\mathrm{out}^{(2)})$ satisfying $R_\mathrm{R}^\infty\qty(\mathcal{N}_j)>0$ ($j\in\{1,2\}$),
the asymptotic conversion rate~\eqref{eq:conversion_rate_si} between $\mathcal{N}_1$ and $\mathcal{N}_2$ under $\tilde{\mathcal{O}}$  satisfying~\eqref{eq:asymptotically_resource_non_generating_operations_si} and~\eqref{eq:condition_asymptotic_continuity_si} is
\begin{align}
    r_{\tilde{\mathcal{O}}}\qty(\mathcal{N}_1\to\mathcal{N}_2)=\frac{R_\mathrm{R}^\infty\qty(\mathcal{N}_1)}{R_\mathrm{R}^\infty\qty(\mathcal{N}_2)},
\end{align}
where $R_\mathrm{R}^\infty$ is defined in~\eqref{eq:regularized_relative_entropy_resource_si}.
\end{theorem}

To show this, in Sec.~\ref{sec:relation}, we will first provide Lemmas~\ref{lem:regularized_relative_entropy_generalized_robustness} and~\ref{lem:converse_robustness} to show Corollary~\ref{cor:robustness_characterization}, which provides a relation between $R_\mathrm{R}^\infty$ and $R_\mathrm{G}$, generalizing Proposition~II.1 and Corollary III.2 in Ref.~\cite{Brandao2010} in QRTs for states to those for CQ channels.
Using Corollary~\ref{cor:robustness_characterization}, we analyze the direct part in Sec.~\ref{sec:direct_second_law}, where we will show $r_{\tilde{\mathcal{O}}}\qty(\mathcal{N}_1\to\mathcal{N}_2)\geq\frac{R_\mathrm{R}^\infty\qty(\mathcal{N}_1)}{R_\mathrm{R}^\infty\qty(\mathcal{N}_2)}$ in Proposition~\ref{prp:direct_asymptotic_conversion}.
As for the other direction of inequality, we analyze the converse part in Sec.~\ref{sec:converse_second_law}, where, by showing an asymptotic version of monotonicity of $R_\mathrm{R}^\infty$ in Lemma~\ref{lem:monotonicity}, we will show $r_{\tilde{\mathcal{O}}}\qty(\mathcal{N}_1\to\mathcal{N}_2)\leq\frac{R_\mathrm{R}^\infty\qty(\mathcal{N}_1)}{R_\mathrm{R}^\infty\qty(\mathcal{N}_2)}$ in Proposition~\ref{prp:converse}.
As a whole, the proof of Theorem~\ref{thm:second_law_si} is summarized as follows.
\begin{proof}[Proof of Theorem~\ref{thm:second_law_si}]
The combination of Propositions~\ref{prp:direct_asymptotic_conversion} and~\ref{prp:converse} completes the proof of Theorem~\ref{thm:second_law_si}.
\end{proof}

\subsection{Relation between the regularized relative entropy of resource and the generalized robustness}
\label{sec:relation}
We discuss a relation between the regularized relative entropy of resource and the generalized robustness for CQ channels, which extends and simplifies their relation between those for quantum states shown by Ref.~\cite{Brandao2010}.

We first provide an upper bound of the relative entropy of resource in terms of the generalized robustness.

\begin{lemma}\label{lem:regularized_relative_entropy_generalized_robustness}
    For any sequence $\qty{\mathcal{N}_n}_n$ of CQ channels, we have
    \begin{align}
    \label{eq:R_R_R_G}
        \liminf_{n\to\infty}\frac{1}{n}R_\mathrm{R}\qty(\mathcal{N}_n)\leq\inf_{\qty{\tilde{\mathcal{N}}_n}_n}\qty{\liminf_{n\to\infty}\frac{1}{n}\log\qty(1+R_\mathrm{G}\qty(\tilde{\mathcal{N}}_n)):\lim_{n\to\infty}\frac{1}{2}\left\|J\qty(\tilde{\mathcal{N}}_n)-J\qty(\mathcal{N}_n)\right\|_1=0}.
    \end{align}
\end{lemma}

\begin{proof}
    Take any sequence $\qty{\tilde{\mathcal{N}}_n}_n$ satisfying
    \begin{align}
    \label{eq:J_continuity_bound}
        \lim_{n\to\infty}\frac{1}{2}\left\|J\qty(\tilde{\mathcal{N}}_n)-J\qty(\mathcal{N}_n)\right\|_1=0.
    \end{align}
    
Given any CQ channel $\mathcal{N}$,
due to the definition of $R_\mathrm{G}\qty(\mathcal{N})$,
there exists a free CQ channel $\mathcal{N}_\mathrm{free}$
and a CQ channel $\mathcal{N}'$ such that
\begin{align}
\frac{J(\mathcal{N})+R_\mathrm{G}\qty(\mathcal{N}) J(\mathcal{N}^\prime)}{1+R_\mathrm{G}\qty(\mathcal{N})}
=J(\mathcal{N}_\mathrm{free}).
\end{align}
Thus, we have
\begin{align}
(1+R_\mathrm{G}\qty(\mathcal{N}))J(\mathcal{N}_\mathrm{free})
-J(\mathcal{N})=
R_\mathrm{G}\qty(\mathcal{N}) J(\mathcal{N}^\prime)\ge 0,
\end{align}
which implies
\begin{align}
J(\mathcal{N}) \le
(1+R_\mathrm{G}\qty(\mathcal{N}))J(\mathcal{N}_\mathrm{free}).
\end{align}
For each $n$, applying this argument to $\mathcal{N}=\tilde{\mathcal{N}}_n$, 
we have a free CQ channel $\mathcal{N}_\mathrm{free}^{(n)}$ such that    \begin{align}
        J\qty(\tilde{\mathcal{N}}_n)\leq \qty(1+R_\mathrm{G}\qty(\tilde{\mathcal{N}}_n))J\qty(\mathcal{N}_\mathrm{free}^{(n)}).
    \end{align}
(For this derivation, also see Ref.~\cite[(55)--(58)]{Takagi2019a} and Ref.~\cite[Lemma~5]{doi:10.1142/S0219749909005298}.
Note that Ref.~\cite{Takagi2019a} considers Choi states for general channels, but the same argument holds for Choi states for CQ channels.)

    Then, due to the operator monotonicity of $\log$ 
    (see also Ref.~\cite[Sec.~A.4]{hayashi2016quantum} along with Ref.~\cite[Proposition~S17]{yamasaki2024generalized}),
    it holds that
    \begin{align}
    \label{eq:D_R_G}
        D\left(J\qty(\tilde{\mathcal{N}}_n)\middle\|J\qty(\mathcal{N}_\mathrm{free}^{(n)})\right)&\leq\log\qty(1+R_\mathrm{G}\qty(\tilde{\mathcal{N}}_n)).
    \end{align}
    Therefore, we have
    \begin{align}
        \liminf_{n\to\infty}\frac{1}{n}R_\mathrm{R}\qty(\mathcal{N}_n)
        &\stackrel{(a)}{=}\liminf_{n\to\infty}\frac{1}{n}R_\mathrm{R}\qty(\tilde{\mathcal{N}}_n)\notag\\
        &\stackrel{(b)}{\leq}\liminf_{n\to\infty}\frac{1}{n}D\left(J\qty(\tilde{\mathcal{N}}_n)\middle\|J\qty(\mathcal{N}_\mathrm{free}^{(n)})\right)\\
        \label{eq:lemma_S4_conclusion}
        &\stackrel{(c)}{\leq}\liminf_{n\to\infty}\frac{1}{n}\log\qty(1+R_\mathrm{G}\qty(\tilde{\mathcal{N}}_n)),
    \end{align}
    where $(a)$ follows from~\eqref{eq:J_continuity_bound} due to the asymptotic continuity~\cite[Lemma~7]{Winter2016} (note that $\min_{\mathcal{N}_\mathrm{free}\in\mathcal{F}} D\left(J\qty(\mathcal{N}_n)\middle\|J\qty(\mathcal{N}_\mathrm{free})\right)=O(n)$ due to the finite dimension and the existence of the full-rank Choi state $J(\mathcal{N}_\mathrm{full})^{\otimes n}$), 
$(b)$ follows from the definition~\eqref{eq:relative_entropy_of_resource_si} of $R_\mathrm{R}\qty(\tilde{\mathcal{N}}_n)$,
        and $(c)$ is~\eqref{eq:D_R_G}.
    Since~\eqref{eq:lemma_S4_conclusion} holds for any choice of $\qty{\tilde{\mathcal{N}}_n}_n$ satisfying~\eqref{eq:J_continuity_bound}, we obtain the conclusion.
\end{proof}

Conversely, we provide a lower bound of the relative entropy of resource in terms of the generalized robustness.

\begin{lemma}
\label{lem:converse_robustness}
For any CQ channel $\mathcal{N}$, there exists a sequence $\qty{\tilde{\mathcal{N}}_{n}}_n$ of CQ channels such that
\begin{align}
&R_\mathrm{R}^\infty(\mathcal{N}) \geq 
\limsup_{n\to\infty}
\frac{1}{n}
\log\qty(1+R_\mathrm{G}\qty(\tilde{\mathcal{N}}_{n})),\label{NMN1}\\
&\lim_{n\to \infty}\frac{1}{2}\left\| J\qty(\tilde{\mathcal{N}}_{n})-J(\mathcal{N}^{\otimes n}) \right\|_1
= 0.\label{NMN2}
\end{align}
\end{lemma}

\begin{proof}
\textit{Preparation}:
We choose any $R>R_\mathrm{R}^\infty\qty(\mathcal{N})$.
By definition of $R_\mathrm{R}^\infty$, there exist an integer $m>0$ and a free CQ channel $\mathcal{N}_\mathrm{free}^{(m)}$ on $m$-fold systems such that
\begin{align}
\frac{1}{m}
D\left( J\qty(\mathcal{N}^{\otimes m}) \middle\|  J\qty(\mathcal{N}_\mathrm{free}^{(m)}) \right)
<R. \Label{AS1}
\end{align}
For this fixed $m$, 
let $\mathcal{E}_{k}$ be a pinching map with respect to the state $J\qty(\mathcal{N}_\mathrm{free}^{(m)\otimes k})$; that is, for its spectral decomposition $J\qty(\mathcal{N}_\mathrm{free}^{(m)\otimes k})=\sum_{j=0}^{d_{k}-1} \lambda_j E_j$, we have $\mathcal{E}_{k}\qty(\sigma)=\sum_{j=0}^{d_{k}-1} E_j\sigma E_j$, where $d_{k}$ is the number of different eigenvalues of $J\qty(\mathcal{N}_\mathrm{free}^{(m)\otimes k})$.
Since $J\qty(\mathcal{N}_\mathrm{free}^{(m)\otimes k})$ is supported on the symmetric subspace under the permutation of $k$ fixed-dimensional systems (under the condition that $m$ is fixed), we have (see Ref.~\cite[Sec.~6.2]{hayashi2017group} and Ref.~\cite[Sec.~4.4]{hayashi2017grouprepresentation})
\begin{align}
\label{eq:d_k_bound}
    d_{k}=O(\mathrm{poly}(k))\quad\text{as $k\to\infty$}.
\end{align}
We define a projection
\begin{align}
P_{k}\coloneqq
\qty{\mathcal{E}_{k}\qty(J\qty(\mathcal{N}^{\otimes km}))  \geq e^{k m R} J\qty(\mathcal{N}_\mathrm{free}^{(m)\otimes k}) },\label{BNGH}
\end{align}
so that $J\qty(\mathcal{N}_\mathrm{free}^{(m)\otimes k})$ and $P_{k}$ should commute.

\textit{Construction of a sequence $\qty{\tilde{\mathcal{N}}_{n}}_n$ of CQ channels}:
To construct the sequence $\qty{\tilde{\mathcal{N}}_n}_n$ of channels, 
we take a subsequence $\{n_k\}_k$ of $\{1,2,\ldots\}$ as
\begin{align}
\label{eq:n_l}
    n_k\coloneqq km,
\end{align}
and define a subsequence $\{\rho_{n_k}\}_k$ of states as
\begin{align}
\label{eq:rho_n_k_si}
&\rho_{n_k}\coloneqq 
(\mathds{1}-P_{k}) J\qty(\mathcal{N}^{\otimes km})(\mathds{1}-P_{k})+
\qty(\qty(\frac{\mathds{1}}{d})^{\otimes km} - \Tr_{\mathrm{out}}\qty[(\mathds{1}-P_{k}) J\qty(\mathcal{N}^{\otimes km})(\mathds{1}-P_{k})])
\otimes \rho_\mathrm{full}^{\otimes km},
\end{align}
where $d$ is the dimension of the system for the input of $\mathcal{N}$, $\Tr_{\mathrm{out}}$ is the partial trace of the output system of $\mathcal{N}^{\otimes km}$, and $\rho_\mathrm{full}$ is the full-rank state output by the free CQ channel $\mathcal{N}_\mathrm{full}$ satisfying
\begin{align}
\label{eq:J_N_full}
    J\qty(\mathcal{N}_\mathrm{full})=\frac{\mathds{1}}{d}\otimes\rho_\mathrm{full}.
\end{align}
Since we have by construction
\begin{align}
    \Tr_{\mathrm{out}}\qty[\rho_{n_k}]=\qty(\frac{\mathds{1}}{d})^{\otimes km},
\end{align}
to see that $\rho_{n_k}$ is a Choi state of some quantum channel, 
it suffices to show
\begin{align}
\label{eq:positivity_choi_state}
    \rho_{n_k}\geq 0. 
\end{align}
To confirm~\eqref{eq:positivity_choi_state}, we use the fact that the Choi states of CQ channels are in the form of $J\qty(\mathcal{N}^{\otimes km})=\sum_j(1/d^{km})\ket{j}\bra{j}\otimes \rho_j$, and $P_k$ is also in the form of $\sum_j\ket{j}\bra{j}\otimes P_{k,j}$ for some projection $P_{k,j}$; using this fact, we obtain\footnote{It is challenging to generalize this inequality from CQ channels to a more general class of quantum channels with quantum inputs.}
\begin{align}
\Tr_\mathrm{out}\qty[(\mathds{1}-P_{k})J\qty(\mathcal{N}^{\otimes km})(\mathds{1}-P_{k})]=\sum_j\frac{\Tr[P_{k,j}\rho_j]}{d^{km}}\ket{j}\bra{j}
 \leq\qty(\frac{\mathds{1}}{d})^{\otimes km},
\end{align}
and hence~\eqref{eq:positivity_choi_state}.
Therefore, there exists a channel $\tilde{\mathcal{N}}_{n_k}$ such that
\begin{align}
    J\qty(\tilde{\mathcal{N}}_{n_k})=\rho_{n_k}.
\end{align}
With this subsequence $\qty{\tilde{\mathcal{N}}_{n_k}}_k$, we define the sequence $\qty{\tilde{\mathcal{N}}_{n}}_n$ as
\begin{align}
\label{eq:definition_tilde_n}
    \tilde{\mathcal{N}}_{n}\coloneqq\tilde{\mathcal{N}}_{n_k}\otimes\mathcal{N}^{\otimes (n-n_k)}
\end{align}
for each $n$ satisfying $n_k\leq n<n_{k+1}$.

\textit{Proof of \eqref{NMN1}}:
To show~\eqref{NMN1}, we will bound $R_\mathrm{G}\qty(\tilde{\mathcal{N}}_{n_k})$ in the following.
Due to~\eqref{eq:J_N_full},
we have
\begin{align}
\qty(\qty(\frac{\mathds{1}}{d})^{\otimes km} - \Tr_{\mathrm{out}}[(\mathds{1}-P_{k}) J\qty(\mathcal{N}^{\otimes km})(\mathds{1}-P_{k})])
\otimes \rho_\mathrm{full}^{\otimes km}
\le J\qty(\mathcal{N}_\mathrm{full}^{\otimes n_k}).\label{BN1}
\end{align}
Also, since the definition \eqref{BNGH} of $P_{k}$
guarantees
\begin{align}
\qty(\mathds{1}-P_{k})\qty( e^{n_k R} J\qty(\mathcal{N}_\mathrm{free}^{(m)\otimes k})-\mathcal{E}_k\qty(J\qty(\mathcal{N}^{\otimes n_k})))\geq 0,
\end{align}
we have
\begin{align}
& \qty(\mathds{1}-P_{k}) J\qty(\mathcal{N}^{\otimes n_k})\qty(\mathds{1}-P_{k}) \le  d_{n_k} \qty(\mathds{1}-P_{k}) {\cal E}_{k}
\qty(J\qty(\mathcal{N}^{\otimes n_k})) \qty(\mathds{1}-P_{k}) 
\le d_{n_k} e^{n_k R} J\qty(\mathcal{N}_\mathrm{free}^{(m)\otimes k}),\label{BN2}
\end{align}
where $d_{n_k}$ is in~\eqref{eq:d_k_bound}, and the first inequality follows from the pinching inequality $J\qty(\mathcal{N}^{\otimes n_k})\leq d_{n_k}{\cal E}_{k}
\qty(J\qty(\mathcal{N}^{\otimes n_k}))$~\cite[Lemma~3.10]{hayashi2016quantum}.
Combining~\eqref{BN1} and~\eqref{BN2}, we have
\begin{align}
J\qty(\tilde{\mathcal{N}}_{n_k})\leq\qty(1+d_{n_k} e^{n_k R}) J\qty(\frac{d_{n_k} e^{n_k R}N_\mathrm{free}^{(m)\otimes k}+\mathcal{N}_\mathrm{full}^{\otimes n_k}}{d_{n_k} e^{n_k R}+1}),
\label{eq:operator_inequality_free}
\end{align}
where the Choi state on the right-hand side is that of a free CQ channel.
Since
\begin{align}
\Tr_\mathrm{out}\qty[\frac{1}{d_{n_k} e^{n_k R}}\qty(
\qty(1+d_{n_k} e^{n_k R}) J\qty(\frac{d_{n_k} e^{n_k R}N_\mathrm{free}^{(m)\otimes k}+\mathcal{N}_\mathrm{full}^{\otimes n_k}}{d_{n_k} e^{n_k R}+1})
-J\qty(\tilde{\mathcal{N}}_{n_k}))]
&=\qty(\frac{\mathds{1}}{d})^{\otimes n_k},  \\
\frac{1}{d_{n_k} e^{n_k R}}\qty(
\qty(1+d_{n_k} e^{n_k R}) J\qty(\frac{d_{n_k} e^{n_k R}N_\mathrm{free}^{(m)\otimes k}+\mathcal{N}_\mathrm{full}^{\otimes n_k}}{d_{n_k} e^{n_k R}+1})
-J\qty(\tilde{\mathcal{N}}_{n_k}))
&\ge  0,
\end{align}
there exists a CQ channel $\tilde{\mathcal{N}}_{n_k}'$
such that
\begin{align}
    \frac{1}{d_{n_k} e^{n_k R}}\Big(
\qty(1+d_{n_k} e^{n_k R}) J\qty(\frac{d_{n_k} e^{n_k R}N_\mathrm{free}^{(m)\otimes k}+\mathcal{N}_\mathrm{full}^{\otimes n_k}}{d_{n_k} e^{n_k R}+1})
-J\qty(\tilde{\mathcal{N}}_{n_k})\Big)
=J\qty(\tilde{\mathcal{N}}_{n_k}').
\end{align}
Therefore, we obtain
\begin{align}
\label{eq:R_G_n_k_bound}
R_\mathrm{G}\qty(\tilde{\mathcal{N}}_{n_k})\leq d_{n_k} e^{n_k R}.
\end{align}
(For this derivation, see also 
Ref.~\cite[(55)--(58)]{Takagi2019a} and 
Ref.~\cite[Lemma~5]{doi:10.1142/S0219749909005298}. 
Note that Ref.~\cite{Takagi2019a} considers Choi states for general channels, but the same argument holds for Choi states for CQ channels.)

We also have
\begin{align}
\label{eq:R_G_n_n_k_bound}
    R_\mathrm{G}\qty(\mathcal{N}^{\otimes (n-n_k)})\leq R_\mathrm{G}\qty(\mathcal{N}^{\otimes m})=O\qty(1)\quad\text{as $n\to\infty$}.
\end{align}
By definition of $R_\mathrm{G}$, there exist CQ channels $\tilde{\mathcal{N}}_{n_k}^{\prime\prime}$ and $\mathcal{N}_{n-n_k}^{\prime\prime}$ such that
\begin{align}
    \frac{\tilde{\mathcal{N}}_{n_k}+R_\mathrm{G}\qty(\tilde{\mathcal{N}}_{n_k})\tilde{\mathcal{N}}_{n_k}^{\prime\prime}}{1+R_\mathrm{G}\qty(\tilde{\mathcal{N}}_{n_k})}&\in\mathcal{F},\\
    \frac{\mathcal{N}^{\otimes (n-n_k)}+R_\mathrm{G}\qty(\mathcal{N}^{\otimes (n-n_k)})\mathcal{N}_{n-n_k}^{\prime\prime}}{1+R_\mathrm{G}\qty(\mathcal{N}^{\otimes (n-n_k)})}&\in\mathcal{F},
\end{align}
and thus, it holds that
\begin{align}
\frac{\qty(\tilde{\mathcal{N}}_{n_k}+R_\mathrm{G}\qty(\tilde{\mathcal{N}}_{n_k})\tilde{\mathcal{N}}_{n_k}^{\prime\prime})\otimes\qty(\mathcal{N}^{\otimes (n-n_k)}+R_\mathrm{G}\qty(\mathcal{N}^{\otimes (n-n_k)})\mathcal{N}_{n-n_k}^{\prime\prime})}{1+ R_\mathrm{G}\qty(\tilde{\mathcal{N}}_{n_k})R_\mathrm{G}\qty(\mathcal{N}^{\otimes (n-n_k)})+R_\mathrm{G}\qty(\tilde{\mathcal{N}}_{n_k})+R_\mathrm{G}\qty(\mathcal{N}^{\otimes (n-n_k)}) }\in\mathcal{F}.
\end{align}
Hence, for $\tilde{\mathcal{N}}_{n}$ in~\eqref{eq:definition_tilde_n}, we have
\begin{align}
\label{eq:R_G_n_bound}
    R_\mathrm{G}\qty(\tilde{\mathcal{N}}_{n})
=    R_\mathrm{G}\qty(    \tilde{\mathcal{N}}_{n_k}\otimes\mathcal{N}^{\otimes (n-n_k)}) 
    \leq R_\mathrm{G}\qty(\tilde{\mathcal{N}}_{n_k})R_\mathrm{G}\qty(\mathcal{N}^{\otimes (n-n_k)})+R_\mathrm{G}\qty(\tilde{\mathcal{N}}_{n_k})+R_\mathrm{G}\qty(\mathcal{N}^{\otimes (n-n_k)}).
\end{align}
Since $\frac{1}{n_k}\log d_{n_k}\to 0$ due to~\eqref{eq:d_k_bound}, we obtain from~\eqref{eq:R_G_n_k_bound},~\eqref{eq:R_G_n_n_k_bound}, and~\eqref{eq:R_G_n_bound}
\begin{align}
\limsup_{n\to\infty}\frac{1}{n}\log\qty(1+R_\mathrm{G}\qty(\tilde{\mathcal{N}}_{n}))\leq R,
\label{eq:limsup_R_G}
\end{align}
which yields~\eqref{NMN1}.

\textit{Proof of \eqref{NMN2}}:
The strong converse part of our proof of generalized quantum Stein's lemma---specifically, by setting $\rho=J\qty(\mathcal{N})$ and $\sigma_n=J\qty(\mathcal{N}_\mathrm{free}^{(m)\otimes k})$ for $n=mk$ in Lemma~\ref{L5}---yields the following inequality for any $\epsilon\in[0,1)$:
\begin{align}
\label{eq:strong_converse_implication}
    \limsup_{k\to\infty}-\frac{1}{km}\log\min_{T}\qty{\Tr[TJ\qty(\mathcal{N}_\mathrm{free}^{(m)\otimes k})]:0\leq T\leq \mathds{1},\,\Tr[TJ\qty(\mathcal{N}^{\otimes km})]\geq1-\epsilon }\leq\frac{1}{m} D\left( J\qty(\mathcal{N}^{\otimes m}) \middle\|  J\qty(\mathcal{N}_\mathrm{free}^{(m)}) \right).
\end{align}
On the other hand, the definition~\eqref{BNGH} of $P_{k}$ implies for all $k$, 
\begin{align}
    \Tr[P_{k}\qty(e^{-k m R}\mathcal{E}_k\qty(J\qty(\mathcal{N}^{\otimes km}))  -  J\qty(\mathcal{N}_\mathrm{free}^{(m)\otimes k})) ]\geq 0,
\end{align}
which yields that
\begin{align}
\Tr[P_{k}J\qty(\mathcal{N}_\mathrm{free}^{(m)\otimes k}) ]
\leq e^{-k m R}\Tr[P_{k}\mathcal{E}_k\qty(J\qty(\mathcal{N}^{\otimes km}))]
\leq e^{-k m R}.
\end{align}
Thus, it follows from~\eqref{AS1} that
\begin{align}
\label{eq:p_l_m_bound}
\limsup_{k\to\infty}-\frac{1}{km}\log\Tr[P_{k}J\qty(\mathcal{N}_\mathrm{free}^{(m)\otimes k}) ]\geq R
>
\frac{1}{m} D\left( J\qty(\mathcal{N}^{\otimes m}) \middle\|  J\qty(\mathcal{N}_\mathrm{free}^{(m)}) \right).
\end{align}
Therefore, due to~\eqref{eq:strong_converse_implication} and~\eqref{eq:p_l_m_bound}, 
we find that the strong converse part of the generalized quantum Stein's lemma implies
\begin{align}
\lim_{k\to\infty}\Tr[P_{k}J\qty(\mathcal{N}^{\otimes km}) ]
= 0,\label{BH1}
\end{align}
which yields
\begin{align}
\lim_{k\to\infty}\frac{1}{2}\left\|\qty(\mathds{1}-P_{k}) J\qty(\mathcal{N}^{\otimes km})\qty(\mathds{1}-P_{k})- J\qty(\mathcal{N}^{\otimes km})\right\|_1 = 0.
\label{BH2}
\end{align}
Due to~\eqref{BH2}, we have
\begin{align}
    \lim_{k\to\infty}\left\|\qty(\frac{\mathds{1}}{d})^{\otimes km}-\Tr_{\mathrm{out}}[(\mathds{1}-P_{k}) J\qty(\mathcal{N}^{\otimes km})(\mathds{1}-P_{k})]\right\|_1= 0,
\end{align}
and thus
\begin{align}
\lim_{k\to\infty}\frac{1}{2}\left\|J\qty(\tilde{\mathcal{N}}_{n_k})-J\qty(\mathcal{N}^{\otimes n_k})\right\|_1 = 0.\label{BH3}
\end{align}
Thus, $\qty{\tilde{\mathcal{N}}_{n}}_n$ satisfies~\eqref{NMN2}.

\end{proof}

By combining the above two lemmas, we obtain a relation between the regularized relative entropy of resource and the generalized robustness.

\begin{corollary}
\label{cor:robustness_characterization}
For any CQ channel $\mathcal{N}$, we have
\begin{align}
R_\mathrm{R}^\infty(\mathcal{N}) &= 
\min_{\qty{\tilde{\mathcal{N}}_{n}}_n}\left\{
\lim_{n\to\infty}
\frac{1}{n}
\log\qty(1+R_\mathrm{G}\qty(\tilde{\mathcal{N}}_{n})):
\lim_{n\to \infty}\frac{1}{2}\left\| J\qty(\tilde{\mathcal{N}}_{n})-J(\mathcal{N}^{\otimes n}) \right\|_1
= 0,\right.\nonumber\\
&\quad\left.\text{the limit $\lim_{n\to\infty}\frac{1}{n}
\log\qty(1+R_\mathrm{G}\qty(\tilde{\mathcal{N}}_{n}))$ exists}
\right\},
\end{align}
where the minimum on the right-hand side exists.
\end{corollary}

\begin{proof}
Due to Lemma~\ref{lem:regularized_relative_entropy_generalized_robustness}, 
it holds that
\begin{align}
R_\mathrm{R}^\infty\qty(\mathcal{N})&=\liminf_{n\to\infty}\frac{1}{n}R_\mathrm{R}(\mathcal{N}^{\otimes n})\notag\\
&\leq
\inf_{\qty{\tilde{\mathcal{N}}_{n}}_n}\left\{
\liminf_{n\to\infty}
\frac{1}{n}
\log\qty(1+R_\mathrm{G}\qty(\tilde{\mathcal{N}}_{n})):
\lim_{n\to \infty}\frac{1}{2}\left\| J\qty(\tilde{\mathcal{N}}_{n})-J(\mathcal{N}^{\otimes n}) \right\|_1
= 0
\right\}\notag\\
&\leq\inf_{\qty{\tilde{\mathcal{N}}_{n}}_n}\left\{
\lim_{n\to\infty}
\frac{1}{n}
\log\qty(1+R_\mathrm{G}\qty(\tilde{\mathcal{N}}_{n})):
\lim_{n\to \infty}\frac{1}{2}\left\| J\qty(\tilde{\mathcal{N}}_{n})-J(\mathcal{N}^{\otimes n}) \right\|_1
= 0,\right.\nonumber\\
&\quad\left.\text{the limit $\lim_{n\to\infty}\frac{1}{n}
\log\qty(1+R_\mathrm{G}\qty(\tilde{\mathcal{N}}_{n}))$ exists}
\right\},
\label{eq:limit_liminf_R_R}
\end{align}
where the last inequality is due to adding the constraint that the limit $\lim_{n\to\infty}\frac{1}{n}
\log\qty(1+R_\mathrm{G}\qty(\tilde{\mathcal{N}}_{n}))$ should exist.
Lemma~\ref{lem:converse_robustness} shows that there exists a sequence $\qty{\tilde{\mathcal{N}}_n}_n$ achieving
\begin{align}
& 
\limsup_{n\to\infty}
\frac{1}{n}
\log\qty(1+R_\mathrm{G}\qty(\tilde{\mathcal{N}}_{n}))\leq R_\mathrm{R}^\infty(\mathcal{N}),\\
&\lim_{n\to \infty}\frac{1}{2}\left\| J\qty(\tilde{\mathcal{N}}_{n})-J(\mathcal{N}^{\otimes n}) \right\|_1
= 0.
\end{align}
Thus, for this sequence $\qty{\tilde{\mathcal{N}}_n}_n$, we see that the limit
\begin{align}
\label{eq:limit_existence}
\lim_{n\to\infty} \frac{1}{n} \log\qty(1+R_\mathrm{G}\qty(\tilde{\mathcal{N}}_{n}))   
=\liminf_{n\to\infty} \frac{1}{n} \log\qty(1+R_\mathrm{G}\qty(\tilde{\mathcal{N}}_{n}))   
=\limsup_{n\to\infty} \frac{1}{n} \log\qty(1+R_\mathrm{G}\qty(\tilde{\mathcal{N}}_{n}))=R_\mathrm{R}^\infty(\mathcal{N})
\end{align}
exists.
Due to the existence of $\qty{\tilde{\mathcal{N}}_n}_n$ achieving~\eqref{eq:limit_existence}, it follows from~\eqref{eq:limit_liminf_R_R} that
\begin{align}
R_\mathrm{R}^\infty(\mathcal{N})&=\min_{\qty{\tilde{\mathcal{N}}_{n}}_n}\left\{
\liminf_{n\to\infty}
\frac{1}{n}
\log\qty(1+R_\mathrm{G}\qty(\tilde{\mathcal{N}}_{n})):
\lim_{n\to \infty}\frac{1}{2}\left\| J\qty(\tilde{\mathcal{N}}_{n})-J(\mathcal{N}^{\otimes n}) \right\|_1
= 0
\right\}\\
&=\min_{\qty{\tilde{\mathcal{N}}_{n}}_n}\left\{
\lim_{n\to\infty}
\frac{1}{n}
\log\qty(1+R_\mathrm{G}\qty(\tilde{\mathcal{N}}_{n})):
\lim_{n\to \infty}\frac{1}{2}\left\| J\qty(\tilde{\mathcal{N}}_{n})-J(\mathcal{N}^{\otimes n}) \right\|_1
= 0,\right.\nonumber\\
&\quad\left.\text{the limit $\lim_{n\to\infty}\frac{1}{n}
\log\qty(1+R_\mathrm{G}\qty(\tilde{\mathcal{N}}_{n}))$ exists}
\right\},
\end{align}
where the minima exist.
\end{proof}

\subsection{The direct part of the second law of QRTs for CQ channels}
\label{sec:direct_second_law}

Using the relation shown in the previous section, we show the direct part of the second law of QRTs for CQ channels.

\begin{proposition}[The direct part of the second law of QRTs for CQ channels]
\label{prp:direct_asymptotic_conversion}
For any CQ channels $\mathcal{N}_1$ and $\mathcal{N}_2$ satisfying
\begin{align}
\label{eq:R_j_condition}
    R_\mathrm{R}^\infty\qty(\mathcal{N}_j)&>0\quad(j\in\{1,2\}),
\end{align}
it holds that
\begin{align}
\label{eq:achievability_bound}
    r_{\tilde{\mathcal{O}}}\qty(\mathcal{N}_1\to\mathcal{N}_2)\geq\frac{R_\mathrm{R}^\infty\qty(\mathcal{N}_1)}{R_\mathrm{R}^\infty\qty(\mathcal{N}_2)}.
\end{align}
\end{proposition}

\begin{proof}
    To prove~\eqref{eq:achievability_bound},
    due to $R_\mathrm{R}^\infty\qty(\mathcal{N}_j)>0$ in~\eqref{eq:R_j_condition}, we choose any
    \begin{align}
    \label{eq:delta}
        \delta\in\qty(0, \min\qty{R_\mathrm{R}^\infty\qty(\mathcal{N}_1),R_\mathrm{R}^\infty\qty(\mathcal{N}_2)}),
    \end{align}
    and we will construct a sequence $\{\Theta_n\}_n\in\tilde{O}$ of operations (superchannels) achieving the asymptotic conversion from $\mathcal{N}_1$ to $\mathcal{N}_2$ at the rate
    \begin{align}
    \label{eq:rate}
        r\coloneqq\frac{R_\mathrm{R}^\infty\qty(\mathcal{N}_1)-\delta}{R_\mathrm{R}^\infty\qty(\mathcal{N}_2)}.
    \end{align}
    In the following, we will first present the construction of $\{\Theta_n\}_n$, then prove that $\{\Theta_n\}_n$ satisfies the conditions for $\tilde{\mathcal{O}}$, and finally show that $\{\Theta_n\}_n$ achieves the asymptotic conversion at rate $r$.

    \textit{Construction of $\{\Theta_n\}_n$}:
    Applying the generalized quantum Stein's lemma to the state $J\qty(\mathcal{N}_1)^{\otimes n}=J\qty(\mathcal{N}_1^{\otimes n})$ and the set $\qty{J\qty(\mathcal{N}_\mathrm{free}):\mathcal{N}_\mathrm{free}\in\mathcal{F}}$ of states, we see that there exist a sequence $\{\epsilon_n\in(0,1)\}_n$ of type I error parameters such that
    \begin{align}
    \label{eq:epsilon_scalng}
        \lim_{n\to\infty}\epsilon_n=0,
    \end{align}
    and a sequence $\{T_n,\mathds{1}-T_n\}$ of POVMs such that for sufficiently large $n$
    \begin{align}
    \label{eq:epsilon_n}
        \Tr[\qty(\mathds{1}-T_n) J\qty(\mathcal{N}_1^{\otimes n})]&\leq\epsilon_n,\\
    \label{eq:beta_n}
        \max_{\mathcal{N}_\mathrm{free}\in\mathcal{F}}\qty{\Tr[T_n J\qty(\mathcal{N}_\mathrm{free})]}&\leq \exp\qty[-n\qty(R_\mathrm{R}^\infty(\mathcal{N}_1)-\frac{\delta}{3})].
    \end{align}
    Due to Corollary~\ref{cor:robustness_characterization},
    we have a sequence $\qty{\mathcal{N}_2^{(rn)}}_n$ of channels satisfying
    \begin{align}
    \label{eq:robustness_scalng}
    &R_\mathrm{R}^\infty\qty(\mathcal{N}_2)=\lim_{n\to\infty}\frac{\log\qty(1+R_\mathrm{G}\qty(\mathcal{N}_2^{(rn)}))}{\left\lceil rn\right\rceil},\\
    \label{eq:distance_scalng}
    &\lim_{n\to\infty}\frac{1}{2}\left\|J\qty(\mathcal{N}_2^{(rn)})-J\qty(\mathcal{N}_2^{\otimes\lceil rn\rceil})\right\|_1=0.
    \end{align}
    Let $\mathcal{N}_2^{(rn)\prime}$ denote an optimal channel in the minimization of the definition of $R_\mathrm{G}\qty(\mathcal{N}_2^{(rn)})$, i.e.,
    \begin{align}
    \label{eq:robustness_bound}
        \frac{\mathcal{N}_2^{(rn)}+R_\mathrm{G}\qty(\mathcal{N}_2^{(rn)})\mathcal{N}_2^{(rn)\prime}}{1+R_\mathrm{G}\qty(\mathcal{N}_2^{(rn)})}\in\mathcal{F}.
    \end{align}
    Then, we define $\Theta_n$ as
    \begin{align}
    \label{eq:Theta_n_definition}
        \Theta_n\qty(\mathcal{N})\coloneqq
            \Tr[T_{n}J\qty(\mathcal{N})]\mathcal{N}_2^{(rn)}+\Tr[\qty(\mathds{1}-T_{n})J\qty(\mathcal{N})]\mathcal{N}_2^{(rn)\prime}.
    \end{align}
    The map $\Theta_n$ is a superchannel that can be implemented using pre- and post-processing with the following physical meaning: given any input CQ channel $\mathcal{N}$ on the $n$-fold systems, we first input the maximally entangled state to $\mathcal{N}$ (or equivalently, the maximally correlated classical (diagonal) state since $\mathcal{N}$ is a CQ channel) to prepare its Choi state $J\qty(\mathcal{N})$, and then perform a POVM measurement $\qty{T_n,\mathds{1}-T_n}$ on $J\qty(\mathcal{N})$;\footnote{In the case of state conversion, i.e., when the input dimension of the CQ channels is one, this is equivalent to directly measuring the given state by $\qty{T_n,\mathds{1}-T_n}$.} then, depending on the measurement outcome, we use $\mathcal{N}_2^{(rn)}$ or $\mathcal{N}_2^{(rn)\prime}$ as the output CQ channel of $\Theta_n$, which will turn out to satisfy the asymptotically resource-non-generating property and the asymptotic continuity, as shown below.
    
    \textit{Conditions for $\{\Theta_n\}_n\in\tilde{\mathcal{O}}$}:
    The asymptotic continuity of $\{\Theta_n\}_n$ in~\eqref{eq:Theta_n_definition} is obvious from the fact that $J\qty(\mathcal{N})\mapsto J\qty(\Theta_n\qty(\mathcal{N}))$ is a CPTP map of Choi states.
    Thus, we will show that $\{\Theta_{n}\}_n$ has the asymptotically resource-non-generating property; that is, in the following, we will bound $R_\mathrm{G}\qty(\Theta_{n}\qty(\mathcal{N}_\mathrm{free}^{\qty(n)}))$ for any free channel $\mathcal{N}_\mathrm{free}^{\qty(n)}$.
    
    For simplicity of notation, we define
    \begin{align}
        s_n&\coloneqq R_\mathrm{G}\qty(\mathcal{N}_2^{(rn)}),\\
        t_n&\coloneqq \Tr[T_nJ\qty(\mathcal{N}_\mathrm{free}^{(n)})].
    \end{align}
    Due to~\eqref{eq:robustness_scalng}, for sufficiently large $n$, $s_{n}$ satisfies
    \begin{align}
    \label{eq:s_n_bound}
        \exp[nr\qty(R_\mathrm{R}^\infty\qty(\mathcal{N}_2)+\frac{\delta}{3r})]\geq 1+s_{n}\geq\exp[nr\qty(R_\mathrm{R}^\infty\qty(\mathcal{N}_2)-\frac{\delta}{3r})]\to\infty\quad\text{as $n\to\infty$}.
    \end{align}
    Due to~\eqref{eq:beta_n}, for sufficiently large $n$, $t_n$ satisfies
    \begin{align}
    \label{eq:t_n_bound}
        t_{n}&\leq\max_{\mathcal{N}_\mathrm{free}\in\mathcal{F}}\Tr[T_{n}J\qty(\mathcal{N}_\mathrm{free})]\leq\exp[-{n}\qty(R_\mathrm{R}^\infty\qty(\mathcal{N}_1)-\frac{\delta}{3})]\to 0\quad\text{as $n\to\infty$}.
    \end{align}
    Then, for sufficiently large $n$, it follows from~\eqref{eq:rate},~\eqref{eq:s_n_bound}, and~\eqref{eq:t_n_bound} that
    \begin{align}
    \label{eq:requirement}
        \frac{1}{1+s_{n}}-t_{n}\geq\exp[-n\qty(R_\mathrm{R}^\infty\qty(\mathcal{N}_1)-\frac{2\delta}{3})]-\exp[-n\qty(R_\mathrm{R}^\infty\qty(\mathcal{N}_1)-\frac{\delta}{3})]\geq 0,
    \end{align}
    where the first inequality is obtained from
    \begin{align}
        \frac{1}{1+s_{n}}&\geq\frac{1}{\exp[nr\qty(R_\mathrm{R}^\infty\qty(\mathcal{N}_2)+\frac{\delta}{3r})]}\notag\\
        &=\frac{1}{\exp[n\qty(rR_\mathrm{R}^\infty\qty(\mathcal{N}_2)+\frac{\delta}{3})]}\notag\\
        &=\frac{1}{\exp[n\qty(R_\mathrm{R}^\infty\qty(\mathcal{N}_1)-\frac{2\delta}{3})]}.
    \end{align}

    Using $s_n$ and $t_n$, we will bound $R_\mathrm{G}\qty(\Theta_{n}\qty(\mathcal{N}_\mathrm{free}^{\qty(n)}))$.
    On one hand, by definition~\eqref{eq:Theta_n_definition} of $\Theta_{n}$, we have
    \begin{align}
    \label{eq:t}
        \Theta_{n}\qty(\mathcal{N}_\mathrm{free}^{\qty(n)})=t_{n}\mathcal{N}_2^{\qty(rn)}+\qty(1-t_{n})\mathcal{N}_2^{\qty(rn)\prime}.
    \end{align}
    On the other hand, it follows from~\eqref{eq:robustness_bound} that
    \begin{align}
    \label{eq:s}
        \frac{1}{1+s_{n}}\mathcal{N}_2^{\qty(rn)}+\frac{s_{n}}{1+s_{n}}\mathcal{N}_2^{\qty(rn)\prime}\in\mathcal{F}.
    \end{align}
    Thus, for sufficiently large $n$ satisfying~\eqref{eq:requirement},
    it holds that
    \begin{align}
        R_\mathrm{G}\qty(\Theta_{n}\qty(\mathcal{N}_\mathrm{free}^{\qty(n)}))&\stackrel{(a)}{\leq}\min\qty{s\geq 0:\frac{\Theta_{n}\qty(\mathcal{N}_\mathrm{free}^{\qty(n)})+s\mathcal{N}_2^{\qty(rn)}}{1+s}\in\mathcal{F}\notag}\\
        \label{eq:S_101}
        &\stackrel{(b)}{\leq}\frac{\frac{1}{1+s_{n}}-t_{n}}{\frac{s_{n}}{1+s_{n}}}\notag\\
        &\stackrel{(c)}{\to} 0\quad\text{as $n\to\infty$},
   \end{align}
    where $(a)$ holds by taking $\mathcal{N}_2^{\qty(rn)}$ as the channel used in the minimization in the definition of $R_\mathrm{G}\qty(\Theta_{n}\qty(\mathcal{N}_\mathrm{free}^{\qty(n)}))$, and $(b)$ follows from~\eqref{eq:t} and~\eqref{eq:s} since $s=\frac{\frac{1}{1+s_{n}}-t_{n}}{\frac{s_{n}}{1+s_{n}}}\geq 0$ is the solution of
    \begin{align}
        \frac{\qty(t_{n}\mathcal{N}_2^{\qty(rn)}+(1-t_{n})\mathcal{N}_2^{\qty(rn)\prime})+s\mathcal{N}_2^{\qty(rn)}}{1+s}=\frac{1}{1+s_{n}}\mathcal{N}_2^{\qty(rn)}+\frac{s_{n}}{1+s_{n}}\mathcal{N}_2^{\qty(rn)\prime}\in\mathcal{F},
    \end{align}
    and $(c)$ is obtained from~\eqref{eq:s_n_bound} and~\eqref{eq:t_n_bound}.
    Therefore, we have shown that $\Theta_{n}$ has the asymptotically resource-non-generating property.

    \textit{Achievability of asymptotic conversion at rate $r$ by $\{\Theta_n\}_n$}:
    We will show that the sequence $\qty{\Theta_{n}}_n$ achieves the asymptotic conversion from $\mathcal{N}_1$ to $\mathcal{N}_2$ at the rate $r$.
    It follows from~\eqref{eq:Theta_n_definition} that
    \begin{align}
        \Theta_{n}\qty(\mathcal{N}_1^{\otimes n})=\Tr[T_{n}J\qty(\mathcal{N}_1^{\otimes n})]\mathcal{N}_2^{\qty(rn)}+\Tr[\qty(\mathds{1}-T_{n})J\qty(\mathcal{N}_1^{\otimes n})]\mathcal{N}_2^{\qty(rn)\prime}.
    \end{align}
    Therefore, we have
    \begin{align}
        &\frac{1}{2}\left\|J\qty(\Theta_{n}\qty(\mathcal{N}_1^{\otimes n}))-J\qty(\mathcal{N}_2^{\otimes\lceil rn\rceil})\right\|_1\nonumber\\
        &\stackrel{(a)}{\leq}\frac{\Tr[T_{n} J\qty(\mathcal{N}_1^{\otimes n})]}{2}\left\|J\qty(\mathcal{N}_2^{\qty(rn)})-J\qty(\mathcal{N}_2^{\otimes\lceil rn\rceil})\right\|_1+\frac{\Tr[\qty(\mathds{1}-T_{n}) J\qty(\mathcal{N}_1^{\otimes n})]}{2}\left(\left\|J\qty(\mathcal{N}_2^{\otimes\lceil rn\rceil})\right\|_1+\left\|J\qty(\mathcal{N}_2^{(rn)\prime})\right\|_1\right)\notag\\
        &\stackrel{(b)}{\leq}\frac{1-\epsilon_{n}}{2}\left\|J\qty(\mathcal{N}_2^{(rn)})-J\qty(\mathcal{N}_2^{\otimes\lceil rn\rceil})\right\|_1+2\epsilon_{n}\notag\\
        &\stackrel{(c)}{\to} 0\quad\text{as $n\to\infty$},
    \end{align}
    where $(a)$ is the subadditivity and the homogeneity of the trace norm, $(b)$ follows from $\frac{1}{2}\left\|J\qty(\mathcal{N}_2^{\otimes\lceil rn\rceil})\right\|_1\leq 1$, $\frac{1}{2}\left\|J\qty(\mathcal{N}_2^{(rn)\prime})\right\|_1\leq 1$, and~\eqref{eq:epsilon_n}, and $(c)$ from~\eqref{eq:epsilon_scalng} and~\eqref{eq:distance_scalng}.
\end{proof}

\subsection{The converse part of the second law of QRTs for CQ channels}
\label{sec:converse_second_law}

In this section, we prove the converse part of the second law of QRTs for CQ channels.

We first show an asymptotic version of the monotonicity of the relative entropy of resource for CQ channels as follows.

\begin{lemma}
\label{lem:monotonicity}
Let $\{\Theta_n\}_n\in\tilde{\mathcal{O}}$ be any sequence of asymptotically free operations.
Then, for any CQ channel $\mathcal{N}$, we have
\begin{align}
R_\mathrm{R}^\infty\qty(\mathcal{N}) \geq 
\liminf_{n\to \infty}\frac{1}{n}R_\mathrm{R}\qty(\Theta_n\qty(\mathcal{N}^{\otimes n})).
\label{NMN6}
\end{align}
\end{lemma}

\begin{proof}
Due to Lemma~\ref{lem:regularized_relative_entropy_generalized_robustness},
we have
\begin{align}
\liminf_{n\to \infty}\frac{1}{n}R_\mathrm{R}\qty(\Theta_n\qty(\mathcal{N}^{\otimes n}))
\leq
\inf_{\{\tilde{\mathcal{N}}_n\}}\qty{
\liminf_{n\to\infty}
\frac{1}{n}
\log (1+R_\mathrm{G}\qty(\tilde{\mathcal{N}}_n)):
\lim_{n\to\infty}\frac{1}{2}\left\| J\qty(\tilde{\mathcal{N}}_n)-J\qty(\Theta_n \qty(\mathcal{N}^{\otimes n})) \right\|_1
= 0
}.
\end{align}
Thus, it suffices to show the following inequality
\begin{align}
\inf_{\qty{\tilde{\mathcal{N}}_n}}\qty{
\liminf_{n\to\infty}
\frac{1}{n}
\log (1+R_\mathrm{G}\qty(\tilde{\mathcal{N}}_n)):
\lim_{n\to\infty}\frac{1}{2}\left\| J\qty(\tilde{\mathcal{N}}_n)-J\qty(\Theta_n \qty(\mathcal{N}^{\otimes n})) \right\|_1
= 0
}\leq R_\mathrm{R}^\infty\qty(\mathcal{N}).
\label{NM3}
\end{align}

To show~\eqref{NM3}, due to Corollary~\ref{cor:robustness_characterization}, we choose a sequence $\qty{\tilde{\mathcal{N}}_n}_n$ of CQ channels satisfying
\begin{align}
\label{eq:limit_R_R_infty_tilde_N}
&\lim_{n\to\infty} \frac{1}{n} \log (1+R_\mathrm{G}\qty(\tilde{\mathcal{N}}_n))=R_\mathrm{R}^\infty\qty(\mathcal{N}),\\
&\lim_{n\to\infty}\frac{1}{2}\left\| J\qty(\tilde{\mathcal{N}}_n)-J\qty(\mathcal{N}^{\otimes n}) \right\|_1
\to 0.
\end{align}
Then, due to the asymptotic continuity~\eqref{eq:condition_asymptotic_continuity} of $\qty{\Theta_n}_n$, we have\footnote{We require the asymptotic continuity for this step while this requirement is unnecessary in QRTs for states, i.e., in the case where the dimension of inputs for CQ channels is one.}
\begin{align}
\lim_{n\to\infty}\frac{1}{2}\left\| J\qty(\Theta_n \qty(\tilde{\mathcal{N}}_n))-J\qty(\Theta_n \qty(\mathcal{N}^{\otimes n})) \right\|_1
= 0.\label{K1}
\end{align}
Set
\begin{align}
\label{eq:r_n}
    r_n\coloneqq R_\mathrm{G}\qty(\tilde{\mathcal{N}}_n).
\end{align}
There exists a CQ channel $\tilde{\mathcal{N}}_n'$ such that $
\frac{\tilde{\mathcal{N}}_n+r_n\tilde{\mathcal{N}}_n'}{1+r_n}$ is a free CQ channel.
The asymptotically resource-non-generating property of $\Theta_n$ implies
\begin{align}
\label{eq:r_n_prime_limit}
r_n'\coloneqq R_\mathrm{G}\qty(\Theta_n \qty(\frac{\tilde{\mathcal{N}}_n+r_n\tilde{\mathcal{N}}_n'}{1+r_n}))\to 0\quad\text{as $n\to\infty$}.
\end{align}
There exists a CQ channel $\tilde{\mathcal{N}}_n''$
such that $
\frac{\Theta_n \qty(\frac{\tilde{\mathcal{N}}_n+r_n\tilde{\mathcal{N}}_n'}{1+r_n})
+r_n'\tilde{\mathcal{N}}_n''}{1+r_n'}$ is a free CQ channel.
That is, we have a free CQ channel
\begin{align}
\frac{\Theta_n \qty(\frac{\tilde{\mathcal{N}}_n+r_n\tilde{\mathcal{N}}_n'}{1+r_n})
+r_n'\tilde{\mathcal{N}}_n''}{1+r_n'}
=\frac{\Theta_n \qty(\tilde{\mathcal{N}}_n)+r_n\Theta_n\qty(\tilde{\mathcal{N}}_n')
+r_n'\qty(1+r_n)\tilde{\mathcal{N}}_n''}{(1+r_n)(1+r_n')}.
\end{align}
Thus, we have
\begin{align}
R_\mathrm{G}\qty(\Theta_n\qty(\tilde{\mathcal{N}}_n))\leq
\qty(1+r_n)\qty(1+r_n')-1.
\end{align}
Therefore, due to~\eqref{eq:limit_R_R_infty_tilde_N},~\eqref{eq:r_n} and~\eqref{eq:r_n_prime_limit}, we have
\begin{align}
\liminf_{n\to\infty}
\frac{1}{n}
\log (1+R_\mathrm{G}\qty(\Theta_n\qty(\tilde{\mathcal{N}}_n)))
&\le
\lim_{n\to\infty}
\frac{1}{n}
\log (1+R_\mathrm{G}\qty(\tilde{\mathcal{N}}_n))=R_\mathrm{R}^\infty\qty(\mathcal{N}).\label{K2}
\end{align}
Since~\eqref{K1} guarantees that the sequence $\qty{\Theta_n\qty(\tilde{\mathcal{N}}_n)}_n$ satisfies the condition of $\qty{\tilde{\mathcal{N}}_n}_n$ on the left-hand side of~\eqref{NM3},~\eqref{K2} implies~\eqref{NM3} and thus the conclusion.
\end{proof}

Using this lemma, we show the converse part of the second law of QRTs for CQ channels.

\begin{proposition}[The converse part of the second law of QRTs for CQ channels]
\label{prp:converse}
For any CQ channels $\mathcal{N}_1$ and $\mathcal{N}_2$,
it holds that
\begin{align}
    r_{\tilde{\mathcal{O}}}\qty(\mathcal{N}_1\to\mathcal{N}_2)\leq\frac{R_\mathrm{R}^\infty\qty(\mathcal{N}_1)}{R_\mathrm{R}^\infty\qty(\mathcal{N}_2)},
\end{align}
where if $R_\mathrm{R}^\infty\qty(\mathcal{N}_2)=0$, the right-hand side should be understood as a trivial upper bound $\infty$.
\end{proposition}

\begin{proof}
    Take any achievable rate $r\leq r_{\tilde{\mathcal{O}}}\qty(\mathcal{N}_1\to\mathcal{N}_2)$.
    We have a sequence $\{\Theta_n\}_n$ of operations in $\tilde{\mathcal{O}}$ achieving
    \begin{align}
    \label{eq:achieve}
        \liminf_{n\to\infty}\frac{1}{2}\left\|J\qty(\Theta_n\qty(\mathcal{N}_1^{\otimes n}))-J\qty(\mathcal{N}_2^{\otimes\lceil rn\rceil})\right\|_1=0.
    \end{align}
    Then, due to Lemma~\ref{lem:monotonicity}, we obtain
    \begin{align}
    R_\mathrm{R}^\infty\qty(\mathcal{N}_1)
    &\geq \liminf_{n\to \infty}\frac{1}{n}R_\mathrm{R}\qty(\Theta_n\qty(\mathcal{N}_1^{\otimes n}))\notag\\
    &\stackrel{(a)}{=} \lim_{n\to \infty}\frac{1}{n}R_\mathrm{R}\qty(\mathcal{N}_2^{\otimes \lceil rn\rceil})\notag\\
    &\stackrel{(b)}{=}rR_\mathrm{R}^\infty\qty(\mathcal{N}_2),
    \end{align}
    where $(a)$ follows from the asymptotic continuity~\cite[Lemma~7]{Winter2016} (note that $\min_{\mathcal{N}_\mathrm{free}\in\mathcal{F}}D\left(J\qty(\mathcal{N}^{\otimes n})\middle\|J\qty(\mathcal{N}_\mathrm{free})\right)=O(n)$ due to the finite dimension and the existence of the full-rank Choi state $J(\mathcal{N}_\mathrm{full})$), and $(b)$ from the weak additivity of the regularized function $R_\mathrm{R}^\infty$.
    Thus, under the condition $R_\mathrm{R}^\infty\qty(\mathcal{N}_2)>0$ for avoiding the trivial bound $r_{\tilde{\mathcal{O}}}\qty(\mathcal{N}_1\to\mathcal{N}_2)\leq\infty$, the asymptotic conversion rate is bounded by
    \begin{align}
        r_{\tilde{\mathcal{O}}}\qty(\mathcal{N}_1\to\mathcal{N}_2)\leq\frac{R_\mathrm{R}^\infty\qty(\mathcal{N}_1)}{R_\mathrm{R}^\infty\qty(\mathcal{N}_2)}.
    \end{align}
\end{proof}

\section{Examples}
\label{sec:examples}

In this section, we present examples of QRTs for both static and dynamical resources that fall within the scope of our results, as well as counterexamples that emerge when certain assumptions in our analysis are relaxed.
In Sec.~\ref{sec:examples_state}, we discuss examples of QRTs for static resources, while Sec.~\ref{sec:examples_channels} covers examples for dynamical resources.
Finally, in Sec.~\ref{sec:counterexample}, we present counterexamples relevant to our analysis of the generalized quantum Stein's lemma.

\subsection{Example of QRTs for states}
\label{sec:examples_state}
Our results on the second law of QRTs for CQ channels apply to QRTs for states by considering the special case where the dimension of the input systems of the CQ channels is one.
As presented in the main text, the sets of free states satisfying Properties~\ref{p1:si},~\ref{p3:si}, and~\ref{p4:methods} of our results on the second law include representative classes of finite-dimensional convex QRTs, such as those for entanglement, magic, asymmetry, and coherence; for further details, see also Ref.~\cite[Sec.~IV~A]{Chitambar2018}.
We here provide applications of our results.

The entanglement theory considers ${\cal H}$ given by a bipartite system, i.e., ${\cal H}={\cal H}_A\otimes {\cal H}_B$, and free states are the set ${\cal F}_\mathrm{SEP}$ of separable states between $A$ and $B$, satisfying the above properties.
A central task in the entanglement theory is entanglement distillation~\cite{PhysRevA.53.2046}, but as discussed in the main text, the distillable entanglement under local operations and classical communication (LOCC) is hard to characterize in general.
However, LOCC is not the only possible set of free operations that can lead to ${\cal F}_\mathrm{SEP}$ as its free states; indeed, previous work has provided characterizations of distillable entanglement under various classes of free operations distinct from LOCC\@.
For example, under one-way LOCC, the distillable entanglement is characterized by coherent information~\cite{devetak2005distillation}.
Also, the Piani relative entropy of entanglement, named after Ref.~\cite{PhysRevLett.103.160504}, characterizes the distillable entanglement under dually non-entangling operations~\cite{lami2024distillable}.
Applying our results, as originally envisioned in Refs.~\cite{Brand_o_2008,brandao2010reversible}, we have another rare example of such a characterization: the distillable entanglement under the asymptotically resource-non-generating operations is equal to the regularized relative entropy of entanglement defined with respect to ${\cal F}_\mathrm{SEP}$.

Indeed, the scope of applications of our results extends even beyond the entanglement theory.
Another important example is the QRT of magic for qubits, i.e., $\mathcal{H}=\mathbb{C}^2$.
In this setting, a canonical choice of free operations---analogous to LOCC in the entanglement theory---is stabilizer operations, which are composed of the preparation of an eigenstate of Pauli operators, the measurement in the Pauli basis, Clifford gates, and classical feedforward~\cite{Veitch2014,Howard2017}.
Correspondingly, the set $\mathcal{F}_\mathrm{STAB}$ of free states is the convex hull of the set of stabilizer states of qubits.
In this case, in place of entanglement distillation, a task of central interest is magic state distillation~\cite{knill2004faulttolerantpostselectedquantumcomputation,knill2005quantum,PhysRevA.71.022316,PhysRevA.86.052329}, which aims at asymptotic state conversion from many copies of a given mixed state $\rho$ into as many as possible copies of magic states, e.g., $\ket{T}=\frac{1}{\sqrt{2}}\qty(\ket{0}+e^{i\pi/4}\ket{1})$, within a vanishingly small target error $\epsilon$ under the stabilizer operations.
The inverse of the corresponding achievable rate is also referred to as the overhead of magic state distillation.
Unlike the information-theoretic nature of entanglement distillation, magic state distillation is primarily motivated by applications to fault-tolerant quantum computation~\cite{knill2004faulttolerantpostselectedquantumcomputation,knill2005quantum,PhysRevA.71.022316,PhysRevA.86.052329}, where the classical feedforward must be computationally efficient.
In this setting, if one uses a conventional protocol for magic state distillation, the achievable conversion rate would vanish at the inverse-polylogarithmic speed as the target error $\epsilon$ becomes small~\cite{PhysRevA.86.052329}.
By contrast, a more recent study~\cite{wills2024constantoverheadmagicstatedistillation} has proposed and analyzed an alternative protocol, proving that the achievable rate of magic state distillation under stabilizer operations can remain strictly positive, even under the constraint on efficient classical feedforward.
Nonetheless, as in the case of entanglement theory, a full characterization of the optimal rate of magic state distillation remains an open and challenging problem in general~\cite{wills2024constantoverheadmagicstatedistillation}.
In this context, our results provide a rare alternative setting in which the optimal rate can be precisely characterized: specifically, our results show that the optimal rate of magic state distillation under asymptotically resource-non-generating operations equals the regularized relative entropy of magic defined with respect to $\mathcal{F}_\mathrm{STAB}$.

Still, the regularized relative entropy of resource, whether it is for entanglement or magic, is not tractable to compute in general, but we also have QRTs where the regularized relative entropy of resource is easy to compute.
The QRT of asymmetry provides such an example.
In the QRT of asymmetry, we consider a compact group $G$ and its unitary (projective) representation $f$ over a Hilbert space ${\cal H}$; for example, we can consider the group $G$ of diagonal unitaries as a special case, which leads to the QRT of coherence~\cite{Streltsov2017}.
Then, we define the projective representation $f^{(n)}$ of $G^n$ over ${\cal H}^{\otimes n}$ as $f^{(n)}(g_1,\ldots,g_n)\coloneqq \bigotimes_{j=1}^n f(g_j)$.
To define the set of free operations, we use a twirling map in the form of
\begin{align}
{\cal G}(\rho)\coloneqq \int_{G^n}
f^{(n)}(g_1,\ldots,g_n) \rho f^{(n)}(g_1,\ldots,g_n)^\dagger
\mu(dg_1) \cdots \mu(dg_n),
\end{align}
where $\mu$ is the Haar measure of $G$.
Then, in the QRT of asymmetry, one often considers the following set of free operations:
\begin{align}
{\cal O}\qty({\cal H}^{\otimes n}\to{\cal H}^{\otimes n})\coloneqq \qty{ \Gamma \in {\cal C}\qty({\cal H}^{\otimes n})
:  {\cal G}_n\circ \Gamma= \Gamma  \circ {\cal G}_n={\cal G}_n }.
\end{align}
Correspondingly, the set ${\cal F}_\mathrm{I}\qty(\mathcal{H})$ of free states of $\mathcal{H}$ becomes that of invariant states, i.e.,
\begin{align}
{\cal F}_\mathrm{I}\qty(\mathcal{H})\coloneqq\qty{ \rho \in {\cal D}\qty({\cal H}): \forall g \in G,f(g) \rho f(g)^\dagger =\rho};
\end{align}
in the same way, ${\cal F}_\mathrm{I}\qty(\mathcal{H}^{\otimes n})$ becomes the set of invariant states for $f^{(n)}$.
The set ${\cal F}_\mathrm{I}$ is convex, and contains the mixed state $\mathds{1}/\dim\qty({\cal H})$; also, the set ${\cal F}_\mathrm{I}\qty(\mathcal{H}^{\otimes(n+m)})$ includes ${\cal F}_\mathrm{I}\qty(\mathcal{H}^{\otimes n})\otimes{\cal F}_\mathrm{I}\qty(\mathcal{H}^{\otimes m})$.
Hence, ${\cal F}_\mathrm{I}$ satisfies Properties~\ref{p1:si},~\ref{p3:si}, and~\ref{p4:si}.
For this ${\cal F}_\mathrm{I}$, the regularized relative entropy of asymmetry is evaluated as
\begin{align}
\lim_{n\to\infty}\frac{1}{n}\min_{ \sigma \in {\cal F}_\mathrm{I}}
D\left(\rho^{\otimes n}\middle|\sigma\right)=\lim_{n\to\infty}\frac{1}{n}\cdot
n \min_{ \sigma \in {\cal F}_\mathrm{I}}
D(\rho|\sigma)
=D\left(\rho\middle| {\cal G}_1(\rho)\right).
\end{align}
When a state $\rho^{\otimes n}$ is given,
and the above free operations are allowed for coding operations,
the above quantity works as the channel capacity~\cite{9729748} and also plays an important role in dense coding~\cite{PRXQuantum.3.030346,TH2024}; that is, the relative entropy of asymmetry has another operational meaning, in addition to the meaning of the asymptotic conversion rate under asymptotically resource-non-generating operations discussed in this work.

\subsection{Example of QRTS for CQ channels with replacers}
\label{sec:examples_channels}

Whereas our results establish the second law of QRTs for a fundamental class of dynamical resources, i.e., classical-quantum (CQ) channels, QRTs for CQ channels may currently be less widely studied than those for states.
However, QRTs for CQ channels are relevant in the analysis of various communication tasks, such as entanglement-assisted quantum channel coding~\cite{PhysRevLett.83.3081}, visible compression of mixed states~\cite{PhysRevA.64.022308}, and the quantum reverse Shannon theorem~\cite{IEEE-IT-6757002,BCR2011}.
As an example of a QRT for CQ channels, we consider the case of replacers. In particular, for a quantum state $\sigma$, we define a replacer channel as one that always outputs $\sigma$ regardless of the input, denoted by
\begin{align}
{\cal N}_\sigma (\rho)\coloneqq\Tr[\rho]\sigma.
\end{align}

\subsubsection{QRTs of free CQ channels with replacers}
\label{SS2A}

In our formulation of QRTs with the second law, we consider asymptotically free operations that convert CQ channels into CQ channels and satisfy two properties: the asymptotically resource-non-generating property~\eqref{eq:asymptotically_resource_non_generating_operations_si} and the asymptotic continuity~\eqref{eq:condition_asymptotic_continuity_si}.
As a typical example of the set of free CQ channels, we focus on a set ${\cal F}_\mathrm{R}$ of replacers, defined as
\begin{align}
{\cal F}_\mathrm{R}\qty(\mathcal{H}_\mathrm{in}\to\mathcal{H}_\mathrm{out})\coloneqq\qty{{\cal N}_\sigma:\sigma \in {\cal D}({\cal H}_\mathrm{out})}.
\end{align}
By definition, this set ${\cal F}_\mathrm{R}$ satisfies Properties~\ref{p1:si}--\ref{p4:si}, as required.
As discussed in the main text, the asymptotically free operations are intended to be interpreted as a relaxation of a certain class of free operations, which are resource-non-generating superchannels with respect to ${\cal F}_\mathrm{R}$ in this case.
However, unlike the case of QRTs for states, where the asymptotically resource-non-generating property alone suffices, in the case of CQ channels, asymptotic continuity is additionally assumed for the asymptotically free operations, making this interpretation more non-trivial. 
To address this point, in the following analysis, we identify appropriate properties of superchannels such that the asymptotically free operations can be properly regarded as a relaxation of the set of such superchannels, thereby clarifying the scope of applications of our framework.

To obtain the asymptotically resource-non-generating property as a relaxation, we consider the set of resource-non-generating operations.
The asymptotically resource-non-generating property is introduced in terms of the generalized robustness $R_\mathrm{G}$ given by~\eqref{eq:R_G_si} with $\mathcal{F}=\mathcal{F}_\mathrm{R}$.
Using $R_\mathrm{G}$, the set of resource-non-generating operations is given by
\begin{align}
&{\cal O}_\mathrm{RNG}\qty(
  \qty(\mathcal{H}_{\mathrm{in}}^{(1)}\to
  \mathcal{H}_{\mathrm{out}}^{(1)})
  \to
  \qty(\mathcal{H}_\mathrm{in}^{(2)}\to\mathcal{H}_\mathrm{out}^{(2)})
  )\coloneqq\nonumber\\
&\qty{ \Theta \in 
  \mathcal{C}_{\mathrm{CQ}} \qty(
  \qty(\mathcal{H}_{\mathrm{in}}^{(1)}\to
  \mathcal{H}_{\mathrm{out}}^{(1)})
  \to
  \qty(\mathcal{H}_\mathrm{in}^{(2)}\to\mathcal{H}_\mathrm{out}^{(2)})
  )
: \forall
\mathcal{N}_{\mathrm{free}} \in {\cal F}_\mathrm{R},\,R_\mathrm{G}\qty(\Theta(\mathcal{N}_{\mathrm{free}}))=0
}.
\end{align}
Correspondingly, the set of asymptotically resource-non-generating operations, satisfying~\eqref{eq:asymptotically_resource_non_generating_operations_si}, is given by
\begin{align}
&\tilde{O}_\mathrm{RNG}\qty(
  \qty(\mathcal{H}_{\mathrm{in}}^{(1)}\to
  \mathcal{H}_{\mathrm{out}}^{(1)})
  \to
  \qty(\mathcal{H}_\mathrm{in}^{(2)}\to\mathcal{H}_\mathrm{out}^{(2)})
  )
  \coloneqq\nonumber\\
  &\left\{ \qty{\Theta_n \in 
  \mathcal{C}_{\mathrm{CQ}}\qty(
  \qty(\mathcal{H}_{\mathrm{in}}^{(1)\otimes n}\to
  \mathcal{H}_{\mathrm{out}}^{(1)\otimes n})
  \to
  \qty(\mathcal{H}_\mathrm{in}^{(2)\otimes n}\to\mathcal{H}_\mathrm{out}^{(2)\otimes n})
  )}_n
:\right.\nonumber\\
&\left.\forall
\qty{\mathcal{N}_{\mathrm{free}}^{(n)} \in {\cal F}_\mathrm{R}\qty(\mathcal{H}_{\mathrm{in}}^{(1)\otimes n}\to
  \mathcal{H}_{\mathrm{out}}^{(1)\otimes n})}_n,\,\lim_{n\to\infty}R_\mathrm{G}\qty(\Theta_n(\mathcal{N}_{\mathrm{free}}^{(n)}))=0
\right\}.
\end{align}
We can interpret $\tilde{O}_\mathrm{RNG}$ as a relaxation of ${\cal O}_\mathrm{RNG}$.

Additionally, we study conditions to obtain the asymptotic continuity as their relaxation.
A superchannel $\Theta\in\mathcal{C}\qty(
\qty(\mathcal{H}_{\mathrm{in}}^{(1)}\to
\mathcal{H}_{\mathrm{out}}^{(1)})
\to
\qty(\mathcal{H}_\mathrm{in}^{(2)}\to\mathcal{H}_\mathrm{out}^{(2)})
)$ is defined as a linear map that transforms any quantum-quantum (QQ) channel
\begin{align}
    \mathcal{N}\in\mathcal{C}\qty(\qty(\mathcal{H}_\mathrm{in}^{(\mathrm{aux})}\otimes\mathcal{H}_\mathrm{in}^{(1)})\to\qty(\mathcal{H}_\mathrm{out}^{(\mathrm{aux})}\otimes\mathcal{H}_\mathrm{out}^{(1)}))
\end{align}
into a channel
\begin{align}
    \qty(\id\otimes\Theta)\qty(\mathcal{N})\in\mathcal{C}\qty(\qty(\mathcal{H}_\mathrm{in}^{(\mathrm{aux})}\otimes\mathcal{H}_\mathrm{in}^{(2)})\to\qty(\mathcal{H}_\mathrm{out}^{(\mathrm{aux})}\otimes\mathcal{H}_\mathrm{out}^{(2)})),
\end{align}
where $\mathcal{H}_\mathrm{in}^{(\mathrm{aux})}$ and $\mathcal{H}_\mathrm{out}^{(\mathrm{aux})}$ represent any auxiliary systems, and $\id$ is the identity linear supermap that converts any channel $\mathcal{N}^{(\mathrm{aux})}\in\mathcal{C}\qty(\mathcal{H}_\mathrm{in}^{(\mathrm{aux})}\to\mathcal{H}_\mathrm{out}^{(\mathrm{aux})})$ to itself~\cite{PhysRevLett.101.060401,PhysRevA.80.022339}.
As shown in Ref.~\cite[(39) and Theorem~5]{PhysRevA.80.022339}, any superchannel $\Theta$ satisfying this condition can be represented by the corresponding completely positive (CP) linear map
\begin{align}
\label{eq:tilde_theta_definition}
    \tilde{\Theta}\coloneqq J\circ\Theta\circ J^{-1}
\end{align}
to transform the Choi state $J\qty(\mathcal{N})\in\mathcal{D}\qty(\mathcal{H}_{\mathrm{in}}^{(1)}\otimes
\mathcal{H}_{\mathrm{out}}^{(1)})$ to the Choi state $\tilde{\Theta}\qty(J\qty(\mathcal{N}))=J\qty(\Theta\qty(\mathcal{N}))\in\mathcal{D}\qty(\mathcal{H}_{\mathrm{in}}^{(2)}\otimes
\mathcal{H}_{\mathrm{out}}^{(2)})$ satisfying the following condition to be a superchannel (SC):
\begin{align}
\label{eq:normalizaton_condition_si}
    &\text{SC condition:}\notag\\
    &\quad\text{$\tilde{\Theta}$ is a CP linear map}\nonumber\\
    &\quad\text{with its Choi operator $J_2\coloneqq J\qty(\tilde{\Theta})=\qty(\id\otimes\tilde{\Theta})\qty(\Phi_{d_{\mathrm{in}}^{(1)}d_{\mathrm{out}}^{(1)}})\geq 0$ satisfying for some operator $J_1\geq 0$}\nonumber\\
    &\quad\Tr_{\mathrm{out}{(2)}}[J_2]=\frac{\mathds{1}_{\mathrm{out}}^{(1)}}{d_{\mathrm{out}}^{(1)}}\otimes J_1,\quad
    \Tr_{\mathrm{in}{(1)}}[J_1]=\frac{\mathds{1}_{\mathrm{in}}^{(2)}}{d_{\mathrm{in}}^{(2)}},
\end{align}
where for each $j\in\qty{1,2}$, we write the partial traces over $\mathcal{H}_\mathrm{in}^{(j)}$ and $\mathcal{H}_\mathrm{out}^{(j)}$ as $\Tr_{\mathrm{in}{(j)}}$ and $\Tr_{\mathrm{out}{(j)}}$, respectively, and write $d_\mathrm{in}^{(j)}\coloneqq\dim\qty(\mathcal{H}_\mathrm{in}^{(j)})$ (for brevity, we may also write $d_j\coloneqq d_\mathrm{in}^{(j)}$) and $d_\mathrm{out}^{(j)}\coloneqq\dim\qty(\mathcal{H}_\mathrm{out}^{(j)})$.
Specifically, for any CQ channel with $d$-dimensional inputs
\begin{align}
    \mathcal{N}\qty(\rho)=\sum_{k=0}^{d-1}\Tr[\ket{k}\bra{k}\rho]\rho_k,
\end{align}
its Choi state is a CQ state
\begin{align}
    J\qty(\mathcal{N})=\frac{1}{d}\sum_{k=0}^{d-1}\ket{k}\bra{k}\otimes\rho_k;
\end{align}
then, a superchannel $\Theta\in\mathcal{C}_\mathrm{CQ}\qty(
  \qty(\mathcal{H}_{\mathrm{in}}^{(1)}\to
  \mathcal{H}_{\mathrm{out}}^{(1)})
  \to
  \qty(\mathcal{H}_\mathrm{in}^{(2)}\to\mathcal{H}_\mathrm{out}^{(2)})
  )$ transforming CQ channels into CQ channels is a special case of the superchannels in $\mathcal{C}\qty(
  \qty(\mathcal{H}_{\mathrm{in}}^{(1)}\to
  \mathcal{H}_{\mathrm{out}}^{(1)})
  \to
  \qty(\mathcal{H}_\mathrm{in}^{(2)}\to\mathcal{H}_\mathrm{out}^{(2)}) )$, and the corresponding CP linear map $\tilde{\Theta}$ in~\eqref{eq:tilde_theta_definition} additionally satisfies the following condition:
\begin{align}
\label{eq:CQ_condition_si}
\text{CQ condition: $\tilde{\Theta}$ transforms CQ states to CQ states}.
\end{align}
However, as argued in Methods, the CP linear map $\tilde{\Theta}$ corresponding to a given superchannel $\Theta$ may not be a CPTP map in general.
Superchannels in this subclass satisfy the following conditions.

\begin{lemma}\label{LLLO}
Assume that a CP linear map $\tilde{\Theta}$ transforming operators on $\mathcal{H}_{\mathrm{in}}^{(1)}\otimes \mathcal{H}_{\mathrm{out}}^{(1)}$ to those on $\mathcal{H}_{\mathrm{in}}^{(2)}\otimes
\mathcal{H}_{\mathrm{out}}^{(2)}$ satisfy the SC condition~\eqref{eq:normalizaton_condition_si} and the CQ condition~\eqref{eq:CQ_condition_si}.
For each CQ channel $\mathcal{N}\in\mathcal{C}_\mathrm{CQ}\qty(\mathcal{H}_{\mathrm{in}}^{(1)}\to
  \mathcal{H}_{\mathrm{out}}^{(1)})$, which is in the form of
\begin{align}
    \mathcal{N}\qty(\rho)=\sum_{k=0}^{d_1-1}\Tr[\ket{k}\bra{k}\rho]\rho_k,
\end{align}
we define a function
\begin{align}
\label{eq:P_K_K_prime}
P_{K|K'}^{\qty(\mathcal{N})}(k|k')\coloneqq
\frac{d_2}{d_1}\Tr[\tilde{\Theta}\qty(\ket{k} \bra{k} \otimes \rho_k) \qty(\ket{k'}\bra{k'} \otimes \mathds{1})].
\end{align}
Then, under the above two conditions, $P_{K|K'}^{\qty(\mathcal{N})}(k|k')$ becomes a conditional probability distribution. 
\end{lemma}

\begin{proof}
For any $k'$, we will show $P_{K|K'}^{\qty(\mathcal{N})}(k|k')\geq 0$ and $\sum_k P_{K|K'}^{\qty(\mathcal{N})}(k|k')=1$.
Since $\rho_k\geq 0$ and $\tilde{\Theta}$ is a CP map, we have
\begin{align}
    P_{K|K'}^{\qty(\mathcal{N})}(k|k')=\frac{d_2}{d_1}\Tr[\tilde{\Theta}\qty(\ket{k} \bra{k} \otimes \rho_k) \qty(\ket{k'}\bra{k'} \otimes \mathds{1})]\geq 0.
\end{align}
Also, it holds that
\begin{align}
\sum_k P_{K|K'}^{\qty(\mathcal{N})}(k|k')&=\sum_k
\frac{d_2}{d_1}\Tr[\tilde{\Theta}\qty(\ket{k} \bra{k} \otimes \rho_k) \qty(\ket{k'}\bra{k'} \otimes \mathds{1})]\\
&\stackrel{(a)}{=}d_2\Tr[\tilde{\Theta}\qty(\frac{1}{d_1}\sum_k\ket{k} \bra{k} \otimes \rho_k) \qty(\ket{k'}\bra{k'} \otimes \mathds{1})]\\
&\stackrel{(b)}{=}d_2\Tr[\frac{\mathds{1}}{d_2}\ket{k'}\bra{k'}]\\
&=1,
\end{align}
where $(a)$ follows from the linearity of $\tilde{\Theta}$, and $(b)$ holds due to the assumptions that $\tilde{\Theta}$ maps a CQ Choi state $J\qty(\mathcal{N})=\frac{1}{d_1}\sum_k\ket{k} \bra{k} \otimes \rho_k$ into a CQ Choi state $\tilde{\Theta}\qty(J\qty(\mathcal{N}))$ satisfying $\Tr_{\mathrm{out},(2)}\qty[\tilde{\Theta}\qty(J\qty(\mathcal{N}))]=\frac{\mathds{1}}{d_2}$ with $\Tr_{\mathrm{out},(2)}$ denoting the partial trace over $\mathcal{H}_\mathrm{out}^{(2)}$.
\end{proof}

To identify the desired conditions for asymptotically free operations, we study a further restricted subclass of superchannels $\Theta$, each corresponding to a CPTP map $\tilde{\Theta}$ with analytically tractable properties rather than merely corresponding to a CP map.
In addition to the two conditions in Lemma~\ref{LLLO},
we assume that $\tilde{\Theta}$ is trace-preserving (TP):
\begin{align}
\label{eq:TP_condition_si}
    \text{TP condition: $\tilde{\Theta}$ is a TP map}.
\end{align}
Moreover, we introduce the following condition, which is called the encoding independence (EI) condition:
\begin{align}
\label{eq:encoding_independence_condition}
    \text{EI condition: the conditional distribution $P_{K|K'}^{(\mathcal{N})}(k|k')$ does not depend on the CQ channel $\mathcal{N}$};
\end{align}
that is, the function defined in~\eqref{eq:P_K_K_prime} is independent of $\rho_k$ appearing on its right-hand side.
Under this condition, we may write $P_{K|K'}^{(\mathcal{N})}(k|k')$ as $P_{K|K'}(k|k')$, omitting the superscript.
With these additional conditions, we reach the following characterization;
in particular, the relation~\eqref{VBE} in the following lemma clarifies its meaning: after transforming a CQ channel ${\cal N}=\sum_{k=0}^{d_1-1}\Tr[\ket{k}\bra{k}\rho]\rho_k$ by a superchannel $\Theta$ in this subclass, the resulting channel's input symbol $k'$ is converted into $k$ by the conditional distribution $P_{K|K'}$ independently of $\rho_k$, and the corresponding output state $\rho_k$ is converted by a CPTP map $\Gamma_{k,k'}$ depending on $k',k$.

\begin{lemma}\label{LLLO2}
Assume that a CP linear map $\tilde{\Theta}$ satisfies the SC and CQ conditions in Lemma~\ref{LLLO} and, additionally, 
the EI condition \eqref{eq:encoding_independence_condition}.
Then, there exists a family $\qty{\Gamma_{k,k'}}_{k,k'}$ of CPTP linear maps such that any CQ channel ${\cal N}\qty(\rho)=\sum_{k=0}^{d_1-1}\Tr[\ket{k}\bra{k}\rho]\rho_k$ is mapped into a CQ channel $\Theta({\cal N})$ satisfying
\begin{align}
\qty(\Theta\qty({\cal N}))\qty(\rho)
=\sum_{k'=0}^{d_2-1}\Tr[\ket{k'}\bra{k'}\rho]\sum_{k=0}^{d_1-1}P_{K|K'}(k|k') \Gamma_{k,k'}\qty(\rho_k),
\label{VBE}
\end{align}
where $P_{K|K'}$ is the conditional probability distribution in~\eqref{eq:encoding_independence_condition}.
Moreover, if the TP condition~\eqref{eq:TP_condition_si} additionally holds,
$P_{K|K'}$ satisfies, for each $k$,
\begin{align}
\frac{1}{d_2}\sum_{k'=0}^{d_2-1}P_{K|K'}(k|k')=\frac{1}{d_1}.
\label{CB9}
\end{align}
\end{lemma}

\begin{proof}
We first show~\eqref{VBE}, and later~\eqref{CB9}.
Using
the EI condition \eqref{eq:encoding_independence_condition},
we define a map
\begin{align}
\label{eq:Gamma_k_k_prime}
\Gamma_{k,k'}\qty(\rho) \coloneqq \Tr_{\mathrm{in},{(2)}}\qty[\tilde\Theta\qty(\frac{d_2}{d_1P_{K|K'}(k|k')}|k\rangle \langle k| \otimes \rho) \qty(|k'\rangle \langle k'| \otimes \mathds{1})].
\end{align}
Since $\tilde{\Theta}$ is a CP linear map, the map $\Gamma_{k,k'}$ is also a CP linear map.
Since~\eqref{eq:P_K_K_prime} and~\eqref{eq:Gamma_k_k_prime} imply
\begin{align}
P_{K|K'}(k|k')
=&
\frac{d_2}{d_1}\Tr[\tilde{\Theta}\qty(\ket{k} \bra{k} \otimes \rho_k) \qty(\ket{k'}\bra{k'} \otimes \mathds{1})]
=
P_{K|K'}(k|k')
\Tr[\Gamma_{k,k'}\qty(\rho_k)],
\end{align}
it holds that
\begin{align}
\label{eq:TP_proof}
    \Tr[\Gamma_{k,k'}\qty(\rho_k)]=1;
\end{align}
i.e., the CP map $\Gamma_{k,k'}$ is trace-preserving.

Since we have
$\tilde{\Theta}\qty(\frac{1}{d_1}\sum_{k=0}^{d_1-1}|k\rangle \langle k| \otimes \rho_k)
=
\frac{1}{d_2}\sum_{k'=0}^{d_2-1} |k'\rangle \langle k'| \otimes 
\qty(\Theta({\cal N}))\qty(\ket{k'}\bra{k'})$ by definition of $\tilde{\Theta}$,
it follows that
\begin{align}
\qty(\Theta({\cal N}))(\ket{k'}\bra{k'}) =& 
\Tr_{\mathrm{in},{(2)}}\qty[\tilde\Theta\qty(\frac{1}{d_1}\sum_{k=0}^{d_1-1}|k\rangle \langle k| \otimes \rho_k) \qty(d_2|k'\rangle \langle k'| \otimes \mathds{1})]\notag\\
=&\sum_{k=0}^{d_1-1} P_{K|K'}(k|k') \Gamma_{k,k'}(\rho_k),\label{VBE5}
\end{align}
where the last line follows from~\eqref{eq:Gamma_k_k_prime}.
This implies~\eqref{VBE} due to the linearity of $\Theta({\cal N})$.

To show~\eqref{CB9}, we write the dual maps of $\tilde{\Theta}$ and $\Gamma_{k,k'}$ as $\tilde{\Theta}^*$ and $\Gamma_{k,k'}^*$, respectively. By definition of dual maps, we have
\begin{align}
&\Tr 
[\qty(\sum_{k'=0}^{d_2-1}|k'\rangle \langle k'| \otimes Y_{k'})
\tilde{\Theta}\qty(\frac{1}{d_1}\sum_{k=0}^{d_1-1}|k\rangle \langle k| \otimes X_k)] =\Tr[
\tilde{\Theta}^*\qty(\sum_{k'=0}^{d_2-1}|k'\rangle \langle k'| \otimes Y_{k'})
\qty(\frac{1}{d_1}\sum_{k=0}^{d_1-1}|k\rangle \langle k| \otimes X_k)] 
\label{ADD1}
\end{align}
Since~\eqref{VBE5} yields
\begin{align}
\tilde{\Theta}\qty(\frac{1}{d_1}\sum_{k=0}^{d_1-1}|k\rangle \langle k| \otimes X_k)
=&
\frac{1}{d_2}\sum_{k'=0}^{d_2-1} |k'\rangle \langle k'| \otimes 
\sum_{k=0}^{d_1-1} P_{K|K'}(k|k') \Gamma_{k,k'}(X_k),
\end{align}
we have
\begin{align}
&\Tr[
\qty(\sum_{k'=0}^{d_2-1}|k'\rangle \langle k'| \otimes Y_{k'})
\tilde{\Theta}\qty(\frac{1}{d_1}\sum_{k=0}^{d_1-1}|k\rangle \langle k| \otimes X_k)]\notag\\
&=
\Tr[
\qty(\sum_{k'=0}^{d_2-1}|k'\rangle \langle k'| \otimes Y_{k'})
\qty(\frac{1}{d_2}\sum_{k'=0}^{d_2-1} |k'\rangle \langle k'| \otimes 
\sum_{k=0}^{d_1-1} P_{K|K'}(k|k') \Gamma_{k,k'}(X_k) )]\\
&=
\frac{1}{d_2}\sum_{k'=0}^{d_2-1}\sum_{k=0}^{d_1-1} P_{K|K'}(k|k') \Tr[
Y_{k'}
\Gamma_{k,k'}\qty(X_k) ]\\
&=
\frac{1}{d_2}\sum_{k'=0}^{d_2-1}\sum_{k=0}^{d_1-1} P_{K|K'}(k|k') \Tr[
\Gamma_{k,k'}^\ast\qty(Y_{k'})
X_k]\\
&=
\Tr[\qty(\frac{d_1}{d_2}\sum_{k=0}^{d_1-1} |k\rangle \langle k| \otimes 
\sum_{k'=0}^{d_2-1}P_{K|K'}(k|k') \Gamma_{k,k'}^*(Y_{k'}))
\qty(\frac{1}{d_1}\sum_{k=0}^{d_1-1}|k\rangle \langle k| \otimes X_k)].
\label{ADD2}
\end{align}
The combination of 
\eqref{ADD1} and \eqref{ADD2} implies that the dual map $\tilde{\Theta}^*$ of $\tilde{\Theta}$ is given by
\begin{align}
\tilde{\Theta}^*
\qty(\sum_{k'=0}^{d_2-1}|k'\rangle \langle k'| \otimes Y_{k'})
=&
\frac{d_1}{d_2}\sum_{k=0}^{d_1-1} |k\rangle \langle k| \otimes 
\sum_{k'=0}^{d_2-1} P_{K|K'}(k|k') \Gamma_{k,k'}^*(Y_{k'}).
\end{align}

Then, the TP condition \eqref{eq:TP_condition_si} assumed for $\tilde{\Theta}$ leads to
\begin{align}
\label{BV4_dual1}
\sum_{k=0}^{d_1-1}
|k\rangle \langle k|\otimes  
\mathds{1} &=
\tilde{\Theta}^* \qty(
\sum_{k'=0}^{d_2-1}
|k'\rangle \langle k'|\otimes  
\mathds{1})\\
&=
\frac{d_1}{d_2} 
\sum_{k=0}^{d_1-1} |k\rangle \langle k| \otimes 
\sum_{k'=0}^{d_2-1} P_{K|K'}(k|k') \Gamma_{k,k'}^*(\mathds{1}).
\label{BV4_dual}
\end{align}
Therefore, we have
\begin{align}
1&=\Tr[|k\rangle \langle k|\otimes \rho_k]\notag \\
&\stackrel{(a)}{=}
\Tr[\qty(|k\rangle \langle k|\otimes \rho_k)
\qty(
\frac{d_1}{d_2} 
\sum_{k'=0}^{d_2-1}\sum_{k=0}^{d_1-1}
|k\rangle \langle k|\otimes 
P_{K|K'}(k|k') \Gamma_{k,k'}^*(\mathds{1})
)]\notag\\
&\stackrel{(b)}{=}
\sum_{k'=0}^{d_2-1} \frac{d_1}{d_2} P_{K|K'}(k|k')
\Tr[\Gamma_{k,k'}(\rho_k)]\notag\\
&\stackrel{(c)}{=}\frac{d_1}{d_2}\sum_{k'=0}^{d_2-1}P_{K|K'}(k|k'),
\end{align}
where $(a)$ follows from inserting the identity operator~\eqref{BV4_dual}, $(b)$ is the duality $\Tr[\rho_k\Gamma_{k,k'}^\ast(\mathds{1})]=\Tr[\Gamma_{k,k'}(\rho_k)\mathds{1}]$, and $(c)$ is obtained from~\eqref{eq:TP_proof}, 
leading to~\eqref{CB9}.
\end{proof}

Given a conditional distribution $P_{K|K'}$ and a set $\{\Gamma_{k,k'}\}$ of CPTP maps, we define the conversion 
$\Theta_{P_{K|K'},\{\Gamma_{k,k'}\}}$ of QQ channels as
\begin{align}
\label{eq:QQ_form}
\qty(\Theta_{P_{K|K'},\{\Gamma_{k,k'}\}} ({\cal N}))(\rho)
\coloneqq
\sum_{k'=0}^{d_2-1}\Tr[\ket{k'}\bra{k'}\rho]\sum_{k=0}^{d_1-1}P_{K|K'}(k|k') \Gamma_{k,k'}\qty(
{\cal N}\qty(\ket{k} \bra{k})).
\end{align}
Because our interest is limited to the conversion of CQ channels, we can, without loss of generality, identify the map in the form of~\eqref{VBE} in Lemma~\ref{LLLO2} with the map in the form of~\eqref{eq:QQ_form}.
Indeed, Lemma~\ref{LLLO2} guarantees that any CP map satisfying the SC, CQ, and EI conditions can be expressed as~\eqref{VBE}, which is henceforth identified with~\eqref{eq:QQ_form}.

Based on this characterization, we define a subset of superchannels satisfying the four conditions assumed in Lemma~\ref{LLLO2}:
\begin{align}
\label{eq:O_CPTP}
  &{\cal O}_\mathrm{CPTP}\qty(
  \qty(\mathcal{H}_{\mathrm{in}}^{(1)}\to
  \mathcal{H}_{\mathrm{out}}^{(1)})
  \to
  \qty(\mathcal{H}_\mathrm{in}^{(2)}\to\mathcal{H}_\mathrm{out}^{(2)})
  )\coloneqq\left\{\Theta\in\mathcal{C}_\mathrm{CQ}\qty(
  \qty(\mathcal{H}_{\mathrm{in}}^{(1)}\to
  \mathcal{H}_{\mathrm{out}}^{(1)})
  \to
  \qty(\mathcal{H}_\mathrm{in}^{(2)}\to\mathcal{H}_\mathrm{out}^{(2)})
  ):\right.\notag\\
  &\quad\left.\text{$\tilde{\Theta}=J\circ\Theta\circ J^{-1}$ satisfies SC, CQ, TP, and EI conditions in Lemma~\ref{LLLO2}}
\right\}.
\end{align}
As stated below, a notable property that holds for this subclass of superchannels is the monotonicity of the trace distance between any Choi states of CQ channels, which is crucial for establishing asymptotic continuity as a relaxation.

\begin{lemma}\Label{LBY}
Any superchannel $\Theta \in \mathcal{O}_\mathrm{TPCP}$ in~\eqref{eq:O_CPTP} leads to the monotonicity of the trace distance between any Choi states of CQ channels $\mathcal{N}$ and $\mathcal{N}'$
\begin{align}
\left\|J\qty(\Theta\qty(\mathcal{N})) -J\qty(\Theta\qty(\mathcal{N}')) \right\|_1
&\leq
\left\|J\qty(\mathcal{N}) -J\qty(\mathcal{N}') \right\|_1.
\label{BNSA}
\end{align}
\end{lemma}

\begin{proof}
By definition~\eqref{eq:O_CPTP} of $\mathcal{O}_\mathrm{TPCP}$, due to Lemma~\ref{LLLO2}, $\Theta$ maps any CQ channel $\mathcal{N}$ with its Choi state
\begin{align}
\label{eq:J_N_S21}
    J\qty(\mathcal{N})=\frac{1}{d_1}\sum_{k}\ket{k}\bra{k}\otimes \rho_k
\end{align}
into
\begin{align}
\label{eq:J_Theta_N_S21}
    J\qty(\Theta\qty(\mathcal{N}))=\frac{1}{d_2}\sum_{k'}\ket{k'}\bra{k'}\otimes P_{K|K'}\qty(k|k')\Gamma_{k,k'}\qty(\rho_k).
\end{align}
Using $P_{K|K'}$, which satisfies~\eqref{CB9} in Lemma~\ref{LLLO2}, we define a function
\begin{align}
\label{eq:P_prime_K_K_prime}
    P_{K'|K}^{\prime}\qty(k'|k)\coloneqq \frac{d_1}{d_2}P_{K|K'}\qty(k|k').
\end{align}
Then, $P_{K'|K}^{\prime}$ also becomes a conditional probability distribution because $P_{K'|K}^{\prime}\qty(k'|k)\geq 0$, and~\eqref{CB9} yields
$\sum_{k'}P_{K'|K}^{\prime}\qty(k'|k)=\frac{d_1}{d_2}\sum_{k'}P_{K|K'}\qty(k|k')=1$.
With this conditional distribution $P_{K'|K}^{\prime}$, for any superchannel $\Theta\in{\cal O}_\mathrm{CPTP}\qty(
  \qty(\mathcal{H}_{\mathrm{in}}^{(1)}\to
  \mathcal{H}_{\mathrm{out}}^{(1)})
  \to
  \qty(\mathcal{H}_\mathrm{in}^{(2)}\to\mathcal{H}_\mathrm{out}^{(2)})
  )$, the corresponding CP linear map $\tilde{\Theta}\in\mathcal{C}\qty(
 \mathcal{H}_{\mathrm{in}}^{(1)}\otimes
  \mathcal{H}_{\mathrm{out}}^{(1)}
  \to
  \mathcal{H}_\mathrm{in}^{(2)}\otimes\mathcal{H}_\mathrm{out}^{(2)}
  )$ becomes a CPTP map given by
\begin{align}
    \tilde{\Theta}\qty(\rho)=\sum_{k'=0}^{d_2-1}\ket{k'}\bra{k'}\otimes\qty(\sum_{k=0}^{d_1-1} P_{K'|K}^{\prime}\qty(k'|k)\Gamma_{k,k'}\qty(\Tr_{\mathrm{in},(1)}\qty[\qty(\ket{k}\bra{k}\otimes\mathds{1}_\mathrm{out}^{(1)})\rho])),
\end{align}
which satisfies $\tilde{\Theta}\qty(J\qty(\mathcal{N}))=J\qty(\Theta\qty(\mathcal{N}))$ due to~\eqref{eq:J_N_S21},~\eqref{eq:J_Theta_N_S21}, and~\eqref{eq:P_prime_K_K_prime}
.
Consequently, from the monotonicity of the trace distance under the CPTP map $\tilde{\Theta}$, it follows that
\begin{align}
\left\|J\qty(\Theta\qty(\mathcal{N})) -J\qty(\Theta\qty(\mathcal{N}')) \right\|_1
=\left\|\tilde{\Theta}\qty(J\qty(\mathcal{N})) -\tilde{\Theta}\qty(J\qty(\mathcal{N}')) \right\|_1
&\leq
\left\|J\qty(\mathcal{N}) -J\qty(\mathcal{N}') \right\|_1.
\end{align}
\end{proof}

This lemma indicates that, for any sequence $\qty{\Theta_n\in\mathcal{O}_\mathrm{CPTP}}$ of superchannels in this subclass, any two sequences $\qty{\mathcal{N}_n}_n$ and $\qty{\mathcal{N}_n^\prime}_n$ of CQ channels satisfying $\lim_{n\to\infty}\frac{1}{2}\|J\qty(\mathcal{N}_n)-J\qty(\mathcal{N}_n^\prime)\|_1=0$ are transformed into sequences of CQ channels satisfying
\begin{align}
\label{eq:asymp_cont_CPTP}
    \lim_{n\to\infty}\frac{1}{2}\left\|J\qty(\Theta_n\qty(\mathcal{N}_n))-J\qty(\Theta_n\qty(\mathcal{N}_n^\prime))\right\|_1=0.
\end{align}
Correspondingly, the set of asymptotically continuous (AC) operations, satisfying~\eqref{eq:condition_asymptotic_continuity_si}, is given by
\begin{align}
&\tilde{O}_\mathrm{AC}\qty(
  \qty(\mathcal{H}_{\mathrm{in}}^{(1)}\to
  \mathcal{H}_{\mathrm{out}}^{(1)})
  \to
  \qty(\mathcal{H}_\mathrm{in}^{(2)}\to\mathcal{H}_\mathrm{out}^{(2)})
  )
  \coloneqq\nonumber\\
  &\left\{ \qty{\Theta_n \in 
  \mathcal{C}_{\mathrm{CQ}}\qty(
  \qty(\mathcal{H}_{\mathrm{in}}^{(1)\otimes n}\to
  \mathcal{H}_{\mathrm{out}}^{(1)\otimes n})
  \to
  \qty(\mathcal{H}_\mathrm{in}^{(2)\otimes n}\to\mathcal{H}_\mathrm{out}^{(2)\otimes n})
  )}_n
:\right.\nonumber\\
&\left.\forall
\qty{\mathcal{N}_n \in {\cal F}_\mathrm{R}\qty(\mathcal{H}_{\mathrm{in}}^{(1)\otimes n}\to
  \mathcal{H}_{\mathrm{out}}^{(1)\otimes n})}_n~\text{and}~\qty{\mathcal{N}_n^\prime \in {\cal F}_\mathrm{R}\qty(\mathcal{H}_{\mathrm{in}}^{(1)\otimes n}\to
  \mathcal{H}_{\mathrm{out}}^{(1)\otimes n})}_n~\text{satisfying}~\lim_{n\to\infty}\left\|J\qty(\mathcal{N}_n)-J\qty(\mathcal{N}_n^\prime)\right\|_1=0,\right.\nonumber\\
  &\left.\lim_{n\to\infty}\left\|J\qty(\Theta\qty(\mathcal{N}_n))-J\qty(\Theta\qty(\mathcal{N}_n^\prime))\right|_1=0
\right\}.
\end{align}
Due to~\eqref{eq:asymp_cont_CPTP}, we can interpret $\tilde{\mathcal{O}}_\mathrm{AC}$ as a relaxation of $\mathcal{O}_\mathrm{CPTP}$.

As a whole, the set of asymptotically free operations $\tilde{\mathcal{O}}$ is given by the intersection of the sets of superchannels satisfying the resource-non-generating property and the asymptotic continuity, i.e.,
\begin{align}
    \tilde{\mathcal{O}}=\tilde{\mathcal{O}}_\mathrm{RNG}\cap\tilde{\mathcal{O}}_\mathrm{AC},
\end{align}
and this can be interpreted as a relaxation of sets of free operations included in $O_\mathrm{RNG}\cap O_\mathrm{CPTP}$.
To characterize such sets of free operations, 
for any superchannel $\Theta\in\mathcal{O}_\mathrm{CPTP}$ with the conditional distribution $P_{K|K'}$ and the CPTP map $\Gamma_{k,k'}$ chosen according to Lemma~\ref{LLLO2},
we consider the following conditions:
\begin{description}
\item[(A1)]
The inclusion $\Theta\qty({\cal F}_\mathrm{R})\subset {\cal F}_\mathrm{R}$ holds.
\item[(A2)]
There exists a CPTP map $\Gamma_0\in\mathcal{C}\qty(\mathcal{H}_{\mathrm{out}}^{(1)}\to \mathcal{H}_{\mathrm{out}}^{(2)})$ such that, independently of $k'$,
\begin{align}
\Gamma_0=\sum_{k=0}^{d_1-1} P_{K|K'}(k|k') \Gamma_{k,k'}.
\end{align}
\item[(A3)]
There exist
two reference systems
$\mathcal{H}_{\mathrm{in}}^R$, $\mathcal{H}_{\mathrm{out}}^R$, 
an entangled state
$\tilde{\rho}$ on $\mathcal{H}_{\mathrm{in}}^R \otimes \mathcal{H}_{\mathrm{out}}^R$,
a POVM $M_{k'}=\qty{M_{k|k'}}_k$, and 
a CPTP map $\tilde{\Gamma}$
from ${\cal D}\qty( \mathcal{H}_{\mathrm{out}}^R
\otimes \mathcal{H}_{\mathrm{out}}^{(1)})$
to ${\cal D}\qty( \mathcal{H}_{\mathrm{out}}^{(2)})$ such that
the relations
\begin{align}
\Tr [(M_{k|k'}\otimes \mathds{1})\rho]
&=
P_{K|K'}(k|k')\\
\tilde{\Gamma}\qty(
\Tr_{\mathrm{in},R} \qty[(M_{k|k'}\otimes \mathds{1})
\tilde{\rho}]
\otimes \rho )
&=
 P_{K|K'}(k|k') \Gamma_{k,k'}(\rho)\label{GHA}
\end{align}
hold for any state $\rho \in {\cal D}\qty(\mathcal{H}_{\mathrm{out}}^{(1)})$, $k$, and $k'$, where $\Tr_{\mathrm{in},R}$ is the partial trance over $\mathcal{H}_\mathrm{in}^R$.
We write the collection of POVMs $(M_{k'})$ as ${\cal M}$.
The conversion is the above form is written as
$\Theta_{{\cal M},\tilde{\Gamma},\tilde{\rho}}$.
\end{description}
The condition (A1) is the one for $\mathcal{O}_\mathrm{RNG}\cap\mathcal{O}_\mathrm{CPTP}$.
The condition (A2) corresponds to the non-signaling condition; then, we write the set of superchannels satisfying (A2) as $\mathcal{O}_\mathrm{NS}$.
When the condition (A3) holds, the superchannel $\Theta$ can be implemented by the combination of the shared entangled state $\tilde{\rho}$ and the CPTP map $\tilde{\Gamma}$; then, we write the set of superchannels satisfying (A3)
as $\mathcal{O}_\mathrm{ENS}$.
The following lemma characterizes the relation of these sets of superchannels as
\begin{align}
    \mathcal{O}_\mathrm{ENS}\subset\mathcal{O}_{NS}=\mathcal{O}_\mathrm{RNG}\cap\mathcal{O}_\mathrm{CPTP}.
\end{align}

\begin{lemma}
The conditions (A1) and (A2) are equivalent.
The condition (A3) implies 
the conditions (A1) and (A2).
\end{lemma}

\begin{proof}
Assume the condition (A1).
Under this condition, for an arbitrary replacer $\mathcal{N}_\sigma$ with $\sigma \in {\cal D}\qty({\cal H}_{\mathrm{out}}^{(1)})$,
the output of the channel $\Theta({\cal N}_\sigma)$ does not depend on the input $k'$; that is, for an arbitrary state $\sigma \in {\cal D}\qty({\cal H}_{\mathrm{out}}^{(1)})$, the state 
\begin{align}
\qty(\Theta({\cal N}_\sigma))\qty(\ket{k'}\bra{k'})=
\sum_{k=0}^{d_1-1} P_{K|K'}(k|k') \Gamma_{k,k'}(\sigma)
\end{align}
does not depend on the input $k'$.
Thus, the map $\sum_{k=0}^{d_1-1} P_{K|K'}(k|k') \Gamma_{k,k'}$ does not depend on the input $k'$, which implies the condition (A2).

Assume the condition (A2).
Then, for an arbitrary replacer $\mathcal{N}_\sigma$ with $\sigma \in {\cal D}\qty({\cal H}_{\mathrm{out}}^{(1)})$ and any element $k'$, due to~\eqref{VBE},
the output of the channel $\Theta\qty({\cal N}_\sigma)$ is given by
\begin{align}
\qty(\Theta\qty({\cal N}_\sigma))\qty(\rho)=\sum_{k'=0}^{d_2-1}\Tr[\ket{k'}\bra{k'}\rho]\sum_{k=0}^{d_1-1}P_{K|K'}(k|k') \Gamma_{k,k'}\qty(\sigma)=\Gamma_0\qty(\sigma),
\end{align}
which implies 
$\Theta\qty({\cal N}_\sigma)\in {\cal F}_\mathrm{RNG}$, meaning that the condition (A1) holds.

Assume the condition (A3).
The relation~\eqref{GHA} implies 
\begin{align}
\sum_{k=0}^{d_1-1}
 P_{K|K'}(k|k') \Gamma_{k,k'}(\rho)
=
\sum_{k=0}^{d_1-1}
\tilde{\Gamma}\qty(
\Tr_{\mathrm{in},R}\qty[(M_{k|k'}\otimes \mathds{1}) \tilde{\rho}]
\otimes \rho )
=
\tilde{\Gamma}\qty(
\Tr_{\mathrm{in},R} \qty[(\mathds{1}\otimes \mathds{1})\tilde{\rho}]
\otimes \rho )
=
\tilde{\Gamma}\qty(
\Tr_{\mathrm{in},R}\qty[\tilde{\rho}]
\otimes \rho),
\end{align}
which yields the condition (A2).

\end{proof}

Additionally, we consider a subclass of $\mathcal{O}_\mathrm{ENS}$.
A superchannel $\Theta \in \mathcal{O}_\mathrm{ENS}$ is said to be classically correlated when there exists a probability distribution $Q_S$ of the common randomness to satisfy the following condition:
\begin{description}
    \item[(A4)]There exist a conditional distribution $P_{K|K',S}(k|k',s)$ and a family of $\qty{\tilde{\Gamma}_{s}}_s$ of CPTP maps such that
\begin{align}
\qty(\Theta({\cal N}))\qty(\ket{k'}\bra{k'})= 
\sum_{k=0}^{d_1-1} \sum_{s}Q_S(s) P_{K|K',S}(k|k',s)
\tilde{\Gamma}_{s}(\rho_k).
\end{align}
\end{description}
We write the set of superchannels satisfying (A4) as $\mathcal{O}_\mathrm{CNS}$.
Under this condition, when $P_{K|K'}(k|k')=\sum_{s}Q_S(s) P_{K|K',S}(k|k',s)$, we have
\begin{align}
\qty(\Theta({\cal N}))\qty(\ket{k'}\bra{k'})= 
\sum_{k=0}^{d_1-1} P_{K|K'}(k|k')
\sum_{s}Q_S(s) \frac{P_{K|K',S}(k|k',s)}{P_{K|K'}(k|k')}
\tilde{\Gamma}_{s}(\rho_k).
\end{align}
When we give $\Gamma_{k,k'}=\sum_{s}Q_S(s) \frac{P_{K|K',S}(k|k',s)}{P_{K|K'}(k|k')}
\tilde{\Gamma}_{s}$, the expression~\eqref{VBE} holds.
Overall, we have the inclusion relation
\begin{align}
\label{eq:replacer_inclusion_relation}
\mathcal{O}_\mathrm{CNS} \subset \mathcal{O}_\mathrm{ENS}
\subset \mathcal{O}_\mathrm{NS} = \mathcal{O}_\mathrm{RNG}\cap\mathrm{O}_\mathrm{CPTP},
\end{align}
and we can interpret the set $\tilde{O}$ of asymptotically free operations for $\mathcal{F}_\mathrm{R}$ as a relaxation of these sets.

\subsubsection{Asymptotic conversion rate}\label{SS2B}
We present an application of the second law of QRTs of CQ channels with replacers.
When we consider the set $\mathcal{F}_\mathrm{R}$ of replaces as free CQ channels, for any CQ channel $\mathcal{N}\in\mathcal{C}_\mathrm{CQ}\qty(\mathcal{H}_\mathrm{in}\to\mathcal{H}_\mathrm{out})$ with $d=\dim\qty(\mathcal{H}_\mathrm{in})$, the relative entropy of resource is calculated simply as
\begin{align}
 R_\mathrm{R}\qty(\mathcal{N})
&=\min_{\mathcal{N}_\mathrm{free}\in\mathcal{F}_\mathrm{R}}D\left(J\qty(\mathcal{N})\middle\|J\qty(\mathcal{N}_\mathrm{free})\right)\\
&=\min_{\sigma} D\left(J\qty(\mathcal{N})\middle\|J\qty(\mathcal{N}_\sigma)\right)\\
&=\min_{\sigma} \sum_{k=1}^d \frac{1}{d}D\left(\rho_k\middle\|\sigma\right)\\
&=\sum_{k=1}^{d}\frac{1}{d}D\left(\rho_k\middle\|
\sum_{k'=1}^{d}\frac{1}{d}\rho_{k'}\right)\\
&\eqqcolon I(\mathrm{in};\mathrm{out})[{\cal N}],
\label{eq:relative_entropy_of_resource4}
\end{align}
where we introduce the notation $I(\mathrm{in};\mathrm{out})[{\cal N}]$ in the last line since this is the same as the mutual information when the input to the channel is chosen from the uniform distribution.
This quantity satisfies the additivity
\begin{align}
R_\mathrm{R}\qty(\mathcal{N}_1\otimes \mathcal{N}_2)
=R_\mathrm{R}\qty(\mathcal{N}_1)+
R_\mathrm{R}\qty(\mathcal{N}_2),
\end{align}
implying that the regularized relative entropy of resource is also simply given by
\begin{align}
\label{eq:replacer_regularized_relative_entropy_of_resource}
R_\mathrm{R}^{\infty}\qty(\mathcal{N})=
R_\mathrm{R}\qty(\mathcal{N})=I(\mathrm{in};\mathrm{out})[{\cal N}].
\end{align}
Under $\mathcal{O}_\mathrm{NS}$, the asymptotic conversion rate of parallel quantum channels is
\begin{align}
    &r_{\mathcal{O}_\mathrm{NS}}\qty(\mathcal{N}_1\to\mathcal{N}_2)\coloneqq\sup\left\{r\geq 0:\exists\qty{\Theta_n\in \mathcal{O}_\mathrm{NS}}_n,\liminf_{n\to\infty}\frac{1}{2}\left\|J\qty(\Theta_n\qty(\mathcal{N}_1^{\otimes n}))-J\qty(\mathcal{N}_2^{\otimes \lceil rn\rceil})\right\|_1=0\right\}.
\end{align}
We define $r_{\mathcal{O}_\mathrm{CNS}}\qty(\mathcal{N}_1\to\mathcal{N}_2)$ and $r_{\mathcal{O}_\mathrm{ENS}}\qty(\mathcal{N}_1\to\mathcal{N}_2)$ in the same way.
Under $\tilde{\mathcal{O}}$, we define the asymptotic conversion rate of parallel quantum channels as 
\begin{align}
    &r_{\tilde{\mathcal{O}}}\qty(\mathcal{N}_1\to\mathcal{N}_2)\coloneqq\sup\left\{r\geq 0:\exists\qty{\Theta_n}_n\in\tilde{\mathcal{O}},\liminf_{n\to\infty}\frac{1}{2}\left\|J\qty(\Theta_n\qty(\mathcal{N}_1^{\otimes n}))-J\qty(\mathcal{N}_2^{\otimes \lceil rn\rceil})\right\|_1=0\right\}.
\end{align}
Since the sequence of conversion $\qty{\Theta_n\in\mathcal{O}_\mathrm{NS}}_n$ belong to $\tilde{\mathcal{O}}$, due to~\eqref{eq:replacer_inclusion_relation}, we have
\begin{align}
r_{\mathcal{O}_\mathrm{CNS}}\qty(\mathcal{N}_1\to\mathcal{N}_2)
\le r_{\mathcal{O}_\mathrm{ENS}}\qty(\mathcal{N}_1\to\mathcal{N}_2)
\le r_{\mathcal{O}_\mathrm{NS}}\qty(\mathcal{N}_1\to\mathcal{N}_2)
\le r_{\tilde{\mathcal{O}}}\qty(\mathcal{N}_1\to\mathcal{N}_2)
=\frac{ I(\mathrm{in}; \mathrm{out})[{\cal N}_1]}{ I(\mathrm{in}; \mathrm{out})[{\cal N}_2]},\label{BA5}
\end{align}
where $I(\mathrm{in}; \mathrm{out})[\cdots]$ in the last line is defined in~\eqref{eq:relative_entropy_of_resource4}, and the last equality follows from~\eqref{eq:replacer_regularized_relative_entropy_of_resource} along with the second law, i.e., Theorem~\ref{thm:second_law_si}.
As suggested by the rightmost term in~\eqref{BA5}, the second law of QRTs generally provides a converse bound on the asymptotic conversion rate under free operations such as $\mathcal{O}_\mathrm{CNS}$, $\mathcal{O}_\mathrm{ENS}$, and $\mathcal{O}_\mathrm{NS}$. However, in the case of replacers, we will show that equality holds in some of these inequalities, making the rightmost term in~\eqref{BA5} an exact characterization of certain asymptotic conversion rates.

First, we consider the simulation of 
a qubit noiseless classical-classical (CC) channel
${\cal N}_\mathrm{NL}: \ket{k'}\bra{k'} \to |k'\rangle \langle k'|$ ($k'\in\{0,1\}$)
by using given a CQ channel ${\cal N}=\sum_{k=0}^{d-1}\Tr[\ket{k}\bra{k}\rho]\rho_k$.
In random coding, the sender and the receiver share common randomness $S$.
Depending on the value of $S$, the sender maps $K'$ to $K$ by choosing a function $f_S$ from $K'$ to $K$.
A typical construction of random coding satisfies the following condition:
\begin{align}
\sum_{k'=0}^{d_2-1} \frac{1}{d_2} \Pr\qty[f_S(k') = k] = \frac{1}{d_1}.
\end{align}
Additionally, the receiver performs a measurement on the received system that also depends on the shared randomness $S$.
A POVM measurement can be considered as a special case of a CPTP map.
Considering the direct part 
of CQ-channel coding theorem based on the random coding method~\cite[Proof of Lemma 4.6]{hayashi2016quantum}, 
we have
\begin{align}
r_{\mathcal{O}_\mathrm{CNS}}\qty(\mathcal{N}\to\mathcal{N}_\mathrm{NL})
\ge I(\mathrm{in}; \mathrm{out})[{\cal N}].
\label{BVU}
\end{align}

On the other hand,
when we convert from $ \mathcal{N}_\mathrm{NL}$ to  
$\mathcal{N}$,
the problem is equivalent to 
visible quantum state compression with shared randomness
when the random variable $K=k$ corresponds to  
the states $\rho_k$ and is generated by the distribution 
$P_K(k)=\frac{1}{d}$. 
In the encoding process, 
depending on the shared randomness $S$,
the encoder converts
the random variable $K$ into a binary variable $L$ in a smaller memory.
In the decoding process, depending on the shared randomness $S$,
the decoder generates a quantum state $\rho_{S,L}$.
In this case, the average trace distance is given by
\begin{align}
\E_K\qty[\| \rho_K - \E_{S,L|K}\qty[\rho_{S,L}] \|_1],
\end{align}
where $\E$ denotes the expectation value.
When all the states $\rho_x$ are commutative with each other, 
Refs.~\cite[Theorem 10.8]{hayashi2016quantum} and~\cite{PhysRevLett.83.3081,PhysRevA.64.022308}
show that
\begin{align}
\frac{1}{r_{\mathcal{O}_\mathrm{CNS}}\qty(\mathcal{N}_\mathrm{NL}\to \mathcal{N})}
\le I(\mathrm{in}; \mathrm{out})[{\cal N}].\label{BVU2}
\end{align}

Since
$r_{\mathcal{O}_\mathrm{CNS}}\qty(\mathcal{N}\to\mathcal{N}_\mathrm{NL})
\le \frac{1}{r_{\mathcal{O}_\mathrm{CNS}}\qty(\mathcal{N}_\mathrm{NL}\to \mathcal{N})}
$,
when the relation~\eqref{BVU2} holds in addition to~\eqref{BVU},
we have
\begin{align}
r_{\mathcal{O}_\mathrm{CNS}}\qty(\mathcal{N}\to\mathcal{N}_\mathrm{NL})
=\frac{1}{r_{\mathcal{O}_\mathrm{CNS}}\qty(\mathcal{N}_\mathrm{NL}\to \mathcal{N})}
= I(\mathrm{in}; \mathrm{out})[{\cal N}].
\label{BVU3}
\end{align}
Therefore, 
when two channels $\mathcal{N}_1,\mathcal{N}_2 $
satisfies \eqref{BVU3},
we have
\begin{align}
r_{\mathcal{O}_\mathrm{CNS}}\qty(\mathcal{N}_1\to\mathcal{N}_2)
=\frac{ I(\mathrm{in}; \mathrm{out})[{\cal N}_1]}{ I(\mathrm{in}; \mathrm{out})[{\cal N}_2]},
\Label{BVU5}
\end{align}
which implies the equality in all inequalities of~\eqref{BA5}.

When the states $\rho_x$ are not commutative with each other, 
it is not easy to show the relation~\eqref{BVU2} itself.
Instead, when shared entanglement is allowed instead of shared randomness,
this problem is the same as the quantum reverse Shannon theorem 
\cite[Theorem 3]{IEEE-IT-6757002}, \cite[Theorem 3.10]{BCR2011}.
In these references on the quantum reverse Shannon theorem, the noiseless quantum channel is used instead of the noiseless classical channel.
Since shared entanglement is allowed, the noiseless quantum channel is simulated by twice the length of the noiseless classical channel.
Therefore, even in the non-commutative case, instead of the relation~\eqref{BVU2}, 
we have
\begin{align}
\frac{1}{r_{\mathcal{O}_\mathrm{ENS}}\qty(\mathcal{N}_\mathrm{NL}\to \mathcal{N})}
\le I(\mathrm{in}; \mathrm{out})[{\cal N}]
\label{BVU7},
\end{align}
whose detailed derivation is given later in Sec.~\ref{SS21}.

Thus, the combination of~\eqref{BVU} and~\eqref{BVU7} implies 
\begin{align}
r_{\mathcal{O}_\mathrm{ENS}}\qty(\mathcal{N}\to\mathcal{N}_\mathrm{NL})
=\frac{1}{r_{\mathcal{O}_\mathrm{ENS}}\qty(\mathcal{N}_\mathrm{NL}\to \mathcal{N})}
= I(\mathrm{in}; \mathrm{out})[{\cal N}].\label{BVU8}
\end{align}
Therefore, any two channels $\mathcal{N}_1$ and $\mathcal{N}_2 $ satisfy
\begin{align}
r_{\mathcal{O}_\mathrm{ENS}}\qty(\mathcal{N}_1\to\mathcal{N}_2)
=\frac{ I(\mathrm{in}; \mathrm{out})[{\cal N}_1]}{ I(\mathrm{in}; \mathrm{out})[{\cal N}_2]},
\label{BVU9}
\end{align}
which implies equality in all but the first inequality of~\eqref{BA5}.

Finally, we remark on whether~\eqref{BVU2} holds in this case.
When the sender holds half of an entangled state and the encoding operation involves converting classical information $K$ to classical information $L$, the sender must perform a measurement to determine the value of $L$.
If this measurement does not depend on the classical information $K$, it can be performed prior to the protocol.
Let $R$ denote the outcome of this measurement.
In this case, the protocol can be simulated using the random variable $R$, implying that the relation~\eqref{BVU2} still holds even when the sender and receiver share randomness instead of entanglement.
However, the above discussion assumes that the measurement is independent of the classical information $K$. Since this assumption does not hold in general, it remains unclear whether relation~\eqref{BVU2} holds when the states $\rho_x$ are not mutually commuting.

\subsubsection{Derivation of~\eqref{BVU7}}\label{SS21}
Here, we explain how to derive~\eqref{BVU7} from Refs.~\cite[Theorem 3]{IEEE-IT-6757002} and~\cite[Theorem 3.10]{BCR2011}.
For this purpose, we consider the conversion of quantum-quantum (QQ) channels given in these references.
We consider the following conversion from a QQ channel in ${\cal C}\qty(\mathcal{H}_{\mathrm{in}}^{(1)}\to\mathcal{H}_{\mathrm{out}}^{(1)})$ to a QQ channel in ${\cal C}\qty(\mathcal{H}_{\mathrm{in}}^{(2)}\to\mathcal{H}_{\mathrm{out}}^{(2)})$.
We choose two reference systems
$\mathcal{H}_{\mathrm{in}}^R$, $\mathcal{H}_{\mathrm{out}}^R$, 
an entangled state
$\tilde{\rho}$ on $\mathcal{H}_{\mathrm{in}}^R \otimes \mathcal{H}_{\mathrm{out}}^R$,
a CPTP map $\Gamma_\mathrm{in}\in{\cal C}\qty(\mathcal{H}_{\mathrm{in}}^R \otimes \mathcal{H}_{\mathrm{in}}^{(2)}\to\mathcal{H}_{\mathrm{in}}^{(1)})$,
and
a CPTP map $\Gamma_\mathrm{out}\in{\cal C}\qty(\mathcal{H}_{\mathrm{out}}^R \otimes \mathcal{H}_{\mathrm{out}}^{(1)}\to\mathcal{H}_{\mathrm{out}}^{(2)})$.
Then,
for a QQ channel ${\cal N}\in{\cal C}\qty(\mathcal{H}_{\mathrm{in}}^{(1)}\to\mathcal{H}_{\mathrm{out}}^{(1)})$ and
a state $\rho \in {\cal D}\qty(\mathcal{H}_{\mathrm{in}}^{(2)})$,
we define $\Theta_{\Gamma_\mathrm{in},\Gamma_\mathrm{out},\tilde{\rho}}$ as
\begin{align}
\qty(\Theta_{\Gamma_\mathrm{in},\Gamma_\mathrm{out},\tilde{\rho}} ({\cal N}))\qty(\rho)\coloneqq
\qty(\Gamma_\mathrm{out}
\circ \qty( {\cal N} \otimes \id_{\mathcal{H}_{\mathrm{out}}^R})
\circ \qty(\Gamma_\mathrm{in} \otimes \id_{\mathcal{H}_{\mathrm{out}}^R}))
(\rho \otimes \tilde{\rho}).
\end{align}

We consider CQ channels ${\cal N}_1\in{\cal C}_\mathrm{CQ}\qty(\mathcal{H}_{\mathrm{in}}^{(1)}\to\mathcal{H}_{\mathrm{out}}^{(1)})$
and ${\cal N}_2\in{\cal C}_\mathrm{CQ}\qty(\mathcal{H}_{\mathrm{in}}^{(2)}\to\mathcal{H}_{\mathrm{out}}^{(2)})$.
We choose the reference system $\mathcal{H}_{\mathrm{in},R}^{(2)}$ with the same dimension as $\mathcal{H}_{\mathrm{in}}^{(2)}$.
We have
\begin{align}
&\max_{\rho \in {\cal D}\qty(\mathcal{H}_{\mathrm{in}}^{(2)}
\otimes \mathcal{H}_{\mathrm{in},R}^{(2)})}
\left\|\qty(\qty(\Theta_{\Gamma_\mathrm{in},\Gamma_\mathrm{out},\tilde{\rho}} ({\cal N}_{1})
-{\cal N}_{2})\otimes\id_{\mathcal{H}_{\mathrm{in},R}^{(2)}})(\rho)\right\|_1 
\label{KV1} \\
&=\max_{k'}
\left\|\qty(\Theta_{\Gamma_\mathrm{in},\Gamma_\mathrm{out},\tilde{\rho}} ({\cal N}_{1})
-{\cal N}_{2})(\ket{k'}\bra{k'})\right\|_1 \\
&\ge\sum_{k'=0}^{d_2-1}\frac{1}{d_2}
\left\|\Theta_{\Gamma_\mathrm{in},\Gamma_\mathrm{out},\tilde{\rho}}({\cal N}_1)(\ket{k'}\bra{k'})
-{\cal N}_{2}(\ket{k'}\bra{k'})\right\|_1\\
&=
\left\|J\qty(\Theta_{\Gamma_\mathrm{in},\Gamma_\mathrm{out},\tilde{\rho}}({\cal N}_1))
-J\qty({\cal N}_{2})\right\|_1.\label{KV2}
\end{align}
When the error is measured by~\eqref{KV1}, with
$|\Phi\rangle$ denoting the maximally entangled state over the system
$\mathcal{H}_{\mathrm{in}} \otimes 
\mathcal{H}_{\mathrm{in},R}$, where
$\mathcal{H}_{\mathrm{in},R}$ is the reference system with the same dimension as
$\mathcal{H}_{\mathrm{in}}$,
Ref.~\cite[Theorem 3.10]{BCR2011}
shows that, for a CQ channel $\mathcal{N}$, the rate
\begin{align}
D\left( \qty({\cal N}\otimes \id)\qty(|\Phi\rangle\langle \Phi|)
\middle\| {\cal N}
\qty(\sum_{k=0}^{d_1-1}\frac{1}{d_1}|k\rangle\langle k|) 
\otimes 
\sum_{k=0}^{d_1-1}\frac{1}{d_1}|k\rangle\langle k|\right)
&=\sum_{k=0}^{d_1-1}\frac{1}{d_1} D\left({\cal N}\qty(\ket{k}\bra{k})\middle\|
\sum_{k'=0}^{d_1-1}\frac{1}{d_1}{\cal N}\qty(\ket{k'}\bra{k'})\right) \\
&=I(\mathrm{in}; \mathrm{out})[{\cal N}] 
\end{align}
is achieved.
Hence, the above relation~\eqref{KV2} guarantees~\eqref{BVU7}, i.e.,
the achievability of $I(\mathrm{in}; \mathrm{out})[{\cal N}]$ for the CQ-channel setting.
(The paper~\cite{IEEE-IT-6757002} also essentially shows the same fact in the proof of Theorem 3 of Ref.~\cite{IEEE-IT-6757002}.)

Here, there is a possibility that the relation~\eqref{CB9} does not hold.
However, this problem can be resolved as follows.
We define the unitary $U_{k''}\coloneqq
\sum_{k=0}^{2^n-1} |k+k''\rangle \langle k''|$
on $n$ qubits.
Then, we have
\begin{align}
\frac{1}{2^n}\sum_{k''=0}^{2^n-1}
\Theta_{\Gamma_\mathrm{in},\Gamma_\mathrm{out},\tilde{\rho}}\qty(
U_{k''}^\dagger \circ {\cal N}_\mathrm{NL}^{\otimes n}\circ U_{k''} )
=\Theta_{\Gamma_\mathrm{in},\Gamma_\mathrm{out},\tilde{\rho}}\qty( {\cal N}_\mathrm{NL} ) .
\end{align}
Also, the conversion
${\cal N} \mapsto \frac{1}{2^n}\sum_{k''=0}^{2^n-1}
\Theta_{\Gamma_\mathrm{in},\Gamma_\mathrm{out},\tilde{\rho}}\qty(
U_{k''}^\dagger \circ {\cal N}\circ U_{k''} )$
satisfies the condition~\eqref{CB9}.
The above conversion can be implemented 
by sharing the randomness to identify $k''$.

\subsection{Counterexamples}
\label{sec:counterexample}

We discuss counterexamples that arise when some of the assumptions in our analysis of the generalized quantum Stein's lemma are removed.
We consider the one-copy case in Sec.~\ref{se:one_copy}, and the asymptotic case in Sec.~\ref{se:asymptotic}.

\subsubsection{One-copy case}
\label{se:one_copy}

When ${\cal S}$ is not a convex set, 
the relation~\eqref{BH4} in Lemma~\ref{L1} does not hold in general.
To see this, we consider examples using 
a representation $U$ of a group $G$ over the system ${\cal H}$.
We assume that any element of ${\cal S}$ moves to each other
by the application of the representation $U$.
That is, when an element $\sigma_0 \in {\cal S}$ is fixed,
for any element $\sigma \in {\cal S}$,
there exists $g\in G$  such that
\begin{align}
\label{eq:sigma_group}
\sigma=U(g) \sigma_0 U(g)^\dagger.
\end{align}
Also, we assume that $G$ is a compact group so that
$G$ has the invariant probability measure $\nu$.
Then, using an element $\sigma_0 \in {\cal S}$,
we can define the average state $\sigma_\mathrm{av}\coloneqq
\int_G U(g) \sigma_0 U(g)^\dagger \nu(dg)$, which does not
depend on the choice of $\sigma_0 \in {\cal S}$.
We have the following lemma that will be useful for analyzing examples later.

\begin{lemma}\Label{LL7}
When the set ${\cal S}$ satisfies the above conditions,
for any $\epsilon>0$,
we have
\begin{align}
\beta_\epsilon (\rho\|{\cal S})=
\beta_\epsilon (\rho\|\sigma_\mathrm{av}),\Label{ZVI}
\end{align}
and there exists an optimal POVM $\{T,\mathds{1}-T\}$ for this hypothesis testing in the form of
\begin{align}
T=\int_G  U(g)^\dagger T U(g) \nu(dg).   
\end{align}
\end{lemma}
\begin{proof}

For any POVM element $T$, it holds that
\begin{align}
U(g')\qty(\int_G  U(g)^\dagger T U(g) \nu(dg) )U(g')^\dagger
=\int_G  U(g)^\dagger T U(g) \nu(dg).
\end{align}
Thus, given any element $\sigma \in {\cal S}$, we have
\begin{align}
\Tr[
\qty(\int_G  U(g)^\dagger T U(g) \nu(dg) )
\sigma
]
&= \Tr[ \qty(\int_G  U(g)^\dagger T U(g) \nu(dg) ) U(g')^\dagger\sigma U(g')] \notag\\
&= \Tr [
\qty(\int_G  U(g)^\dagger T U(g) \nu(dg) )
\qty(\int_G U(g')^\dagger\sigma U(g') \nu(dg')) ]\notag\\
&= \Tr[ 
\qty(\int_G  U(g)^\dagger T U(g) \nu(dg) )\sigma_\mathrm{av}].
\Label{NM81}
\end{align}

Due to~\eqref{eq:sigma_group}, we have
\begin{align}
& \Tr [
\qty(\int_G  U(g)^\dagger T U(g) \nu(dg) )
\sigma]
=
\max_{\sigma' \in {\cal S}} \Tr [
\qty(\int_G  U(g)^\dagger T U(g) \nu(dg) )
\sigma' ].\Label{NM82}
\end{align}
Therefore, we have
\begin{align}
\max_{\sigma \in {\cal S}}\Tr[T \sigma]
&=\max_{g \in G}\Tr[T U(g) \sigma_0 U(g)^\dagger]\notag\\ 
&\ge 
\int_G 
\Tr[T U(g) \sigma_0 U(g)^\dagger] \nu(d g) \notag\\
&=
\Tr[\qty(\int_G  U(g)^\dagger  T U(g) \nu(d g)) \sigma_0] \notag\\
&=
\max_{\sigma \in {\cal S}} \Tr[ 
\qty(\int_G  U(g)^\dagger T U(g) \nu(dg) )
\sigma].
\label{eq:max_T_average}
\end{align}
Also, for any $\sigma$, we have 
\begin{align}
\min_{T}\Tr [T\sigma ]
\leq\min_{T}\Tr[\qty(\int_G U(g) T U(g)^\dagger
\nu(d g))\sigma ].
\label{eq:max_T_average2}
\end{align}
Due to~\eqref{eq:max_T_average} and~\eqref{eq:max_T_average2}, the optimal POVM element $T$ can be restricted to 
$T=\int_G U(g) T U(g)^\dagger\nu(d g)$ in the hypothesis testing of $\beta_\epsilon(\rho,{\cal S}) $.

Since the combination of \eqref{NM81} and \eqref{NM82} implies
\begin{align}
\max_{\sigma \in {\cal S}} \Tr [
\qty(\int_G  U(g)^\dagger T U(g) \nu(dg) )
\sigma ]= \Tr [
\qty(\int_G  U(g)^\dagger T U(g) \nu(dg) )\sigma_\mathrm{av}],
\end{align}
we obtain \eqref{ZVI}.
\end{proof}

\begin{example}\Label{Ex1}
We consider the two-dimensional system spanned by $\{|0\rangle, |1\rangle\}$.
We choose $G$ to be $\mathbb{Z}_2=\{0,1\}$ and consider the representation
$U$ as
\begin{align}
U(0)\coloneqq\mathds{1},\quad U(1)\coloneqq|0\rangle \langle 0|- |1\rangle \langle 1|.
\end{align}
With $\mu \in [0,1/2]$,
we set $\rho$ and $\sigma_0$ as
\begin{align}
\rho\coloneqq \frac{\mathds{1}}{2},\quad \sigma_0\coloneqq\sigma[\mu],
\end{align}
where $\sigma[\mu]$ is defined as
\begin{align}
\sigma[\mu]\coloneqq \frac{1}{2}
\left(
\begin{array}{cc}
1 & 2\mu -1 \\
2\mu -1 & 1
\end{array}
\right).
\end{align}
Hence, $U(1) \sigma_0 U(1)^\dagger= \sigma[1-\mu]$.

Then, for $\epsilon \in [0,1/2]$, 
we can calculate 
\begin{align}
\max_{\sigma \in {\cal S}} \beta_\epsilon(\rho\|\sigma)
=(1-2\epsilon)(1-\mu)+ \mu
=1-2(1-\mu)\epsilon.
\end{align}
Since $\sigma_\mathrm{av}=\mathds{1}/2$, we have
\begin{align}
 \beta_\epsilon(\rho\|{\cal S})=\beta_\epsilon\left(\frac{\mathds{1}}{2}\middle\|\frac{\mathds{1}}{2}\right)
=1-\epsilon
> \max_{\sigma \in {\cal S}} \beta_\epsilon(\rho\|\sigma).
\end{align}
\end{example}

\begin{example}\Label{Ex2}
For the two dimensional system spanned by $\{|0\rangle, |1\rangle\}$,
we consider another example.
Let $\rho$ be a pure state $|0\rangle \langle 0|$,
and ${\cal S}$ be the set $\{ \sigma_\theta \}_{\theta \in [0,2\pi)}$,
where 
\begin{align}
\sigma_\theta&\coloneqq U(\theta)|\phi_{p,+} \rangle \langle \psi_{p,+}| U(\theta)^\dagger,\quad
|\phi_{p,\pm} \rangle\coloneqq \sqrt{p}|0\rangle \pm\sqrt{1-p}|1\rangle, \\
U(\theta)&\coloneqq|0\rangle \langle 0|+ e^{i\theta}|1\rangle \langle 1|.
\end{align}
This case satisfies the condition of Lemma \ref{LL7}.

When $\epsilon \ge p$,
the POVM $\{T,\mathds{1}-T\}$ for $T=\mathds{1}-\sigma_\theta$ satisfies the condition $\Tr[(\mathds{1}-T)\rho]\le \epsilon$ and is optimal.
Then, we have
\begin{align}
\beta_\epsilon\left(\rho\middle\| \sigma_\theta\right)=0.
\end{align}
That is, 
\begin{align}
\max_{\sigma \in {\cal S}}\beta_\epsilon(\rho\| \sigma)=0.\Label{F1}
\end{align}

When $\epsilon < p$,
$T=\mathds{1}-
U(\theta)|\phi_{\epsilon,+} \rangle\langle \phi_{\epsilon,+}|
U(\theta)^\dagger
=U(\theta)|\phi_{1-\epsilon,-} \rangle\langle \phi_{1-\epsilon,-}|
$ 
satisfies the condition $\Tr[(\mathds{1}-T)\rho]\le \epsilon$ and is optimal.
Then, we have
\begin{align}
\beta_\epsilon(\rho\| \sigma_\theta)=
\qty(\sqrt{(1-\epsilon)p}-\sqrt{\epsilon (1-p)})^2.
\end{align}
That is, 
\begin{align}
\max_{\sigma \in {\cal S}}\beta_\epsilon(\rho\| \sigma)=
\qty(\sqrt{(1-\epsilon)p}-\sqrt{\epsilon (1-p)})^2.\Label{F2}
\end{align}

Next, we consider $\beta_\epsilon (\rho\|{\cal S})$.
Due to Lemma \ref{LL7},
$T$ can be restricted to 
$\frac{1}{2\pi} \int_0^{2\pi} U(\theta) T U(\theta)^\dagger
d\theta$, which has the form 
$T(a_0,a_1):=a_0|0\rangle \langle 0| +a_1|1\rangle \langle 1|$.
Since $T(a_0,a_1)=U(\theta) T(a_0,a_1)U(\theta)^\dagger$
and 
$\frac{1}{2\pi}\int_0^{2\pi} \sigma_\theta d\theta
=p|0\rangle \langle 0| +(1-p)|1\rangle \langle 1|$,
we see that $T=(1-\epsilon)|0\rangle \langle 0|$ satisfies the condition $\Tr[ (\mathds{1}-T)\rho]\le \epsilon$,
and is optimal.
Thus, we have
\begin{align}
\beta_\epsilon \left(\rho\middle\|p|0\rangle \langle 0| +(1-p)|1\rangle \langle 1|\right)
=(1-\epsilon)p.\Label{F3}
\end{align}
Since~\eqref{F3} is strictly larger than~\eqref{F1} and~\eqref{F2},
the relation~\eqref{BH4} does not hold in this case.
\end{example}

\subsubsection{Asymptotic case}
\label{se:asymptotic}

In Ref.~\cite{hiai2009quantum}, the case of 
${\cal S}_n=\qty{\sigma^{\otimes n}}_{\sigma \in {\cal S}}$
over ${\cal H}^{\otimes n}$ is considered.
In this case, we have the representation $\tilde{U}^n$
of $G$ as
\begin{align}
\tilde{U}^n(g)\coloneqq \underbrace{U(g) \otimes \cdots \otimes U(g)}_{n}.
\end{align}
Using $\sigma \in {\cal S}_n$, 
we define 
\begin{align}
\sigma_{0}^n:=\int_G \tilde{U}^n(g) \sigma \tilde{U}^n(g)^\dagger \nu(dg).
\end{align}
Applying Lemma \ref{LL7}, we have
\begin{align}
\beta_\epsilon\left(\rho^{\otimes n}\middle\|{\cal S}\right)
=\beta_\epsilon\left(\rho^{\otimes n}\middle\|\sigma_\mathrm{av}^n\right).\Label{ZNI}
\end{align}
In this case, we have the following.

\begin{proposition}[\protect{\cite[Theorem 4.4]{hiai2009quantum}}]
\Label{P11}
The limit $\lim_{n\to \infty}\frac{1}{n}D(\rho^{\otimes n}
\|\sigma_\mathrm{av}^n)$ exists. We have
\begin{align}
\lim_{n\to \infty}-\frac{1}{n}\log \beta_\epsilon \left(\rho^{\otimes n}\middle\|\sigma_\mathrm{av}^n\right)
=\lim_{n\to \infty}\frac{1}{n}D\left(\rho^{\otimes n}\middle\|\sigma_\mathrm{av}^n\right).
\end{align}
\end{proposition}

Combining~\eqref{ZNI} and Proposition~\ref{P11}, we have
\begin{align}
\lim_{n\to \infty}-\frac{1}{n}\log \beta_\epsilon \left(\rho^{\otimes n}\middle\|{\cal S}_n\right)=\lim_{n\to \infty}\frac{1}{n}D\left(\rho^{\otimes n} \middle\|\sigma_\mathrm{av}^n\right).
\end{align}
In contrast, for an optimal state $\sigma^\prime$ in the minimization of $\min_{\sigma\in {\cal S}} D(\rho\|\sigma)$, it holds that
\begin{align}
\lim_{n\to \infty}-\frac{1}{n}\log 
\max_{\sigma_n \in {\cal S}_n}
\beta_\epsilon \left(\rho^{\otimes n}\middle\|\sigma_n\right)
&\leq
\lim_{n\to \infty}
-\frac{1}{n}\log 
\beta_\epsilon \left(\rho^{\otimes n}\middle\|{\sigma}^{\prime \otimes n}\right)
=D\left(\rho\middle\|\sigma'\right)
=\min_{\sigma\in {\cal S}} D(\rho\|\sigma).
\end{align}
Since~\eqref{BH4Y} implies
\begin{align}
\lim_{n\to \infty}-\frac{1}{n}\log \beta_\epsilon \left(\rho^{\otimes n}\middle\|\mathcal{S}_n\right)
\le \lim_{n\to \infty}-\frac{1}{n}\log 
\max_{\sigma_n \in {\cal S}_n}
\beta_\epsilon \left(\rho^{\otimes n}\middle\|\sigma_n\right),  
\end{align}
we have
\begin{align}
\lim_{n\to \infty}\frac{1}{n}D\left(\rho^{\otimes n} \middle\|\sigma_\mathrm{av}^n\right)
\le \min_{\sigma\in {\cal S}} D(\rho\|\sigma).
\Label{ES2}
\end{align}
The problem of whether~\eqref{XBI} holds
is reduced to the problem of whether the equality in~\eqref{ES2} holds.

\begin{example}\Label{Ex3}
We consider the same example as Example~\ref{Ex1}.
In Ref.~\cite[Example 6.1]{hiai2009quantum} with $\lambda =1/2$,
it is shown that
\begin{align}
&\lim_{n\to \infty}\frac{1}{n}D\left(\rho^{\otimes n} \middle\|\sigma_\mathrm{av}^n\right)= \min D\left( \frac{\mathds{1}}{2}\middle\|\sigma[\mu]\right),\\
&D\left( \frac{\mathds{1}}{2} \middle\|\sigma[1-\mu]\right) 
=\min_{\sigma \in {\cal S}}D( \rho \|\sigma )
 =D\left( \frac{\mathds{1}}{2} \middle\|\sigma[\lambda]\right).
\end{align}
In this example, the equality in~\eqref{ES2} holds.
\end{example}

\begin{example}\Label{Ex4}
We consider the same example as Example~\ref{Ex3}.
The example here coincides with Ref.~\cite[Example 6.5]{hiai2009quantum} with $\lambda=1$ and $\mu=p$.
Then, it is shown that
\begin{align}
    \lim_{n\to \infty}\frac{1}{n}D\left(\rho^{\otimes n} \middle\|\sigma_\mathrm{av}^n\right)=-\log p.
\end{align}
On the other hand, we have
\begin{align}
\min_{\sigma \in {\cal S}}D( \rho \|\sigma ) =D\left( \rho \middle\|\sigma_0 \right)=\infty.
\end{align}
In this example, 
the equality in~\eqref{ES2} does not hold.
\end{example}

\bibliography{citation}

\end{document}